\documentclass[11pt, a4paper]{article}
\usepackage{amsmath}
\usepackage{graphicx}
\usepackage{enumerate}
\usepackage{natbib}
\usepackage{subcaption}
\RequirePackage[hmargin =1in, head=18pt, a4paper]{geometry}
\usepackage{setspace}
\usepackage{wl-art}
\theoremstyle{definition}

\usepackage{enumitem}

\newcommand{\blind}{0}

\definecolor{correctionblue}{rgb}{0.87,0.39,0.29}
\definecolor{originalblue}{RGB}{28,85,145}
\newif\ifshowcorrections
\showcorrectionstrue

\begin{document}

\if1\blind
{\title{Prediction Sets and Conformal Inference with \\ Interval Outcomes}
    \author{}
\maketitle}
\fi
\if0\blind{
    \title{\Large Prediction Sets and Conformal Inference with Interval Outcomes
        \thanks{\fontsize{9}{11} \selectfont
    \textbf{Liu}: \href{mailto:weiguang.liu@ucl.ac.uk}{weiguang.liu@ucl.ac.uk}; \textbf{de Paula}: \href{mailto:a.paula@ucl.ac.uk}{a.paula@ucl.ac.uk}; \textbf{Tamer}: \href{mailto:elietamer@fas.harvard.edu}{elietamer@fas.harvard.edu}. We thank participants at the 7th International Conference on Econometrics and Statistics, the CeMMAP Econometrics Research Day at UCL, the 2024 INFORMS Annual Meeting, the CUHK Workshop on Recent Development in Econometrics, Stochastic Dominance and Quantile-based Methods in Financial Econometrics Workshop, the 2025 IAAE annual meeting, the 2025 World Congress of the Econometric Society, as well as the seminars at University College London, the University of Cambridge, the University of Manchester and the London School of Economics for their valuable feedback. We thank the generous funding from the UK Research and Innovation (UKRI) under the UK government's Horizon Europe funding guarantee (Grant Ref: EP/X02931X/1) and the Economic and Social Research Council (ESRC) funding through the ESRC Institute for the Microeconomic Analysis of Public Policy (Grant Ref: ES/T014334/1). The Python code for implementing the proposed conformal prediction procedure and replicating the numerical experiments and empirical study is available at \url{https://github.com/lwg342/Prediction-Sets-and-Conformal-Inference-with-Interval-Outcomes}. {We thank Xiaowei Xu for providing the Adzuna job advert data used in \cite{dias2021WorkerMobility}. }} }
    \author{{Weiguang Liu} \\ \small{UCL} \and \'{A}ureo de Paula \\ \small{UCL, IFS and CeMMAP}\and Elie Tamer \\ \small{Harvard University } }
    \date{\today}
    \maketitle
}\fi

\begin{abstract}
    \fontsize{11}{12}\selectfont
    Given data on a random variable $Y$, a prediction set with miscoverage level $\alpha \in (0,1)$ is a set that contains a new draw of $Y$ with probability at least $1-\alpha$. Among all prediction sets satisfying this coverage property, the oracle prediction set is the one with minimal volume. The oracle prediction set offers a complementary view of the distribution of $Y$, beyond point estimators such as the mean and quantiles, and has attracted considerable interest recently. This paper develops methods for estimating such prediction sets conditional on observed covariates when $Y$ is \textit{censored} or \textit{interval-valued}. We characterise the oracle prediction set under partial identification induced by interval censoring and propose consistent estimators for both oracle prediction intervals and more general oracle prediction sets consisting of multiple disjoint intervals. In addition, we apply conformal inference to construct finite-sample valid prediction sets for interval outcomes that remain consistent as the sample size grows, using a conformity score tailored to interval data. The proposed procedure accounts for irreducible prediction uncertainty due to the stochastic nature of outcomes, modelling uncertainty arising from partial identification, and sampling uncertainty that vanishes as sample size increases. We conduct Monte Carlo simulations and two empirical applications using UK job advert data and the US Current Population Survey. The results demonstrate the robustness and efficiency of the proposed methods.

\end{abstract}

\noindent%

\vfill

\onehalfspacing

\section{Introduction}

Interval data are pervasive. Surveys often resort to ``bracketing'' strategies to address item non-response. For example, respondents are often asked whether their money in savings accounts, in investments, or their income falls into a sequence of brackets such as $[\$10K, \$20K]$, $[\$20K, \$50K]$, etc, while also being given the option to provide a point answer to such questions (see \cite{heeringa1995unfolding} and \cite{moore2001using}). Such data typically involve outcomes of interest such as income, wealth, unemployment duration, and so on. This paper provides a new method for predicting an outcome variable of interest $Y$ when we observe a sample of intervals $[Y^L, Y^U]$ such that $Y \in [Y^L, Y^U]$ rather than direct observations of $Y$.

More specifically, suppose we have a random variable $Y$ defined on $\mathcal Y \subseteq \mathbb{R}$ with distribution $P_Y$. One parameter of interest is a {\it prediction set} $C\subseteq \mathcal{Y}$ with miscoverage level $\alpha  \in (0,1)$, which is any measurable subset of $\mathcal{Y}$ such that $P_Y( Y \in C) \geq 1-\alpha$. The oracle prediction set is the set with the minimal volume that has this coverage property, which is related to the level set of the density of $Y$ when this density exists. With covariates, we have access to a sample $(X_{i}, Y_{i}) \in \mathcal{X} \times \mathcal{Y}$, $i = 1,\dots,n$, where $\mathcal Y\subseteq\mathbb R$ is the support of $Y$ and $\mathcal X\subseteq\mathbb R^d$ is a covariate support with \(P_X(\mathcal{X}) = 1\) on which the conditional distributions are defined. We use these data to construct prediction sets $C_n(x)\equiv C(x; X_{1}, Y_{1}, \dots, X_{n},Y_{n}) \subseteq \mathcal{Y}$ for a given miscoverage level $\alpha$ at $x \in \mathcal{X}$. It is desirable that such a prediction set $C_n$ is both consistent for the oracle prediction set and has a {\it finite-sample coverage property:}
\begin{equation}\label{sample prediction set}
    \Prob(Y_{n+1} \in C_n(X_{n+1})) \geq 1 - \alpha,
\end{equation}
where $\Prob$ is the distribution over $(X_{i}, Y_i)$ for $i= 1,\dots, n+1$. This paper constructs such a set when $Y$ is interval-valued. The exact finite-sample coverage property that our construction provides is given in more detail in \autoref{sec:conformal}.

Such an ``oracle prediction'' exercise is familiar in the literature, especially in the context of forecasting in linear regressions (see \cite{diebold2015ForecastingEconomics}). For example, assuming that the outcome $Y$ satisfies $Y_{i}\mid X_{i} \sim \mathcal N(X_{i}'\beta, \sigma^2)$ with known $\sigma^{2}$, the population oracle prediction set is $X_{n+1}'\beta \pm 1.96\sigma$. Given data $\mathcal{D} = \cqty{(X_1, Y_1), \ldots, (X_n, Y_n)}$, a consistent estimator for such a set is $X_{n+1}'\hat\beta \, \pm \,  1.96\sigma$, where $\hat \beta$ is the least squares estimator of $\beta$. On the other hand, ``an operational density forecast that accounts for parameter uncertainty'' is (\cite{diebold2015ForecastingEconomics})
\[
    N\left(X_{n+1}' \hat{\beta}, \ {\sigma}^2 \left(1 + X_{n+1}' (\mathbf{X}'\mathbf{X})^{-1} X_{n+1} \right)\right).
\]
This ``operational density'' summarises both the sampling uncertainty in estimating $\hat \beta$ and the disturbance uncertainty arising from the intrinsic randomness in the conditional distribution $Y_{n+1}\mid X_{n+1}$. One can use it to derive a forecast density with a particular coverage property.

The prediction set is closely related to quantile methods if we have direct observations of the outcome variable $Y$. If the oracle prediction set is a single interval, it is given by $C(X_{n+1}) \equiv [\tau_{0}(X_{n+1}), \tau_{1}(X_{n+1})]$, where $\tau_0(X_{n+1})$ and $\tau_1(X_{n+1})$ are appropriately chosen quantiles for $Y$ given $X$ so that one can obtain a coverage rate $(1-\alpha)$ with minimal volume. Taking the Gaussian specification above, with $\alpha=5\%$ and assuming $\beta$ is known, $\tau_0(X_{n+1}) = X_{n+1}'\beta - 1.96\sigma$ and $\tau_1(X_{n+1}) = X_{n+1}'\beta + 1.96\sigma$, where $1.96\sigma$ is the $95\%$-quantile of $|\varepsilon_{n+1}| \equiv |Y_{n+1}-X_{n+1}'\beta|$. When $\sigma$ is known as presumed above, but $\beta$ is not, estimates for these can be obtained as $\hat \tau_0(X_{n+1}) = X_{n+1}'\hat \beta - 1.96\sigma$ and $\hat \tau_1(X_{n+1}) = X_{n+1}'\hat \beta + 1.96\sigma$, where $\hat \beta$ is the OLS or ML estimate of $\beta$. This nonetheless ignores the parameter uncertainty in estimating $\beta$. Although this suffices for interval consistency, as noted in \cite{romano2019ConformalizedQuantile}, replacing $\beta$ with $\hat \beta$ ``is not guaranteed to satisfy the coverage statement (\dots) when $C(X_{n+1})$ is replaced by the estimated interval $\hat C(X_{n+1})$.''  To address this shortcoming, \cite{romano2019ConformalizedQuantile} employ conformal inference and modify the interval $\hat C(X_{n+1}) = [\hat \tau_0(X_{n+1}),\hat \tau_1(X_{n+1})]$ by adjusting the endpoints to $\hat \tau_0(X_{n+1})-\vartheta_{1-\alpha}$ and $\hat \tau_1(X_{n+1})+\vartheta_{1-\alpha}$.  Here, $\vartheta_{1-\alpha}$ is the $(1-\alpha)(1+1/n_2)$-th empirical quantile of the set of conformity scores $\cqty{\max\{\hat \tau_0(X_i)-Y_i,Y_i-\hat \tau_1(X_i)\}: i\in \mathcal{I}_2}$ estimated in a calibration sub-sample $\mathcal{I}_2$ with sample size $n_2$.  This corrects for under- or over-coverage and allows for the construction of prediction intervals satisfying desired coverage levels in finite samples using quantile estimates.

However, several questions and limitations surface when the outcome variable $Y$ is interval-censored. First, the definition of the oracle prediction set needs to be extended in partially identified models. We also note that although we can construct a valid prediction interval using quantile estimates, such intervals are not necessarily minimal in volume without assumptions on the shape of the distribution (such as the Gaussian assumption in the example above), and this approach is not applicable when the oracle prediction set consists of multiple disjoint intervals. Furthermore, a new conformity score is needed to accommodate interval-valued observations. Our method addresses these limitations and offers a complementary perspective on the distribution of outcomes beyond existing quantile-based methods.

We first define the oracle prediction sets under partial identification. We then characterise the oracle prediction sets under interval censoring via a feasible optimisation problem that relates the oracle prediction sets to the conditional distribution of the observed upper and lower brackets, $(Y^L, Y^U)$. We use a kernel-based conditional distribution function estimator as a basis for the construction of a consistent set estimator for the oracle prediction set. The oracle prediction set may be the union of disjoint intervals, for instance, when the prediction density is multimodal. We provide an approach that allows for the prediction set estimator to be a union of disjoint intervals with the prescribed coverage. In cases with a multimodal density, a union of intervals can have a smaller volume than one interval. We also allow for both random censoring with continuously distributed brackets and fixed censoring with discretely distributed brackets, such as those arising from ``unfolding brackets'' or specific survey schemes, see \cite{heeringa1995unfolding} and \cite{moore2001using}. See also the bracketing approach to outcomes using the Health and Retirement Survey in \cite{manski2002InferenceRegressions}.

We use recent developments on conformal inference to propose a conformal prediction set $\tilde{C}_{n}$ using a sample of exchangeable data, which has finite-sample validity while also enjoying the asymptotic property of being consistent for the oracle prediction set under partial identification. Our conformal procedure is designed for interval-valued outcomes, which requires the design of a particular conformity score function to handle interval data. In particular, the score function is non-positive if and only if the constructed prediction set contains the interval outcome. This conformity score function therefore accounts for both over- and under-coverage, as in \cite{romano2019ConformalizedQuantile}. As a competing alternative prediction set estimator, we show in the appendix that a conformal quantile regression method provides a valid prediction interval for interval-valued outcomes, which can be easier to implement in practice, but in general will not be efficient in terms of the volume of the prediction set. In addition, quantile regression-based methods are not suitable when the prediction set is not a single interval.

We illustrate our construction using a set of Monte Carlo experiments, where we consider both fixed and random censoring. Our proposed prediction set is shown to have robust finite sample coverage properties and to have a smaller volume compared with the alternative prediction sets constructed with quantile regression methods.

We apply our methods to two real-world datasets. The first is the UK job advert dataset from Adzuna, an online UK job advert aggregator. Adzuna data are widely used as an indicator of UK economic activity. One feature of these data is that many job ads report wages as ranges. We use these data to construct predictive wage intervals as a function of UK location and job categories. The second application uses the Current Population Survey (CPS) data. Among people who report their income, a significant fraction of the observations are interval-valued. Interval-valued observations are sometimes discarded or handled by imputation methods that rely on assumptions such as missing at random and may cause bias if misspecification happens. Our proposed prediction set is robust in that we do not require a specification of the censoring mechanism.

\subsection*{Literature review}

The paper combines insights from the partial identification literature and conformal inference, which has recently been gaining significant interest. The partial identification literature (see \cite{manski2003PartialIdentification}) takes the view that one should try to make inferences under minimal assumptions, even if these assumptions do not allow for point identification of the parameter of interest. The motivation for such a view is that inferences are less credible and conclusions are suspect when they rely on non-testable or non-credible assumptions. In our context for example, if we are interested in $\beta = E[Y]$ and we observe $[Y^L, Y^U]$ rather than $Y$, then a partial identification approach would just say that, given the data, the identified set for $\beta$  (or all we can learn about $\beta$ given the data) is the interval $\left[E[Y^L], E[Y^U]\right].$ Imputation based inference uses models to predict a $Y^P$ within $[Y^L, Y^U]$ and takes $E[Y^P]$ as a proxy for $\beta = E[Y]$. The credibility of this exercise depends on the imputation model, i.e. the set of assumptions it relies on. In the context of interval data, \cite{manski2002InferenceRegressions} study identification of linear regression coefficients when the data (on regressors or outcomes) is interval valued using a partial identification approach. See \cite{tamer2010PartialIdentification}, \cite{molinari2020Chapter5}, and \cite{kline2023RecentDevelopments} for reviews on the developments in the partial identification literature.

In order to obtain a prediction interval that has finite-sample coverage guarantees, we use conformal inference, which is a distribution-free method that builds upon the idea of exchangeability and permutation tests (\cite{vovk2005AlgorithmicLearning}, \cite{vovk2009OnlinePredictive}). The idea of conformal inference has become increasingly popular in the machine learning literature since it requires no assumptions on the underlying distribution other than the observations being exchangeable, and it also allows the use of any prediction protocol that treats the data symmetrically and hence preserves exchangeability. \cite{barber2023ConformalPrediction} discusses the validity of conformal inference when the data is not i.i.d. and the predictor is not symmetric. Conformal inference has been applied to obtain valid distribution-free prediction sets \citep{lei2013DistributionfreePrediction}, in non-parametric regression \citep{lei2014DistributionfreePrediction} and in high-dimensional settings \citep{lei2018DistributionfreePredictive}. \cite{romano2019ConformalizedQuantile} considers conformal quantile regression to construct prediction intervals instead of conditional mean regression. \cite{sadinle2019LeastAmbiguous} proposes a way to construct set-valued classifiers. \cite{chernozhukov2021DistributionalConformal} considers the construction of valid prediction intervals with the distributional regression method. We add to this literature by allowing the outcome observations to be interval-censored or interval-valued.  Relatedly, \cite{candes2023ConformalizedSurvival} considers conformal inference for survival analysis. Using our notation, this amounts to observing either $Y_i$ or the interval $(Y_i^L, \infty)$ for each $i$. They propose a weighted conformal inference procedure \citep{tibshirani2019ConformalPrediction} under the assumption that censoring is conditionally independent of the outcome given predictors $X_i$. We take a different approach and ask how much information we can extract from a sample with interval-valued outcomes. For example, our procedure allows for the case when all outcomes are interval censored (i.e., $Y_i$ is never observed) and the procedure uses the information from the upper bracket $Y^U_i$. Also, our procedure does not require the conditional independence assumption, which can be violated in many empirical contexts.
Nonresponse raises a separate issue from interval censoring. When response depends on $X$ but is conditionally independent of the outcome brackets, it can induce covariate shift between respondents and the target population. \cite{tibshirani2019ConformalPrediction} develop weighted conformal prediction under covariate shift, which could be applied to our conformity scores to obtain target-population marginal coverage using known inverse response probabilities. \cite{yang2024DoublyRobustCalibration} propose doubly robust calibration combining a propensity model with a conditional score-distribution model, yielding asymptotic validity when either is consistently estimated under additional regularity conditions. These methods suggest extensions of our procedure to nonresponse settings under similar conditions as those entertained in the references above.

There are several related areas connected to our theoretical development. The classical way to construct estimators for the oracle prediction set relies on the level set estimation (see, for example, \cite{wilks1941DeterminationSample}, \cite{polonik1995MeasuringMass} and \cite{samworth2010AsymptoticsOptimal}). The uniform consistency of the nonparametric estimator of the conditional distribution is an extension of classical results on the uniform convergence rate of the kernel-based density and regression estimators, such as \cite{einmahl2000EmpiricalProcess}, \cite{gine2002RatesStrong} and \cite{hansen2008UniformConvergence}. This paper is also related to the literature on quantile regression (see \cite{koenker2017QuantileRegression} for a review).

\subsection*{Outline}
This paper is organised as follows. In \autoref{sec:setup}, we provide the setup and definitions of optimal sets and approaches to obtaining estimators of such sets with censoring. \autoref{sec:consistency} provides a study of the consistency of such sets using kernel estimators of the conditional distribution. \autoref{sec:conformal} provides results on finite sample coverage using a conformal inference procedure under interval censoring. \autoref{sec:simulation} provides some Monte Carlo evidence, and \autoref{sec:empirical} applies our methods to UK job advert data as well as the US CPS (Current Population Survey) Data. \autoref{sec:conclusion} concludes. The proofs for the theorems in the main text, as well as some further results and comments, are collected in the Appendix.

\section{Setup and definitions }\label{sec:setup}
Suppose there is a random sample \(\bar{\mathcal{D}} = \{\bar{Z}_{i} = (X_{i}, Y_{i}, Y_{i}^{L}, Y_{i}^{U}) \in \mathbb{R}^{d+3}: i \in \mathcal{I}\}\) with sample size \(\abs{\mathcal{I}} = n\) drawn independently from a joint distribution \(P\). We only observe the sub-vector \(Z_{i} = (X_{i}, Y_{i}^{L}, Y_{i}^{U})\), and \(Y_{i}\) is not directly observed but assumed to satisfy the condition \(Y_{i}^{L} \leq Y_{i} \leq Y_{i}^{U}\). {For each $i\in \mathcal{I}$, the predictor \(X_i\) is a \(d\)-dimensional random vector taking values in \(\mathcal{X}\).} The latent outcome \(Y_{i}\) and the observed lower and upper bounds \(Y_{i}^{L}\) and \(Y_{i}^{U}\) are one-dimensional. Let \(\bar{Z} = (X, Y, Y^L, Y^{U})\) denote a generic sample from the joint distribution \(P\). We will use notation such as $P_Y$ to denote the marginal distribution of $Y$, and $P_{Y\mid X}$ to denote the distribution of $Y$ conditional on $X$. Note here that the joint distribution \(P_{X, Y^L, Y^U}\) of $(X,Y^{L}, Y^{U})$ is identified given the observed sample \(\mathcal{D} = \cqty{Z_ i = (X_i, Y_i^L, Y_i^U) \in \mathbb{R}^{d+2}: i \in \mathcal{I}}\), but the joint distribution \(P_{X, Y}\) of predictors $X$ and latent outcome $Y$ is not point identified given the observed sample $\mathcal D$ without restrictions on the censoring mechanism.

A prediction set is defined as a set $C\subseteq \mathcal{X} \times \mathcal{Y}$, with the interpretation that, given a realisation of the predictor $X = x$, we predict that $Y$ is likely to fall inside the section $C(x) = \cqty{y : (x, y) \in C}$. In the absence of interval censoring, estimation of a prediction set that satisfies certain optimality criteria (i.e. the minimal prediction set with coverage guarantees) has been studied in the literature. See \cite{lei2014DistributionfreePrediction}, \cite*{barber2021LimitsDistributionfree} among others. As we discuss later in this section, the oracle prediction set without interval censoring is related to the upper level set of the density function of $Y$ given $X$ under point identification.

The situation is significantly different when we are faced with interval censoring and $P_{X,Y}$ is partially identified. In the rest of this section, we first define an optimality criterion for a prediction set under censoring, and analyse why the classical oracle prediction set (which corresponds to an upper level set) is no longer applicable under censoring. We then propose a feasible estimation procedure for the oracle prediction set. We discuss the potential reasons why the procedure can result in conservative prediction sets and show that the proposed procedure does not incur any unnecessary conservativeness.

The two criteria for an oracle prediction set are validity and efficiency \citep[see][]{vovk2009OnlinePredictive,lei2013DistributionfreePrediction}. Let $\mathcal{P}_{I}$ be the identified set of the joint distribution of $(X,Y)$ under censoring. That is, $\mathcal{P}_{I}$ is the set of joint distributions of $(X,Y)$ that are compatible with the joint distribution $P_{X,Y^{L},Y^U}$ and some censoring mechanism. A prediction set $C$ is {valid} if it satisfies the following coverage properties for a new sample $(X_{n+1}, Y_{n+1})$ drawn from the same distribution $P$ at a given miscoverage level $\alpha$.
\begin{definition}[Validity under Partial Identification]\label{def:validity}
    We say a prediction set \(C\) is \textit{marginally valid under partial identification} with miscoverage level $\alpha$, if,
    \begin{equation}\label{eq:marginal validity under partial identification}
        \inf_{P\in \mathcal{P}_I} P \pqty{Y_{n+1} \in C(X_{n+1})} \geq 1 - \alpha,
    \end{equation}
    and \(C\) is \textit{conditionally valid under partial identification} with miscoverage level $\alpha$ if, for $P_X$-almost every \(x\in \mathcal{X}\),
    \begin{equation}\label{eq:conditional validity under partial identification}
        \inf_{P\in \mathcal{P}_I} P \pqty{Y_{n+1} \in C(x) \mid X_{n+1} = x} \geq 1 - \alpha.
    \end{equation}
\end{definition}
When the model is point identified, \(\mathcal{P}_{I}=\cqty{P_{X,Y}}\) is a singleton set, the conditions in Equations \eqref{eq:marginal validity under partial identification} and \eqref{eq:conditional validity under partial identification} reduce to the standard conditions of validity, see for example, \cite{lei2014DistributionfreePrediction}. We will simply refer to these two conditions as {marginal validity} and {conditional validity} henceforth.

It is natural to require a prediction set to be valid, so that we have confidence it will contain the true value of the outcome with a certain probability. However, the requirement of validity by itself is vacuous. For example, the prediction set $C$ such that $C(x) =(-\infty, \infty)$, for all $x \in \mathcal{X}$, is technically valid at any miscoverage level, but it is also uninformative. Therefore, our target parameter of interest is an \textit{oracle prediction set}, denoted by $C^{*}_{\mathcal{P}_I}$, which is defined as the minimal-volume prediction set that satisfies the validity conditions in Definition~\ref{def:validity}. For simplicity, we will omit the subscript and write $C^* =C^{*}_{\mathcal{P}_I}$ whenever the dependence on the identified set ${\mathcal{P}_I}$ is clear. Since the minimal-volume prediction set is intuitively the most informative among all valid prediction sets, this property is also called \textit{efficiency}.\footnote{This notion of efficiency \citep[see][]{vovk2009OnlinePredictive,lei2013DistributionfreePrediction,lei2014DistributionfreePrediction} differs from the usual definition, where an estimator is considered efficient if it has the smallest variance within a class of estimators. However, the concepts are related, as an estimator with a smaller variance generally leads to the construction of tighter and more informative confidence intervals.}

For interpretability, we will focus on the prediction sets such that $C(x)$ takes the form of a union of a (potentially unknown) number of disjoint intervals for each $x\in \mathcal{X}$.  Specifically, we assume that $C(x) \in \mathcal{C}$, $P_X$-almost surely.

For \(K\in\mathbb N\), define
\begin{equation*}
    \mathcal C^{(K)}
    :=\left\{\bigsqcup_{m=1}^{K}[a_m,b_m]:
    a_1<b_1<\cdots<a_K<b_K\right\}.
\end{equation*}
Then, for \(M\in\mathbb N\), let
\begin{equation*}
    \mathcal C_M:=\bigcup_{K=1}^{M}\mathcal C^{(K)},
    \qquad
    \mathcal C:=\bigcup_{K=1}^{\infty}\mathcal C^{(K)}.
\end{equation*}
Here \(\sqcup\) denotes a disjoint union. \(\mathcal C^{(K)}\) is the set of disjoint unions of exactly \(K\) closed intervals. \(\mathcal C_M\) denotes the set of unions of at most \(M\) disjoint intervals, and \(\mathcal C\) denotes the set of finite disjoint unions of closed intervals.
We make the observation that, for \(C_1\in\mathcal C^{(K_1)}\) and \(C_2\in\mathcal C^{(K_2)}\), we have \(C_1\cup C_2\in\mathcal C_{K_1+K_2}\), since overlapping or touching intervals may merge into a single component.

The following oracle prediction set under partial identification will be our estimation target.
\begin{definition}[Oracle Prediction Set under Partial Identification]\label{def: optimal}
    The \textit{oracle prediction set under partial identification} for a given miscoverage level $\alpha$ is defined as \(C^{*} = C^{*}_{\mathcal{P}_{I}}\subseteq \mathcal{X}\times \mathbb{R}\) which solves the following problem for $P_{X}$-almost every \(x\in \mathcal{X}\),
    \begin{equation}\label{eq:level set union}
        C^*(x) \in \arg \min_{C(x)\in \mathcal{C}} \mu(C(x)) \mathtext{s.t.} \inf_{P\in \mathcal{P}_I} P\pqty{Y\in C(x)\mid X = x} \geq 1 - \alpha,
    \end{equation}
    where  \(\mu\) denotes the Lebesgue measure. Since conditional validity implies marginal validity, the oracle prediction set will also be marginally valid under partial identification.
\end{definition}

Before proposing a feasible estimation procedure for the oracle prediction set $C^{*}$, we explain why the existing estimation methods are not directly applicable for prediction sets in the case of interval outcome observations. This analysis highlights the challenges posed by partial identification and the difference between our approach and the classical level set method studied in the literature.

When the joint distribution of $(X,Y)$ is point identified, and $\mathcal{P}_{I} =\cqty{P_{X,Y}}$, it is well known in the literature that the oracle prediction set is related to the upper level set of the conditional density of $Y$ assuming this density is well defined. Let $p(y| x)$ denote the conditional density of $Y$ at $y$ given $X = x$. An upper level set for $p(y\mid x)$ with threshold $\lambda$ at $x \in \mathcal{X}$ is defined as,
\begin{equation*}
    L(x, \lambda) = \cqty{ y: p(y\mid x) \geq \lambda}.
\end{equation*}
Let $\lambda_{x}^{\alpha}$ be chosen such that  $$\int \indicator{p(y\mid x) \geq \lambda_{x}^{\alpha}} p(y\mid x) \dd y = 1 -\alpha,$$
assuming such a threshold exists and $L(x,\lambda_x^\alpha)\in\mathcal C$.
The oracle prediction set under point identification is $C^{*}_{P_{Y\mid X}}=\cqty{(x,y):x\in\mathcal X,\ y\in L(x,\lambda_x^\alpha)}$, which is both conditionally and marginally valid (see \cite{lei2014DistributionfreePrediction}). For a given $x$, it collects the values of $y$ with sufficiently high density and intuitively assembles the set of $y$'s with the highest probability for a given volume.  Once the threshold is chosen to yield the desired miscoverage level, one obtains a minimal-volume set. Estimating the oracle prediction set $C^{*}_{P_{Y|X}}$ therefore reduces to the problem of estimating the upper level sets $L(x, \lambda)$ of the conditional density $p(y \mid x)$ and selecting an appropriate threshold parameter.

However, this characterisation of the oracle prediction set in terms of level sets is no longer applicable under interval censoring, as $P_{Y\mid X}$ is not identified. In fact, the next proposition states that no meaningful bounds can be obtained for the level set of the conditional density of $Y$ given $X$ under interval censoring. For brevity, we temporarily drop the conditioning predictors $X$ in the following proposition. (It can be interpreted as conditional on $X =x$ for $x \in \mathcal{X}$.)
\begin{proposition}\label{prop: level set not identified}
    Suppose $(Y^L,Y^U)$ has a joint density. For $P_{Y^L,Y^U}$-almost every $(a,b)$ with $a<b$, every $\lambda<\infty$, and every $y\in[a,b]$, there exists some $P'\in\mathcal P_I$ such that $y\in L'(\lambda)$, where $L'(\lambda)$ is the upper level set for the density of $P'$.
\end{proposition}
Proposition~\ref{prop: level set not identified} states that for $P_{Y^L,Y^U}$-almost every $(a,b)$ with $a<b$ and every $y\in[a,b]$, we can find a random variable $Y'$ that is observationally equivalent to $Y$ under censoring, and this $Y'$ admits a level set that contains $y$ for any threshold $\lambda$. This is possible because one can adjust the density of $Y'$ at any point $y$ satisfying the stated conditions to be arbitrarily large while ensuring that this density remains within the identified set $P_I$.

The following two lemmas characterise the identified set of the conditional distribution of $Y$ given $X$ under interval censoring. Lemma \ref{lemma: identified set} relates the sharp identified set of the conditional distribution of $Y$ given $X$ to the {conditional distribution} of $(Y^L, Y^U)$, which can be directly identified and estimated from the data, and provides a feasible lower bound for the restriction in Equation (\ref{eq:level set union}). Lemma \ref{lemma: inf P(y) = P(yl,yu)} shows that the lower bound is tight, and we are not unnecessarily conservative in constructing the oracle prediction set.

\begin{lemma}[Theorem SIR-2.3 in \cite{molinari2020Chapter5}]\label{lemma: identified set}
    The sharp identified set for the conditional distribution of \(Y\) given \(X\), under interval censoring is given by
    \begin{equation}\label{eq:identified set}
        \cqty{P_{Y\mid X}: P_{Y\mid X}(Y \in [t_{0},t_{1}]\mid X = x) \geq P\pqty{[Y^{L},Y^U] \subseteq [t_{0},t_{1}] \mid X = x}, \forall t_{0} \leq t_{1}}.
    \end{equation}
\end{lemma}
Lemma \ref{lemma: identified set} suggests that we could implement a feasible version of Equation \eqref{eq:level set union}, by replacing the optimisation constraint $\inf_{P\in \mathcal{P}_I} P\pqty{Y\in C(x)\mid X = x} \geq 1 - \alpha$ with $$P\pqty{[Y^{L},Y^U] \subseteq C(x)\mid X = x} \geq 1 - \alpha.$$ This implies the original constraint, since for $C \in \mathcal C$, we can write $C = \bigsqcup_m [a_m,b_m]$ as a finite disjoint union, and
\begin{align*}
    \inf_{P\in \mathcal{P}_I} P\pqty{Y\in C\mid X = x} &\geq \sum_m \inf_{P\in \mathcal{P}_I}P\pqty{Y\in [a_m, b_m] \mid X = x}  \\ &\geq \sum_m P\pqty{[Y^{L},Y^U] \subseteq [a_m, b_m] \mid X = x} \\
    &= P\pqty{[Y^{L},Y^U] \subseteq C\mid X = x}
\end{align*}
for a disjoint union of intervals $C\in \mathcal C$.  As a result, a prediction set that satisfies the feasible version of the constraint will have a coverage guarantee since
\begin{align*}
    P_{Y\mid X}(Y\in C\mid X =x) &\geq \inf_{P\in \mathcal{P}_I} P\pqty{Y\in C\mid X = x} \\
    &\geq P\pqty{[Y^{L},Y^U] \subseteq C\mid X = x} \geq 1 - \alpha.
\end{align*}
A prediction set $C$ is conservative if one of the three inequalities is strict. The first inequality reflects partial identification of the joint distribution of $(X,Y)$. The last inequality depends on the choice of $C$ and the distribution of $(X,Y^L,Y^U)$. Lemma \ref{lemma: inf P(y) = P(yl,yu)} shows that for $C \in \mathcal{C}$, the second inequality is an equality and hence implementing the feasible version of the optimisation problem induces no additional conservativeness.
\begin{lemma}\label{lemma: inf P(y) = P(yl,yu)}
    For $P_X$-almost every $x\in \mathcal{X}$, and any $C \in \mathcal{C}$, we have under interval censoring,
    \begin{equation*}
        \inf_{P\in \mathcal{P}_I} P\pqty{Y\in C\mid X = x}= P\pqty{[Y^{L},Y^U] \subseteq C\mid X = x}.
    \end{equation*}
\end{lemma}
The above lemma shows that even though the oracle prediction set can be conservative, i.e. overcovering $Y$ for the distribution actually generating the data, the conservativeness is only due to partial identification and the conditional distribution of $(Y^L, Y^U)$ and not because the constraint we focus on differs from the original one.

In fact, given Lemmas \ref{lemma: identified set} and \ref{lemma: inf P(y) = P(yl,yu)}, we could replace the condition in Equation (\ref{eq:level set union}) which depends on the (unobserved) conditional distribution of the latent outcome $Y$ with a condition that depends on the conditional distribution of $(Y^{L}, Y^U)$, which can be estimated given the observed sample $\mathcal{D}$. We then arrive at the following feasible estimation problem:
\begin{proposition}[Feasible Estimation]\label{def: feasible optimal}
    The \textit{oracle prediction set} $C^*$ is the solution to the following feasible optimisation problem for $x$ in the support of $X$,
    \begin{equation}\label{eq:feasible criterion}
        C^*(x) {\in} \arg \min_{C(x)\in \mathcal{C}} \mu(C(x)) \mathtext{s.t.} P\pqty{[Y^{L},Y^U] \subseteq C(x)\mid X = x} \geq 1 - \alpha.
    \end{equation}
    In the special case when  $C^*$ is a single interval, we will denote it as $C^*_{I}$.
\end{proposition}

A natural estimator $\hat{C}$ for the oracle prediction set can be obtained by solving the sample version of the optimisation problem above. The next two sections study the consistency of such estimators for the oracle prediction set and provide a modification of the estimated prediction set that has finite sample coverage guarantees using conformal inference.
\section{Consistency}\label{sec:consistency}

We first consider the estimation of the oracle prediction interval $C^*_{I}(x)= [\tau_{0}(x), \tau_{1}(x)]$, for $x\in \mathcal{X}$. This analysis highlights the conditions required to ensure the consistency of the prediction set under different schemes of interval censoring. Subsequently, we will extend our analysis to the estimation of oracle prediction sets whose sections $C^*(x)\in\mathcal C$ may consist of multiple intervals. To measure the difference between two sets $A, B \subseteq \mathbb{R}$, we use the volume of their symmetric difference, denoted by $\mu(A \bigtriangleup B)$. The symmetric difference is defined as $A \bigtriangleup B = (A \setminus B) \cup (B \setminus A)$. This is commonly used in the literature on level set estimation and conformal inference. Another commonly used notion of set distance is the Hausdorff distance. In Appendix~\ref{appendix: set metrics}, we discuss the relationship between these two metrics and why the Hausdorff distance is unsuitable in certain pathological cases. For estimation, let $\overline{\mathcal C}_M$ be the class obtained from $\mathcal C_M$ by allowing $a_m=b_m$, and set $\overline{\mathcal C}=\bigcup_{M\geq1}\overline{\mathcal C}_M$.

Let \(\hat{C}_{I}(x)= [\hat{\tau}_{0}(x), \hat{\tau}_{1}(x)]\) denote the solution to the empirical analogue of the optimisation problem defined in \autoref{def: feasible optimal} given a random sample $\mathcal{D}$:
\begin{equation}\label{eq:CIhat}
    \min_{t_{0} \leq t_{1}} (t_{1} - t_{0}) \mathtext{s.t.} P_{n}(t_{0},t_{1};x) \geq 1 - \alpha.
\end{equation}
Here we use the shorthand notation \(P(t_{0},t_{1};x) = P(t_{0}\leq Y^{L} \leq Y^{U} \leq t_{1}\mid X = x)\), and \(P_{n}(t_{0},t_{1};x)\) is an estimator for \(P(t_{0}, t_{1};x)\) based on the random sample \(\mathcal{D} = \{(X_{i}, Y_{i}^{L}, Y_{i}^{U}): i\in \mathcal{I}\}\) with sample size \(n\).

\begin{assumption}[Estimation]\label{asmp:estimation}
    The estimator \(P_{n}(t_{0},t_{1};x)\) is consistent for \(P(t_{0},t_{1};x)\) uniformly over $t_0\leq t_1$ and $x\in\mathcal X_n$, for some $\mathcal X_n\subseteq\mathcal X$,
    \begin{equation*}
        \sup_{x\in \mathcal{X}_{n}}\sup_{t_{0}\leq t_{1}} \abs{P_{n}(t_{0},t_{1};x) - P(t_{0},t_{1}; x)} = o_{p}(1).
    \end{equation*}
\end{assumption}
This assumption does not specify the estimator \(P_{n}\) we use. For example, a candidate estimator \(P_{n}(t_{0},t_{1};x)\) based on kernel smoothing is
\begin{equation}\label{eq:Phat}
    {P}_{n}(t_{0},t_{1};x):= \frac{\sum_{i\in \mathcal{I}} \indicator{t_{0}\leq Y_{i}^{L}\leq Y_{i}^{U}\leq t_{1}}K_{h}\pqty{X_{i} - x}}{\sum_{i\in \mathcal{I}} K_{h}\pqty{X_{i} - x}},
\end{equation}
where \(K_{h}(v) = h^{-d}K(v/h)\) for a kernel smoothing function \(K:\mathbb R^d\to[0,\infty)\) and bandwidth parameter \(h>0\), see \cite{tsybakov2009IntroductionNonparametric} and \cite{hansen2008UniformConvergence}. With a compactly supported kernel, the denominator can vanish when no sample observation lies in the kernel window around $x$. Whenever the denominator vanishes, we set the ratio to zero as in Appendix~\ref{appendix: kernel smoothing} and we use $\hat C(x)=\mathbb R$.
For this kernel estimator, $\mathcal{X}_{n}$ can be an expanding subset of $\mathcal{X}$ at a suitable rate to achieve uniform consistency as in \cite{hansen2008UniformConvergence}, Theorem 8 (see discussion in Appendix \ref{appendix: kernel smoothing}).
For high-dimensional covariates, we can apply distributional random forests \citep{cevid2022DistributionalRandomForests} to estimate $P([Y^L,Y^U]\subseteq C\mid X=x)$. Appendix~\ref{appendix:distributional-estimation} gives the steps for fitting the forest, calculating observation weights, estimating $P([Y^L,Y^U]\subseteq C\mid X=x)$, constructing prediction sets and applying conformalization.

Given a miscoverage level $\alpha \in (0,1)$, the next assumption we make is that \(C_{I}^*(x) = [\tau_{0}(x), \tau_{1}(x)]\) is the unique interval that satisfies the $1-\alpha$ coverage condition with the shortest length.
\begin{assumption}[Identification]\label{asmp:identification}
    There exists a unique solution \(C^*_{I}(x)=[\tau_{0}(x), \tau_{1}(x)]\) to the optimisation problem in \autoref{def: feasible optimal}. That is, \(P(\tau_{0}(x),\tau_{1}(x); x) \geq 1 -\alpha\), and, for any \((t_{0}, t_{1}) \neq (\tau_{0}(x), \tau_{1}(x))\) such that \(t_{1}- t_{0} \leq \tau_{1}(x) - \tau_{0}(x)\), then \(P(t_{0},t_{1}; x) < 1 - \alpha\). Additionally, for any \(\epsilon > 0\), there exists \(\delta > 0\) such that for all \(x\), if
    \begin{equation*}
        t_{1}(x) - t_{0}(x) \leq \tau_{1}(x) - \tau_{0}(x),
        \qquad
        \mu\pqty{[t_{0}(x), t_{1}(x)] \bigtriangleup
        [\tau_{0}(x), \tau_{1}(x)]} > \epsilon,
    \end{equation*}
    then \(P(t_{0}, t_{1}; x) < 1 - \alpha - \delta.\)
\end{assumption}
Assumption \ref{asmp:identification} states that we cannot achieve the same coverage probability with a different interval $(t_{0}(x),t_{1}(x))$ that is not larger than $(\tau_{0}(x), \tau_{1}(x))$. Figure \ref{fig:joint density} shows a specific joint density of \(Y^{L}\) and \(Y^{U}\). The set \(\{(a,b): t_{0}\leq a \leq b \leq t_{1}\}\) is a triangular region in \(\mathbb{R}^{2}\). In Figure \ref{fig:joint density}, the plotted triangular region corresponds to $t_{0}= -0.75$ and $t_{1}= 0.5$. $P(t_{0},t_{1};x)$ will be the integral of the joint density over the triangular region. The first part of Assumption \ref{asmp:identification} thus requires that if either the triangular region is smaller or if we move the triangular region away from the optimal position, we will strictly lose coverage probability of \((Y^{L}, Y^U)\). The second part is a regularity condition ensuring that the ``coverage loss'' is bounded away from $0$ over the support of $X$. This rules out the case when \(P(t_{0},t_{1};x)\) gets flatter around \((\tau_{0}(x),\tau_{1}(x))\) for some \(x\).
\begin{figure}[htbp]
    \centering
    \includegraphics[width=0.5\textwidth]{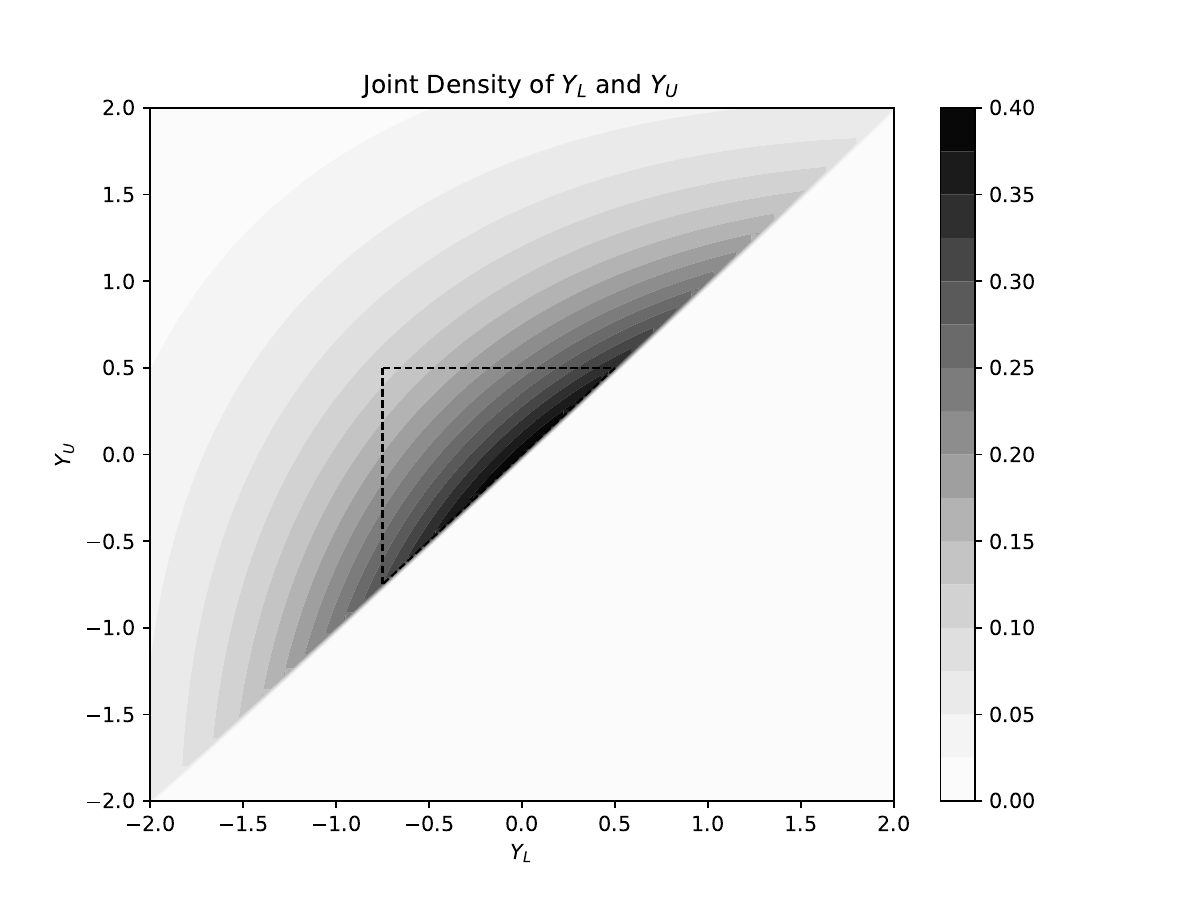}
    \caption{\footnotesize Joint density of $Y^L=Y-\kappa$ and $Y^U=Y+\kappa$, where $Y\sim N(0,1)$ and $\kappa$ is exponential with rate 2, independently of $Y$. The integral over the triangular region equals $P(-0.75,0.5;x)$.}
    \label{fig:joint density}
\end{figure}

In the absence of fixed censoring, we assume the following regularity condition for the conditional distribution \(P(t_{0},t_{1};x)\). The case with fixed censoring will be discussed later.
\begin{assumption}[Regularity]\label{asmp:regularity}
    There exists an \(\epsilon_{0}>0\) and positive constants \(c_{1}, \gamma_{1}\), such that for all \(x\in \mathcal{X}\) and \(\tau_{0}(x),\tau_{1}(x)\) as defined in Assumption \ref{asmp:identification},

    \begin{equation*}
        P\pqty{\tau_{0}(x)- \epsilon, \tau_{1}(x)+ \epsilon; x} - P\pqty{\tau_{0}(x), \tau_{1}(x);x} \geq c_{1}\epsilon^{\gamma_{1}},
    \end{equation*}
    for any \(0 < \epsilon < \epsilon_{0}\). In addition, for all \(0<\epsilon < \epsilon_{0}\), \(x\in\mathcal X\) and \(t_{0}\leq t_{1}\), \(P(t_{0}- \epsilon,t_{1}+ \epsilon;x) - P(t_{0},t_{1};x)\leq c_{2}\epsilon^{\gamma_{2}}\) for constants \(c_{2}, \gamma_{2}>0\).
\end{assumption}
Assumption \ref{asmp:regularity} is a smoothness assumption on \(P(t_{0},t_{1};x)\). In particular, we are assuming that by considering a slightly larger interval \([\tau_{0}- \epsilon, \tau_{1}+ \epsilon]\) we can strictly increase our coverage probability. The upper bound controls how much the probability \(P(\tau_0(x), \tau_1(x); x)\) can increase when the interval is slightly enlarged, uniformly over the candidate intervals and covariate values. This can be violated if the conditional distribution of \((Y^{L}, Y^{U})\) is discrete, which might happen under a fixed censoring scheme. This assumption is related to the \(\gamma\)-exponent condition in the level set literature, see \cite{polonik1995MeasuringMass} and \cite{lei2014DistributionfreePrediction}.

The following theorem shows that the estimated prediction interval is consistent for the oracle prediction interval under random censoring.
\begin{theorem}\label{thm:consistency}
    Under Assumptions \ref{asmp:estimation}, \ref{asmp:identification},  and \ref{asmp:regularity}, the prediction interval estimator \(\hat{C}_{I}(x)\) defined in Equation \eqref{eq:CIhat} is uniformly consistent for the oracle prediction interval \(C_{I}^{*}(x)\), that is,
    \begin{equation*}
        \sup_{x\in \mathcal{X}_{n}}{\mu\pqty{\hat{C}_{I}(x)\bigtriangleup C_{I}^{*}(x)} } = o_{p}(1)
    \end{equation*}
    as the sample size \(n\to \infty\).
\end{theorem}

Assumption \ref{asmp:regularity} can be violated when there is fixed censoring, i.e., when there are some fixed brackets, and we observe only the bracket in which an outcome falls. Suppose the conditional distribution of $(Y^L, Y^U)$ can be decomposed into a continuous part and a discrete part, where the discrete part has support over the set of points \(\{(L_{k}, U_{k})\in \mathbb{R}^2: 1 \leq k \leq K\}\).

In the case of a mixture of random and fixed censoring, the regularity condition in  Assumption~\ref{asmp:regularity} is violated. We impose the following assumption, which slightly modifies the estimator to allow for non-smoothness of the conditional distribution.
\begin{assumption}[Approximate solution]\label{asmp:Approximate solution fixed}
    The estimator for the oracle prediction interval \(\hat{C}_{I}(x) = [\hat{\tau}_{0}(x), \hat{\tau}_{1}(x)]\) is constructed by solving the following optimization problem, for $x \in \mathcal{X}$,
    \begin{equation}\label{eq:approx solution interval}
        \min_{t_{0} \leq t_{1}} \, \, t_{1}- t_{0} \mathtext{s.t.} P_{n}(t_{0}, t_{1}; x) \geq 1 - \alpha - \psi_{n},
    \end{equation}
    for some $\psi_{n}>0$ and $\psi_{n}\to 0$, such that for some $\mathcal X_n$,
    \begin{equation}\label{eq:cond rate psi_n}
        P\pqty{\sup_{x\in \mathcal{X}_n} \sup_{t_{0}\leq t_{1}}\abs{P_{n}(t_{0},t_{1};x) - P(t_{0},t_{1};x)} > \psi_{n}}=o(1).
    \end{equation}
\end{assumption}

To see the reason why fixed censoring creates a problem, and we should not solve the exact constraint, consider the following example.
\begin{example}\label{example:fixed censoring error}
    Suppose $(Y_i^L, Y_i^U) = (0,1)$ or $(0,2)$ with equal probability without predictor information $X$. The oracle prediction set for a miscoverage level $\alpha = 0.5$  (see Equation \ref{eq:level set union}) is $C^* = [0,1]$. (Notice that $\inf_{P\in \mathcal{P}_I} P\pqty{Y\in C^*} = 0.5$ and a shorter interval would not achieve the desired miscoverage level.)  If we solve the exact empirical analogue of the oracle prediction set with the natural estimator $P_n(a,b) = n^{-1} \sum \indicator{[Y_i^L, Y_i^U] \subseteq [a,b]}$ as in \autoref{eq:CIhat}, we obtain a prediction interval
    $$\check{C} = \bqty{0, 1 + \indicator{\bar{Y}^U > 1.5}},$$
    where $\bar{Y}^U = \frac{1}{n}\sum_i Y_i^U$. Since the lower bound for $(Y^L, Y^U)$ is 0 with probability one, it is easy to see that the lower bracket of $\check{C}$ is $0$.  Because
    $$P_n(0,b) =
    \begin{cases}
        0 &\mathtext{if} 0 < b < 1 \\
        n^{-1} \sum \indicator{Y_i^U =1} &\mathtext{if} 1\leq b < 2 \\
        1 &\mathtext{if} b\geq 2
    \end{cases},$$
    the upper bracket of $\check{C}$ is $2$ if $P_n(0,1) < 0.5$, i.e. there are fewer observations of $Y_i^U =1$ than $Y_i^U = 2$ in the sample, which is equivalent to having the average $\bar{Y}^U > 1.5$. Since $\bar{Y}^U$ has a scaled Binomial distribution $\bar{Y}^U \sim 1 + n^{-1} \text{Binom}(n,0.5)$, and the median for a $\text{Binom}(n,p)$ distribution is between $\lfloor np \rfloor$ and $\lceil np \rceil$, we have that $P(\bar{Y}^U>1.5) \to 0.5$ as $n \to \infty$. This means that $P(\check{C} = [0,2]) \to 0.5$ as $n\to \infty$ and \(\check{C}\) is not consistent. On the other hand, for any $0 < \psi_n < 0.5$, $P_n(0,1) < 0.5-\psi_n$ is equivalent to $\bar{Y}^U > 1.5 + \psi_n$, which has probability approaching $0$ by the central limit theorem if $\sqrt{n}\psi_n \to \infty$ as the sample size $n \to \infty$. Solving the modified optimisation problem in \eqref{eq:approx solution interval} will then lead to a consistent estimator of the oracle prediction set in this example. Also note that the condition in \eqref{eq:cond rate psi_n} is a generalisation of the condition $\sqrt{n}\psi_n \to \infty$.
\end{example}

With \autoref{asmp:Approximate solution fixed} in place of Assumptions \ref{asmp:estimation} and \ref{asmp:regularity}, we obtain the following result:
\begin{theorem}\label{thm:consistency fixed}
    Under Assumptions \ref{asmp:identification}, and \ref{asmp:Approximate solution fixed}, we have the uniform consistency of $\hat{C}_{I}$ for the oracle prediction interval $C_{I}^{*}$, that is, as the sample size $n\to \infty$,
    \begin{equation*}
        \sup_{x\in \mathcal{X}_n} \mu\pqty{\hat{C}_{I}(x)\bigtriangleup C^{*}_{I}(x)} = o_{p}(1).
    \end{equation*}
\end{theorem}

We now turn to the more general case where the estimation target is the {oracle prediction set} $C^*$ defined in \autoref{def: feasible optimal}.

The next assumption is on the identification of the oracle prediction set, which generalises Assumption \ref{asmp:identification}. Write $M_*(x)$ for the number of disjoint components of $C^*(x)$.  We assume a fixed upper bound $\bar M$ on $M_*(x)$.
\begin{assumption}\label{asmp:identification multiple}
    \begin{enumerate}
        \item There exists a unique solution $C^*$ to \autoref{def: feasible optimal}. For any $\epsilon > 0$, there exists $\delta = \delta(\epsilon) > 0$, such that for every $x\in\mathcal X$ and any set ${C}'(x) \in \overline{\mathcal C}$ with $\mu({C}'(x)) \leq \mu(C^{*}(x))$ and $\mu(C^{*}(x) \bigtriangleup C'(x)) > \epsilon$, we have $P([Y^L, Y^U] \subseteq {C}'(x)\mid X = x ) < 1 - \alpha-\delta$.
        \item The oracle \(C^*(x)\in\mathcal C_{\bar M}\), for some \(\bar M\in\mathbb N\) and every \(x\in\mathcal X\).
    \end{enumerate}
\end{assumption}

Assumption~\ref{asmp:identification multiple} requires that a set of no greater volume than $C^*(x)$ loses a positive amount of coverage if it stays away from $C^*(x)$, with a gap uniform in $x$. Theorem~\ref{thm:random-radius-separation}(3) in Appendix~\ref{appendix:identification} gives one set of primitive conditions for this uniform separation assumption.

Let $P(C(x);x) \equiv P([Y^{L}, Y^{U}]\subseteq C(x) \mid X = x)$ and suppose that we have a uniformly consistent estimator $P_n(C; x)$. We define the estimator $\hat{C}(x), x\in \mathcal{X}$, for the oracle prediction set $C^*$, allowing for potential fixed censoring.
\begin{assumption}\label{asmp:approximate solution multiple}
    The estimators $\hat{{C}}(x), x\in \mathcal{X}$ are constructed by solving the following problem, for any $x\in \mathcal{X}$, and $M \geq \bar{M}$,
    \begin{equation*}
        \min_{C(x)\in \overline{\mathcal C}_M} \mu(C(x)) \mathtext{s.t.} P_{n}(C(x); x) \geq 1 - \alpha - \psi_{n}.
    \end{equation*}
    Here $\psi_{n}=o(1)$ is a sequence of positive numbers and $P_{n}(C;x)$ estimates $P(C;x)$ such that, for some $\mathcal X_n\subseteq\mathcal X$,
    $$P\pqty{{\sup_{C\in \overline{\mathcal C}_M}}\sup_{x\in \mathcal{X}_n} \abs{P_{n}(C(x);x) - P(C(x);x)} > \psi_{n}}=o(1),$$
    as $n\to \infty$.
\end{assumption}

A few remarks are in order. First, Appendix~\ref{appendix: kernel smoothing} gives conditions for uniform consistency of the kernel estimator $P_n(C;x)$ and a convergence rate that can guide the choice of $\psi_n$. Empirically, varying $\psi_n$ between zero and small positive values has little effect on the results. Uniformity is imposed over $\overline{\mathcal C}_M$, as the proof uses the VC property of the associated class of indicator functions for fixed $M$. Second, Appendix~\ref{appendix:optimized-algorithms} describes computational algorithms for solving the sample optimisation problem.

Also, we do not assume that the exact value of $\bar M$  is known. For any fixed $M\geq\bar M$, the class $\mathcal C_M$ contains $C^*$ and therefore yields the same population minimum volume. If $M<M_*(x)$ at a given $x$, one may instead target a restricted oracle there, defined as a minimum-volume feasible set in $\mathcal C_M$. An analogous consistency argument requires existence and uniform separation of this target within $\overline{\mathcal C}_M$, together with uniform estimation and a suitable choice of $\psi_n$. Appendix~\ref{appendix:empirical-m2} compares $M=1$ and $M=2$ in both empirical applications to assess sensitivity to this choice. Estimation of $M_*(x)$ is an interesting question, intuitively related to estimating the number of modes.

\begin{theorem}\label{thm:consistency multiple}
    Under Assumptions \ref{asmp:identification multiple} and \ref{asmp:approximate solution multiple}, we have that $\hat{C}$ is a consistent estimator of ${C}^*$ in the sense that, as $n\to \infty$,
    \begin{equation*}
        \sup_{x\in \mathcal{X}_n} {\mu\pqty{\hat{C}(x) \bigtriangleup C^*(x)}} = o_{p}(1).
    \end{equation*}
\end{theorem}

We have proposed consistent estimators for the oracle prediction set. There has been a recent interest in constructing prediction sets that are finite-sample valid via conformal inference. The next section discusses how to construct a conformal prediction set that is finite-sample valid under interval censoring.

\section{Conformal inference with interval outcomes}\label{sec:conformal}

Given a consistent estimator for the oracle prediction set \(\hat{C}\), this section considers the problem of constructing a prediction set with finite-sample validity. An estimated prediction set \(\tilde{C}\) with a random sample $\mathcal{D} = \cqty{(X_{i}, Y^{L}_{i}, Y^{U}_{i}), i \in \mathcal{I}}$  is said to be \textit{finite-sample marginally valid}, if for all \(P\),
\begin{equation*}
    \Prob \pqty{Y_{n+1} \in \tilde{C}(X_{n+1})} \geq 1 - \alpha,
\end{equation*}
where \(\Prob\) denotes the joint probability of \(\cqty{\bar{Z}_{i} = (X_{i}, Y_{i}, Y^L_{i}, Y_{i}^{U}): i \in \mathcal{I}\cup\cqty{{n+1}}}\). Notice that this is the finite-sample validity condition typically considered in conformal inference literature, where the probability is the joint distribution for both the observed sample and the individual to be predicted. While the oracle prediction set \(C^*(x)\) depends on the true distribution \(P\), conformal inference constructions lead to a prediction set that is finite-sample valid for all \(P\), making it distribution-free.

In particular, we will achieve finite-sample validity using the method of \textit{split conformal inference} (see \cite*{romano2019ConformalizedQuantile} and \cite*{lei2018DistributionfreePredictive}). One feasible split conformal inference procedure for the oracle prediction interval $C^*_I$ in our setting can be described as follows. We first randomly split the sample \(\mathcal{D}\) into a training set \(\mathcal{D}_{1} = \cqty{Z_{i}: i\in \mathcal{I}_{1}}\) and a calibration set \(\mathcal{D}_{2}= \cqty{Z_{i}: i \in \mathcal{I}_{2}}\). Let $n_r=|\mathcal I_r|$ for $r=1,2$, and take the split independently of the observations. We will assume that \(n_2\) is large enough that \(\alpha\geq1/(n_2+1)\). We calculate the estimator \(\hat{C}_{I}(x) = [\hat{\tau}_{0}(x), \hat{\tau}_{1}(x)]\) for the oracle prediction set with the training set for \(x\in \mathcal{X}\) as in \autoref{eq:CIhat}. The following scores \(s_{j} = s(Y^{L}_{j}, Y^{U}_{j}, X_{j}; \hat{\tau}_{0}, \hat{\tau}_{1})\) are then computed for each \(j \in \mathcal{I}_{2}\) in the calibration set, given \(\hat{C}_{I}(x) = [\hat{\tau}_{0}(x), \hat{\tau}_{1}(x)]\), where
\begin{equation*}
    s\pqty{y^{L}, y^{U}, x; t_{0}(\cdot), t_{1}(\cdot)} = \max \cqty{t_{0}(x) - y^{L}, y^{U} - t_{1}(x)}
\end{equation*}
This score function is a modified version of the score function used in the case of conformalized quantile regression with no interval censoring studied in \cite{romano2019ConformalizedQuantile}. The score is non-positive exactly when $[y^L,y^U]\subseteq[t_0(x),t_1(x)]$. By allowing for negative scores, the conformal prediction set \(\tilde{C}_{I}(x)\) penalises both overcoverage and undercoverage.

We compute the quantile of the conformity scores as
\[
    \vartheta_{1 - \alpha} = q\pqty{\cqty{s_{j}: j \in \mathcal{I}_{2}}; (1 - \alpha)(1+1/n_2)}.
\]
Here \(q(\mathcal S;\zeta) :=\inf\{u\in\mathbb R: {\abs{\mathcal S}}^{-1}\sum_{s\in\mathcal S}\indicator{s\leq u} \geq\zeta\}\) is the lower empirical \(\zeta\)-quantile of a finite nonempty multiset \(\mathcal S\), for \(\zeta\in(0,1]\).
The conformal prediction set is then
\begin{equation}\label{eq:conformalized interval}
    \tilde{C}_{I}(x) = [\tilde{\tau}_{0}(x), \tilde{\tau}_{1}(x)]  = \bqty{\hat{\tau}_{0}(x) - \vartheta_{1 - \alpha}, \hat{\tau}_{1}(x) + \vartheta_{1 - \alpha}}.
\end{equation}
Intuitively, when our estimated prediction set \(\hat{C}_{I}\) is wider than necessary, the intervals \(\{(Y_{j}^{L}, Y_{j}^{U}), j \in \mathcal{I}_{2}\}\) tend to be contained within \(\hat{C}_{I}\), leading to a higher proportion of negative scores in the calibration set. As a result, the chosen quantile \(\vartheta_{1 - \alpha}\) is negative, and the conformal prediction interval adjusts for the over-coverage of \(\hat{C}_I\). The correction for under-coverage follows the same logic.

When the prediction set $\hat{C}$ is a union of multiple intervals, we can construct a finite-sample valid conformal prediction set as follows. First we split the samples into a training set $\mathcal{I}_{1}$ and a calibration set $\mathcal{I}_{2}$. We estimate $\hat C$, with $\hat C(x)\in\overline{\mathcal C}_M$, using data in the training set, and then compute the conformity scores $s_{j}= s(Y^{L}_{j}, Y^{U}_{j}, X_{j}; \hat{C})$, for each $j\in \mathcal{I}_{2}$. In order to allow for $\hat{C}$, which is a union of intervals, we compute the conformity score as follows. For $C(x)\in\overline{\mathcal C}_M$, write $C(x)=\bigsqcup_m[t_{0m}(x),t_{1m}(x)]$. Throughout the formulas below, component indices range only over the intervals actually present, of which there are at most $M$. Order them so that $t_{0m}(x)\leq t_{1m}(x)<t_{0,m+1}(x)$ for adjacent components. Define, for every $u\in\mathbb R$, \(\mathcal E(C(x),u) :=\bigcup_m [t_{0m}(x)-u,t_{1m}(x)+u],\)\footnote{We use the convention $[a,b]=\varnothing$ whenever $a>b$.} and the conformity score by
    \begin{equation}\label{eq:score function}
        s(y^{L}, y^{U}, x; C)
        = \inf\cqty{u: [y^L,y^U]\subseteq\mathcal E(C(x),u)}.
    \end{equation}
    Because the family $\mathcal E(C(x),u)$ is nested and right-continuous in $u$, the infimum in Equation~\eqref{eq:score function} is attained.  Hence, for every $u\in\mathbb R$,
    \begin{equation}\label{eq:score event equivalence}
        s(y^{L},y^{U},x;C)\leq u
        \quad\Longleftrightarrow\quad
        [y^L,y^U]\subseteq\mathcal E(C(x),u).
    \end{equation}

This score can be computed explicitly as follows. Fix $x$ and write
$a_m=t_{0m}(x)$ and $b_m=t_{1m}(x)$.
The distance from a point $z$ to $C(x)$ is
\[
    d_x(z)
    := \inf_{w\in C(x)} |z-w|
    = \min_m
      \max\{a_m-z,\,0,\,z-b_m\}.
\]
For each gap between consecutive components, let
\[
    \xi_m := \frac{b_m+a_{m+1}}{2},
    \qquad\text{for consecutive components }m,m+1,
\]
and define
\[
    g_x(y^L,y^U)
    := \max\left(
        \{0\}\cup
        \left\{
            \frac{a_{m+1}-b_m}{2}:
            y^L\leq \xi_m\leq y^U
        \right\}
    \right).
\]
Thus $g_x(y^L,y^U)=0$ when no gap midpoint lies in $[y^L,y^U]$.
The score is
\[
    s(y^L,y^U,x;C)
    =
    \begin{cases}
        \max\{a_j-y^L,\,y^U-b_j\},
        & \text{if } a_j\leq y^L\leq y^U\leq b_j
          \text{ for some } j,\\[4pt]
        \max\{d_x(y^L),\,d_x(y^U),\,g_x(y^L,y^U)\},
        & \text{otherwise}.
    \end{cases}
\]
To see this, if the observed interval is contained in one component,
the first expression gives the smallest signed adjustment needed
for containment. Otherwise the required adjustment is positive and
equals $\sup_{z\in[y^L,y^U]}d_x(z)$. This maximum occurs at an endpoint
of $[y^L,y^U]$ or at the midpoint of a gap between consecutive
components lying in $[y^L,y^U]$.

Notice that $s_j\leq0$ if and only if $[Y_j^L,Y_j^U]\subseteq\hat C(X_j)$. Compute $\vartheta_{1-\alpha}$ from these scores using the empirical-quantile formula above. Let $\hat\tau_{0m}(x)$ and $\hat\tau_{1m}(x)$ be the endpoints of the $m$th component of $\hat C(x)$. The conformal prediction set is then given by \footnote{Appendix~\ref{appendix: multiple scores} allows different endpoint adjustments, chosen using the training sample, with scores defined as in Equation~\eqref{eq:score function}. For other approaches to conformity-score construction, see \cite{thurin2025OptimalTransportbased}.}
\begin{equation}\label{eq:conformal prediction set}
    \tilde{C}(x)
    =\mathcal E(\hat C(x),\vartheta_{1-\alpha})
    =\bigcup_m
    [\hat{\tau}_{0m}(x)-\vartheta_{1-\alpha},
    \hat{\tau}_{1m}(x)+\vartheta_{1-\alpha}].
\end{equation}
For both constructions, we assign score $-\infty$ when the fitted set is $\mathbb R$ and retain $\mathbb R$ after conformalization.

The next theorem states the finite-sample marginal validity of the conformal prediction set \(\tilde{C}\) defined in \autoref{eq:conformal prediction set} above. It implies the validity of the conformal prediction interval $\tilde{C}_I$.

\begin{theorem}\label{thm:marginal}
    Under the assumption that \(\{(X_{i}, Y_{i}, Y_{i}^{L}, Y_{i}^{U}): i\in \mathcal{I}\cup \{n+1\}\}\) are i.i.d.\footnote{In order to perform conformal inference, it is sufficient that the sample satisfies the weaker condition of exchangeability. Here, independence is assumed for simplicity and for establishing the asymptotics for the nonparametric estimators we will use. \cite{barber2023ConformalPrediction} studies conformal inference beyond exchangeability. Extending our results to time series is an interesting future topic.}, the conformal prediction set $\tilde{C}$ in \autoref{eq:conformal prediction set} is finite-sample marginally valid, that is, for all \(P\),
    \begin{equation*}
        \Prob\pqty{Y_{n+1}\in \tilde{C}(X_{n+1})} \geq 1 - \alpha.
    \end{equation*}
    where $\Prob$ denotes the joint probability of \(\cqty{\bar{Z}_{i}: i \in \mathcal{I}\cup\cqty{{n+1}}}\).
\end{theorem}

The following theorem shows that the conformal prediction set \(\tilde{C}(x)\) preserves the consistency of the estimator \(\hat{C}(x)\) for the oracle prediction set. Recall the definition of $s(y^{L}, y^{U}, x; C)$ in Equation \eqref{eq:score function}, let \(s_{j} = s(Y^{L}_{j}, Y^{U}_{j}, X_{j}; \hat{C})\), and \(s_{j}^{*} = s\pqty{Y_{j}^{L}, Y_{j}^{U}, X_{j}; C^*}\).

\begin{theorem}\label{thm:consistency of conformal}
    Let \(\mathcal{D} = \{(X_i,Y_i^L,Y_i^U):i\in\mathcal I\}\) be an i.i.d. sample split into training and calibration sets indexed by \(\mathcal I_1\) and \(\mathcal I_2\), respectively. Let $\tilde{C}(x)$ be the conformal prediction set constructed by the preceding procedure in \autoref{eq:conformal prediction set}.
        Suppose that the sets \(\mathcal X_{n_1}\) in Assumption~\ref{asmp:approximate solution multiple} are deterministic and satisfy \(P_X(X\in\mathcal X_{n_1})\longrightarrow1.\) Suppose also that, for some \(\eta_0>0\), the oracle score satisfies
        \begin{equation}\label{eq:right crossing oracle score}
            P\pqty{s(Y^L,Y^U,X;C^*)\leq t}>1-\alpha
            \qquad\text{for every }t\in(0,\eta_0).
        \end{equation}
    Then, under the assumptions of \autoref{thm:consistency multiple},
        \begin{equation*}
            \sup_{x\in\mathcal X_{n_1}} \mu\pqty{{\tilde{C}(x) \bigtriangleup C^*(x)}} = o_{p}(1)
        \end{equation*}
    as the sample sizes \(n_{1}, n_{2}\to \infty\), here \(n_{1}= \abs{\mathcal{I}_{1}}\) and \(n_{2}= \abs{\mathcal{I}_{2}}\).
\end{theorem}

Let $F_*(t)=P\pqty{s(Y^L,Y^U,X;C^*)\leq t}$. Condition \eqref{eq:right crossing oracle score} holds in standard regular settings. It holds if $F_*(0)>1-\alpha$, and, when $F_*(0)=1-\alpha$, it holds if the oracle score has a positive density immediately to the right of zero. In the single-interval setting, it also follows from Assumption~\ref{asmp:regularity}. The condition may fail under fixed censoring if $F_*(0)=1-\alpha$ and the oracle score has a positive gap above zero. In that case, finite-sample validity is preserved, but the conformal adjustment need not vanish and consistency for $C^*$ is not guaranteed. Remark~\ref{remark:fixed censoring conformal error} in Appendix~\ref{appendix: proofs} quantifies the resulting asymptotic symmetric-difference error in terms of the oracle-score gap.

It is tempting to look for a finite-sample conditionally valid prediction set \(\tilde{C}(x)\) such that for all \(P\),
\begin{equation*}
    \Prob \pqty{Y_{n+1} \in \tilde{C}(x) \mid X_{n+1} = x} \geq 1 - \alpha.
\end{equation*}
However, it is impossible to construct a prediction set that is finite-sample conditionally valid and consistent for the oracle prediction set \citep{lei2014DistributionfreePrediction, barber2021LimitsDistributionfree, gibbs2024ConformalPrediction}, and hence we consider the weaker notion of local validity. Let \(\mathcal{A}= \cqty{A_{k}: k\leq K}\) be a partition of \(\mathcal{X}\).  A prediction interval \({C}_{\mathrm{loc}}(x)\) is \textit{finite-sample locally valid} with respect to \(\mathcal{A}\), if, for all \(P\) and every cell $A_k$ with $P_X(A_k)>0$,
\begin{equation*}
    \Prob \pqty{Y_{n+1} \in {C}_{\mathrm{loc}}(X_{n+1}) \mid {X_{n+1}\in A_k} } \geq 1 - \alpha.
\end{equation*}

Let $\mathcal I_{2,k}=\{i\in\mathcal I_2:X_i\in A_k\}$ and $n_{2,k}=|\mathcal I_{2,k}|$. Define
\[
    \vartheta_{1-\alpha,k}
    =q\pqty{\cqty{s_j:j\in\mathcal I_{2,k}};
        (1-\alpha)(1+1/n_{2,k})}.
\]
When $\alpha<1/(n_{2,k}+1)$, including empty cells, set $\vartheta_{1-\alpha,k}=+\infty$ and interpret $\mathcal E(C,+\infty)=\mathbb R$.
For $x\in A_k$, the local conformal prediction set is
\begin{equation}\label{eq:local conformal pred set}
    \tilde C_{\mathrm{loc}}(x)
    =\mathcal E\pqty{\hat C(x),\vartheta_{1-\alpha,k}}.
\end{equation}

The finite-sample local validity is a result of exchangeability, similar to \cite{lei2014DistributionfreePrediction}.
\begin{theorem}\label{thm:local validity}
    Suppose \(\{(X_{i}, Y_{i}^{L}, Y_{i}^{U}): i\in \mathcal{I}\cup\{n+1\}\}\) are i.i.d.\ Let \(\mathcal A\) be a fixed partition of \(\mathcal X\). Then the local conformal prediction set \(\tilde{C}_{\mathrm{loc}}(x)\) is finite-sample locally valid with respect to \(\mathcal{A}\).
\end{theorem}

In the next section, we demonstrate the performance of the conformal and local conformal prediction sets in a set of numerical experiments.
\section{Numerical experiments}\label{sec:simulation}
We report here the results of a series of numerical experiments to evaluate the performance of the conformal prediction sets. In all the experiments, we randomly draw a sample of predictors \(\cqty{X_{i}, i= 1,2,\dots, n}\), with sample size \(n=2,500\) and each \(X_{i}\sim \text{Unif}[-1.5, 1.5]\). The miscoverage level is fixed at \(\alpha = 0.1\). The unobserved outcomes \(Y_{i}\) and the observed intervals \((Y_{i}^{L}, Y_{i}^{U})\) are then generated based on one of the models described below.

\begin{itemize}
    \item \textit{Model A.} The unobserved outcomes \(Y_i\) are generated according to a mixture model with heteroscedasticity, similar to the simulations in \cite{lei2014DistributionfreePrediction}.
        \begin{equation}\label{eq:mixture model}
            (Y_i\mid X_i=x,B_i)\sim\mathcal N\bigl(f(x)+(2B_i-1)g(x),\sigma^2(x)\bigr)
        \end{equation}
        Here, the auxiliary random variable \(B_i\) is independently drawn from a Bernoulli distribution with probability \(0.5\), and the following functions are used.
        \begin{align*}
            f(x) = 2(x-1)^{2}(x+1);\, g(x) = 4 \sqrt{(x + 0.5) \indicator{x\geq -0.5}};\,  \sigma^{2}(x) = \frac{1}{4} + \abs{x},
        \end{align*}
        The observed intervals \((Y_i^L, Y_i^U)\) are obtained by adding independent noise to the unobserved outcomes.
        \begin{equation}\label{eq:simulation random censor}
            Y_{i}^{L} = Y_{i} - e_{1,i}, \quad Y_{i}^{U} = Y_{i} + e_{2,i}, \quad e_{1,i}, e_{2,i} \sim_{\text{i.i.d.}} |N(0,1)|.
        \end{equation}
        In this design, the conditional distribution \(P_{Y\mid X}\) is unimodal for smaller values of \(X\) and bi-modal for larger values of \(X\). Regression and quantile regression-based prediction sets will have difficulty capturing the bifurcation shape of the oracle prediction set.

    \item \textit{Model B.}  The unobserved outcomes are generated in the same way as in \autoref{eq:mixture model} above. Each outcome is reported in a fixed bracket with probability 0.2 and observed exactly otherwise. Let \(\iota_{i}\) be a random draw from a Bernoulli distribution with probability 0.2,
        \begin{equation}\label{eq:simulation fixed censor}
            \begin{cases}
                Y_{i}^{L} = Y_{i}^{U} = Y_{i} &\mathtext{if} \iota_{i} = 0;\\
                Y_{i}^{L} = \lfloor Y_{i} \rfloor, Y_{i}^{U} = Y_{i}^{L} +1 &\mathtext{if} \iota_{i} = 1.
            \end{cases}
        \end{equation}
        This model is designed to evaluate the performance of the conformalized prediction set under fixed censoring. We choose to censor 20\% of the outcomes, as approximately 20\% of the income data is interval valued in the dataset used for our empirical studies.
    \item \textit{Model C.}  The unobserved outcomes are generated according to the following model,
        \begin{equation*}
            Y_{i} = f(X_{i}) + \epsilon_{i}, \quad \epsilon_{i} \sim_{\text{i.i.d.}} \chi^{2}_{1.5},
        \end{equation*}
        and \((Y_{i}^{L}, Y_{i}^{U})\) are generated according to \autoref{eq:simulation fixed censor}. In this case, the conditional distribution is unimodal and skewed.
\end{itemize}

The observations \(\cqty{(X_{i}, Y_{i}^{L}, Y_{i}^{U}), i = 1,\dots,n}\), are randomly split into a training set $\mathcal{I}_{1}$ and a calibration set $\mathcal{I}_{2}$, with the training set containing 75\% of the observations. We estimate the prediction set \(\hat{C}(x)\) based on the kernel smoothing estimator \(P_{n_1}(C;x)\) with the Epanechnikov kernel and bandwidth $h=0.25$ for the conditional distribution $P(C; x)$, as described in Equation (\ref{eq:Phat}), using the observations in the training set. We use $M=2$ and $\psi_{n_1}=0$, with the exact endpoint solver on a grid of 50 equally spaced covariate points and linear interpolation of the ordered endpoints between them. The conformal prediction sets \(\tilde{C}\) are then constructed according to Equation (\ref{eq:conformal prediction set}) using the calibration set. Local conformal prediction sets \(\tilde{C}_{\mathrm{loc}}\) are obtained using the local conformalization procedure, where the support \(\mathcal{X}\) is divided into 5 equal-sized bins.

We then compare the average coverage and volume of the conformal prediction set \(\tilde{C}\) and the local conformal prediction set \(\tilde{C}_{\mathrm{loc}}\) over 100 repetitions with two alternative conformal prediction intervals: \(\tilde{C}_{q2}\), based on a quantile regression assuming a quadratic model, and \(\tilde{C}_{q3}\), based on a quantile regression assuming a cubic model. See \cite{romano2019ConformalizedQuantile} and \autoref{sec:quantile} for more details about the construction of a prediction interval based on quantile methods. The coverage is computed by randomly generating 5,000 samples of \((Y^{L}, Y^{U})\) at each grid point and calculating the proportion of whole intervals covered by the prediction sets. The volume is calculated as the integrated Lebesgue measure of the prediction sets over the support \(\mathcal{X}\), using a trapezoidal approximation on the same grid.

Figure \ref{fig:compare_prediction_sets} shows the comparison of the prediction sets from the numerical experiments. The prediction sets \(\tilde{C}\) and \(\tilde{C}_{\mathrm{loc}}\) capture the bifurcation region of the conditional distribution in \textit{Model A} and \textit{Model B} and have a smaller average integrated volume compared to the quantile regression-based prediction sets. Even when the conditional distribution is unimodal, our proposed estimators still have smaller volume as demonstrated in \textit{Model C}. In terms of coverage, \(\tilde{C}, \tilde{C}_{\mathrm{loc}}\) and \(\tilde{C}_{q3}\) appear to have coverage averaged over the evaluation grid close to the nominal level $1-\alpha = 0.9$. Notably, \(\tilde{C}_{\mathrm{loc}}\) appears to have better coverage than \(\tilde{C}\) near the boundary of the support \(\mathcal{X}\).

\begin{figure}[p]
    \centering
    \includegraphics[width=\textwidth,height=0.86\textheight,keepaspectratio]{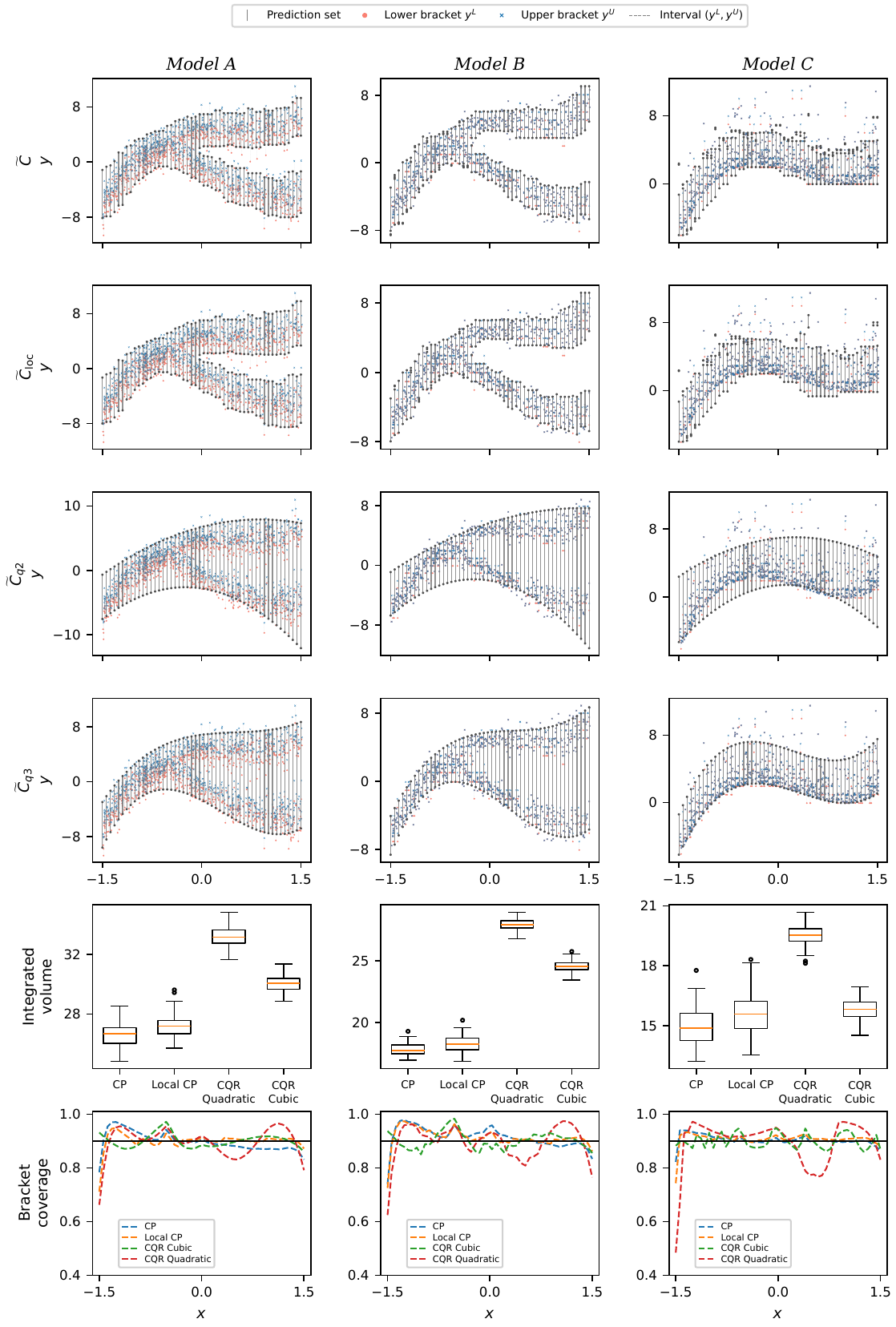}
    \caption{\small Comparison of prediction sets in simulation study. Rows 1--4 show the first 800 observations from the first replication and the conformal prediction set \(\tilde{C}\), the local conformal prediction set \(\tilde{C}_{\mathrm{loc}}\), and the conformal prediction sets constructed with quadratic and cubic quantile regression, \(\tilde{C}_{q2}\) and \(\tilde{C}_{q3}\). Row 5 shows boxplots of integrated volume over 100 repetitions, and row 6 shows mean whole-interval coverage at each grid point.}
    \label{fig:compare_prediction_sets}
\end{figure}

\section{Empirical study}\label{sec:empirical}

\subsection{UK Job Advert Data}\label{sec:empirical_uk}

We have also applied our method to the UK job advert dataset from Adzuna, an online job advert aggregator in the UK. Adzuna data have been used widely as an indicator of economic activity in the UK, including by the Office for National Statistics as experimental statistics \citep{ons2021UsingAdzuna}. It has also been used to develop a measure of labour market opportunities for heterogeneous types of workers \citep{dias2021WorkerMobility}.

The original sample consists of job adverts posted in the United Kingdom between January and March of 2022, which comprise 121,746 observations. Each advertisement contains information on job location, job category, and salary, when available.  Salary information is missing for 37.8\% (46,077) of the postings. Among the advertisements that report salaries, the majority are listed as salary intervals, accounting for approximately 62\% (46,606) of all postings with salary information, making this dataset well-suited for our interval-valued outcome framework.

The outcome of interest is the annual salary associated with each job advert, with $(Y_i^L, Y_i^U)$ denoting the reported lower and upper salary bounds for the $i$-th advert, and we work with the logarithm of the salary bounds. In the analysis below, we focus on the IT industry, which is the largest job category by number of postings (see \autoref{fig:job distribution by category}). For the 16,707 job advertisements in the IT industry during January--March 2022, 46\% (7,742 postings) report salary information after excluding inverted bounds, of which 64\% (4,981 postings) provide salaries in intervals. In our analysis, we restrict attention to the subset of postings that report salary information. Job location is represented by geographic coordinates (latitude and longitude), which we treat as continuous covariates.
Estimation and conformalization are performed on the log-salary scale, and the reported bounds are transformed back to pounds.

For a given location (latitude and longitude), we estimate the 90\% oracle prediction interval (see \autoref{def: optimal}) for the log-salary distribution in the IT category conditional on location. We employ the kernel smoothing estimator of the conditional probability using a product Epanechnikov kernel applied to geographic coordinates and solve for the empirical analogue of the feasible optimization problem in Equation~\eqref{eq:feasible criterion}, yielding prediction interval $\hat C_I(x)=[\hat\tau_0(x),\hat\tau_1(x)]$. To visualise spatial variation in predicted salary distributions, we evaluate these estimated bounds before conformalization on a grid of locations covering the UK and interpolate the resulting bounds. We take $M=1$, and provide in Appendix~\ref{appendix:empirical-m2} a comparison with $M=2$, which does not reduce mean conformalized length on the log-salary scale at the 90\% target. For the city intervals in Table~\ref{tab:city_salary_bounds}, we refit and construct $\tilde C_{\mathrm{loc}}(x)$ according to Equation~\eqref{eq:local conformal pred set}, using the specification in Appendix~\ref{appendix:empirical-m2}. The local conformalization uses eight geographic regions determined from training locations alone and fixed before conformalization. The local coverage guarantee applies within these regions. Appendix~\ref{appendix:empirical-m2} gives the sample restrictions, bandwidth and partition choices for these conformalized results, and further diagnostics.

Figure~\ref{fig:job-advert} plots the estimated lower and upper salary bounds from $\hat C_I(x)$ at the 90\% target across geographic locations in the UK. For reference, the median annual pay for full-time employees was £33,000 for the tax year ending on 5 April 2022.\footnote{Source: Office for National Statistics, \href{https://www.ons.gov.uk/employmentandlabourmarket/peopleinwork/earningsandworkinghours/bulletins/annualsurveyofhoursandearnings/2022}{\emph{Employee earnings in the UK: 2022}}, Annual Survey of Hours and Earnings.} We can observe a few patterns in the spatial distribution of the predicted salary bounds. First, the estimated lower bounds exhibit moderate geographic heterogeneity and are relatively even across space and are typically below £30,000. Lower bounds tend to be higher around major cities, while rural areas display lower predicted lower bounds or no job postings. Second, the upper bounds display a much thicker right tail and substantially stronger spatial concentration, with predicted values ranging roughly from £50,000 to £230,000. Areas surrounding major metropolitan centres, such as London, Manchester, Bristol, Birmingham and Edinburgh, stand out with markedly higher upper bounds than the rest of the country, indicating pronounced geographic concentration of high-paying jobs. As a result, the predicted salary interval is widest in London and other large cities, reflecting greater wage dispersion and job heterogeneity in these labour markets.

Table~\ref{tab:city_salary_bounds} reports the 90\% local conformal prediction intervals for annual IT salaries at selected major UK cities, evaluated at their geographic coordinates using the first prespecified split. The results illustrate cross-city heterogeneity, particularly in the upper bounds of the predicted salary distributions. London has the highest upper bound and the widest prediction interval among the listed cities. Edinburgh has the second widest interval, while Glasgow has the narrowest. Across all cities, predicted lower bounds are relatively similar and below \pounds25,000, whereas differences in upper bounds account for most of the variation in interval widths. We also note that, although the salary intervals reported in individual job postings are typically narrow, their levels vary substantially across postings even within the IT industry. The reported widths reflect both this dispersion and the regional conformalization adjustment.

Overall, these findings highlight the importance of accounting for interval-valued salary information when studying spatial variation in labour market opportunities. In particular, these patterns reinforce the importance of analysing the full conditional salary distribution rather than relying on point estimates, as the upper tail of the distribution varies markedly across locations. Our approach captures spatial variation in both the predicted lower and upper bounds, thereby providing richer insights than simple point estimates such as conditional means or medians, without imposing parametric assumptions or imputing a single salary within each reported range.

\begin{table}[htbp]
\centering\small
\setlength{\tabcolsep}{4pt}
\caption{Local conformal prediction intervals for IT salaries in major UK cities (90\% target)}
\label{tab:city_salary_bounds}
\begin{tabular}{lrrrrr}
\toprule
City & Latitude & Longitude & Lower (\pounds) & Upper (\pounds) & Width (\pounds) \\
\midrule
London & 51.50°N & 0.14°W & 22,961 & 187,842 & 164,881\\
Manchester & 53.48°N & 2.23°W & 17,000 & 156,000 & 139,000\\
Birmingham & 52.48°N & 1.89°W & 14,480 & 149,168 & 134,688\\
Edinburgh & 55.94°N & 3.33°W & 18,947 & 171,000 & 152,053\\
Glasgow & 55.86°N & 4.26°W & 21,053 & 136,800 & 115,747\\
Bristol & 51.46°N & 2.59°W & 23,717 & 157,858 & 134,141\\
\bottomrule
\end{tabular}
\par\smallskip
\begin{minipage}{0.98\textwidth}\footnotesize
\textit{Notes.} These are $\tilde C_{\mathrm{loc}}(x)$ with $M=1$ from the first prespecified 60/20/20 split (seed 19260817), transformed from log salaries to pounds. Conformalization uses the geographic cell containing each location. The local guarantee is averaged over that cell. Widths are differences of the rounded endpoints. Appendix~\ref{appendix:empirical-m2} reports 100-split coverage and volume summaries.
\end{minipage}
\end{table}

\begin{figure}[htbp]
    \centering
    \includegraphics[width=0.7\textwidth]{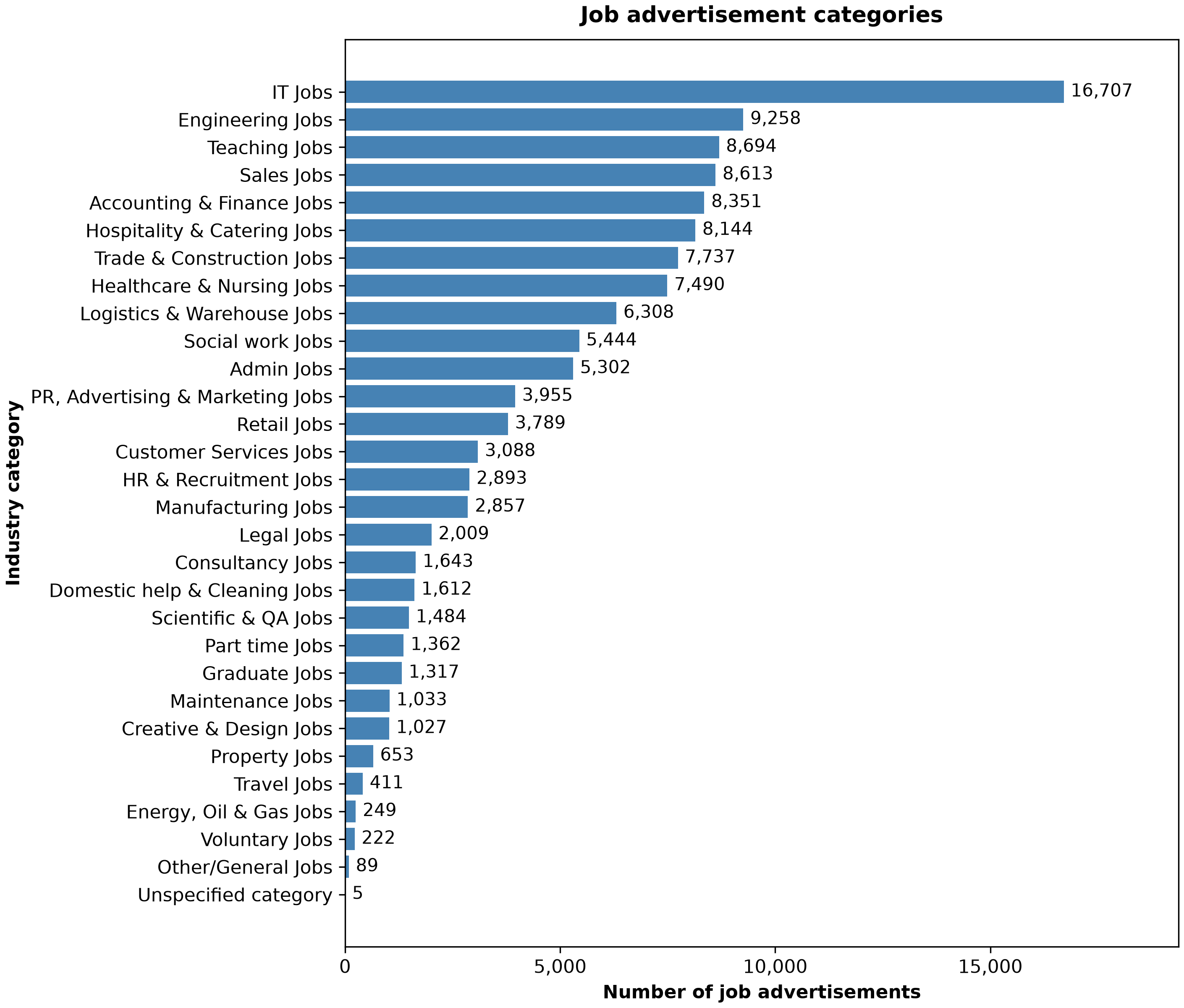}
    \caption{Job advert counts by category, January--March 2022 ($n=121{,}746$).}
    \label{fig:job distribution by category}
\end{figure}

\begin{figure}[htbp]
    \centering
    \begin{subfigure}[b]{0.48\textwidth}
        \centering
        \includegraphics[width=\linewidth]{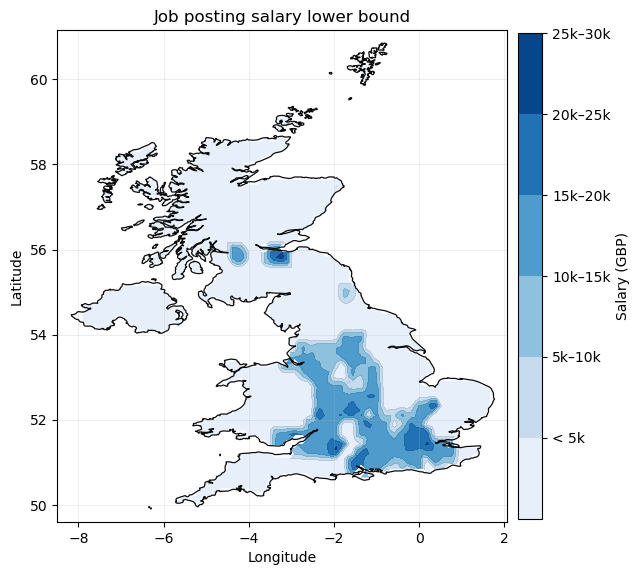}
        \caption{Lower bound}
    \end{subfigure}
    \hfill
    \begin{subfigure}[b]{0.48\textwidth}
        \centering
        \includegraphics[width=\linewidth]{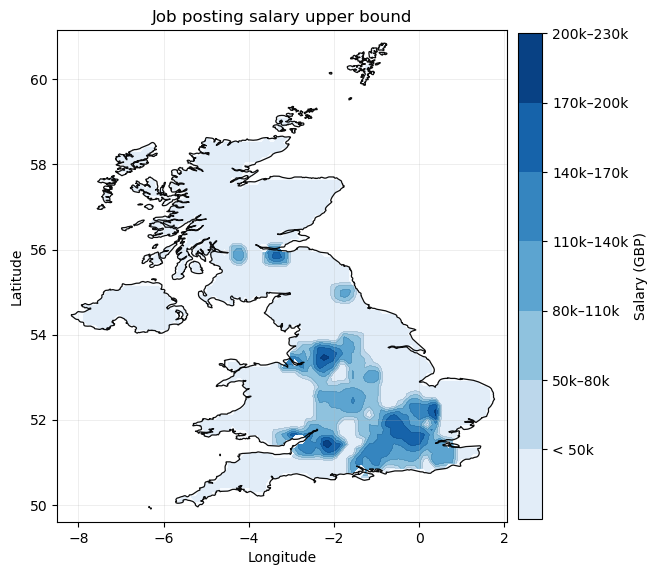}
        \caption{Upper bound}
    \end{subfigure}
    \caption{Estimated prediction interval $\hat C_I(x)$ for annual salaries in the UK IT job market at the 90\% target, before conformalization. The panels show the lower and upper endpoints, transformed back to pounds and interpolated for visualisation.}
    \label{fig:job-advert}
\end{figure}

\newpage

\subsection{US CPS data}

We apply our method to the 2023 US Current Population Survey (CPS) dataset and construct a prediction interval for income given years of education and experience. Following \cite{heckman2008EarningsFunctions}, we construct the following variables from the dataset for U.S.-born individuals with positive recorded earnings from their longest job:
\begin{itemize}
    \item \textit{Income:} We include all individuals who reported their income or income range in the survey. Since 2014, when a person refuses to report their income, the interviewer will ask for a range, and the person's income will be imputed by matching with similar people who have income in the same range. In the 2019 ASEC Updates\footnote{See \href{https://cps.ipums.org/cps/asec_2019_changes.shtml}{https://cps.ipums.org/cps/asec\_2019\_changes.shtml}}, it is stated that ``if respondents did not give a value for income received from a given source, the interviewer now follows up with a question about income ranges. Respondents who did not give a specific value were asked if they received over \$60,000, between \$45,000 and \$60,000, or less than \$45,000 from this income source. Those who replied that they earned less than \$45,000 from this source were further asked if they earned more than \$30,000, between \$15,000 and \$30,000, or less than \$15,000 from this income source.'' For respondents who report only income above $\$60,000$, we impose an upper bound of $\$400,000$ as a modelling choice for the empirical analysis. The individuals who did not report their income or their income range are excluded.

    \item \textit{Age:} We include individuals who are between 18 and 65 years old.
    \item \textit{Education:} Following \cite{heckman2008EarningsFunctions}, we construct the years of education for each individual by converting the highest qualification they have obtained into years of education.
    \item \textit{Experience:} We construct the years of experience for each individual as their age minus the years of education minus 6 (which is presumed to be the entry age into education).
\end{itemize}


The CPS dataset consists of 146,133 observations. There are 53,171 individuals who satisfy the birthplace, age and positive-earnings restrictions above. Among the individuals that satisfy these criteria, 39,333 reported their income either exactly or as a band and the response rate is thus around 74\%. The issue of nonresponse in surveys has been studied in the literature. See, for example, \cite{lillard1986WhatWe} and \cite{bollinger2019TroubleTails}, where they document an increasing proportion of nonresponse of earnings in survey data and the potential bias it introduces. Our prediction sets could inform nonresponse imputation, but this would require assumptions on the response mechanism and a separate target-population conformalization argument \citep{lei2021ConformalInference}. In fact, allowing respondents to report their income in bands could help mitigate the non-response problem.  For the 39,333 individuals who reported their incomes and satisfy these criteria, 8,819 of them reported either their gross or net income as a band, which thus amounts to roughly 22.4\% of the respondents.

Estimation and conformalization are performed on the log-earnings scale, and the reported bounds are transformed back to U.S. dollars.
Figure \ref{fig:histogram_education_range_indicator_cps} shows the proportion of individuals who reported their income in a range by years of education.
For estimation, we further restrict the sample to individuals with at least 10 years of education and between 0 and 50 years of experience, leaving 38,651 observations. We use a single prediction interval ($M=1$) in this application. Appendix~\ref{appendix:empirical-m2} assesses this choice through a paired comparison with $M=2$ in both empirical datasets. Appendix~\ref{appendix:cps-modality} provides supplementary income-distribution diagnostics.

\begin{figure}[htbp]
    \centering\includegraphics[width=0.75 \textwidth]{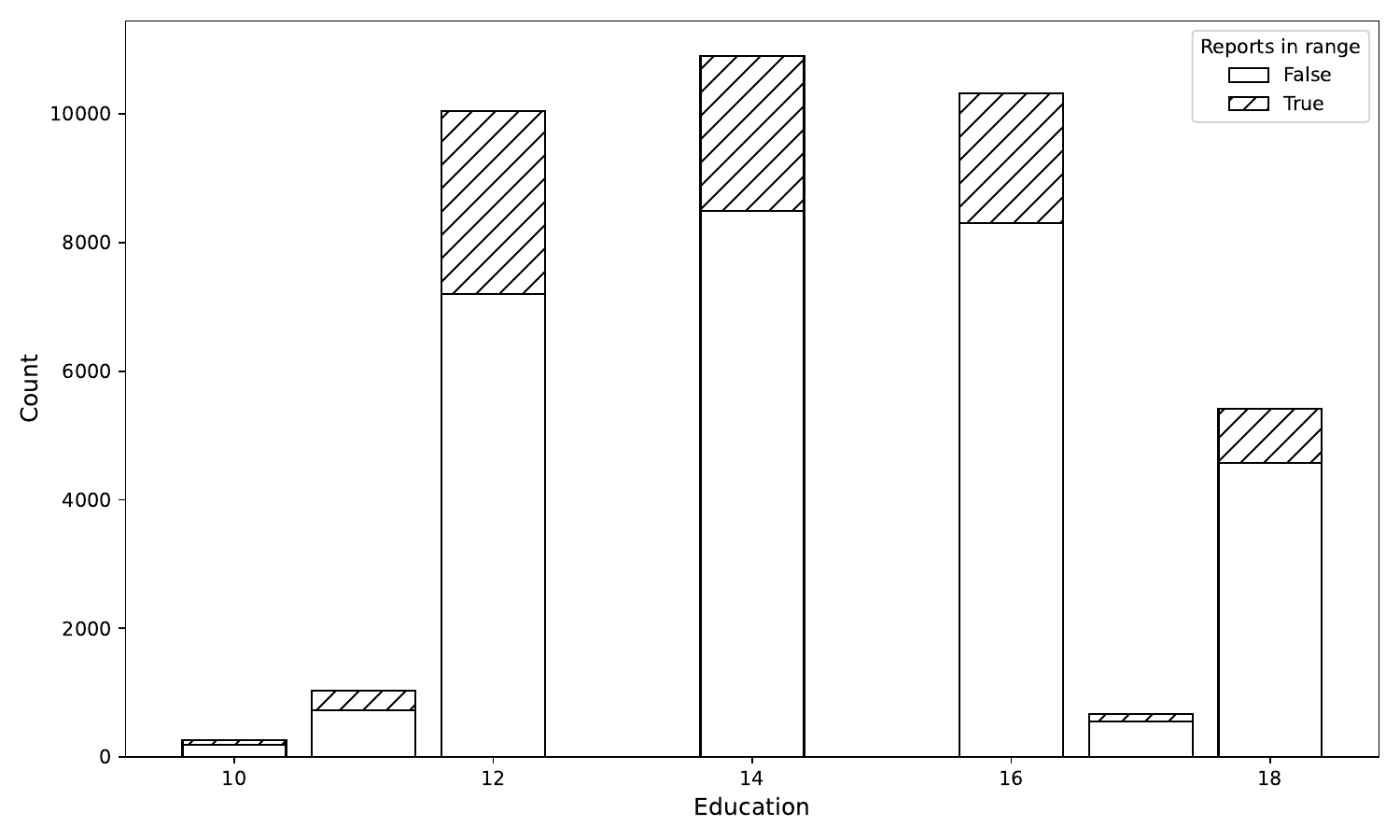}
    \caption{Proportion of individuals who reported income in a range by years of education.}
    \label{fig:histogram_education_range_indicator_cps}
\end{figure}

In our empirical study, we first randomly reserve 20\% of the data as a hold-out dataset. The remaining 80\% is then further divided into training and calibration sets, with 75\% of the data used for training to estimate the prediction interval \(\hat{C}_{I}\) and 25\% used for constructing the conformalized prediction interval \(\tilde{C}_{\mathrm{loc}}\). We use the exact $M=1$ solver and local conformalization with the CPS specification in Appendix~\ref{appendix:empirical-checks}. Table \ref{tab:CPS} presents the average estimated prediction intervals across the 100 iterations. The method column indicates the way we estimate the prediction interval. ``P'' stands for prediction interval $\hat{C}_{I}$ in (\ref{eq:CIhat}) and ``CP'' stands for conformal prediction interval $\tilde{C}_{\mathrm{loc}}$. If ``M'' is appended, it indicates the prediction intervals are estimated with imputed income for individuals who reported their income in an interval. It can be seen from the table that we would, in general, predict a higher income for individuals with more years of education.

The lower bounds of the prediction intervals generally increase with years of education, as do the 10\% prediction intervals shown in the columns with $\alpha=0.9$. This is consistent with what has been documented in the literature. Having more years of experience also affects predicted income, although the pattern is not monotone across the displayed groups. By varying the choice of $\alpha$, we can gain useful insights into the predicted income distribution.

The impact of including interval-valued income data on the prediction interval is more pronounced at smaller $\alpha$ values, as evidenced by the results for $\alpha=0.1$ in the table. This is because, with a 90\% coverage target, the tails of the income distribution become more critical. Given that 20\% of the observations are reported in ranges and many fall into the highest category, incorporating interval-valued data raises the upper endpoints of the prediction intervals. Consequently, the prediction intervals constructed with interval-valued data are significantly different from the ones constructed with imputed income data.

\autoref{tab:table2} shows coverage of the whole reported income brackets, including exact reports, in the hold-out dataset for the conformal prediction interval $\tilde C_{\mathrm{loc}}$ constructed using both interval-valued and exactly reported income data (shown in the “Censored” columns), as well as for the conformal prediction intervals constructed when the interval-valued incomes are imputed (shown in the “Imputed” columns). The significant difference in prediction intervals constructed with interval-valued data and with imputed data may be because there are only five categories for individuals who report their income in a range. This seems somewhat coarse in terms of informing about the true income distribution. The interval-based procedure has whole-bracket coverage close to the nominal rate, while the imputation-based procedure has lower whole-bracket coverage. Note that in Tables \ref{tab:CPS} and \ref{tab:table2}, the displayed miscoverage levels $\alpha$ correspond to prediction sets with nominal coverage $1-\alpha$, which can be informative about the shape of the data distribution. Notice that standard reporting of confidence intervals on parameters is meant to summarise sampling uncertainty, while here not only do these conformalized sets take into account sampling uncertainty, but they also estimate prediction sets of the underlying distributions of interest.

{\footnotesize\begin{table}[p]

    \renewcommand{\arraystretch}{1}
    \centering
    \caption{Prediction intervals for annual income}
    \label{tab:CPS}
    \scalebox{0.8}{\resizebox{\textwidth}{!}{%
            \begin{tabular}{lllllll}
                \toprule
                &  & Edu. & 12 & 14 & 16 & 18 \\
                $\alpha$ & Method & Exp. &  &  &  &  \\
                \midrule
                \multirow[t]{16}{*}{0.1} & \multirow[t]{4}{*}{CP} & 10 & (3,591 - 162,255) & (6,222 - 156,111) & (24,215 - 437,969) & (30,810 - 434,088) \\
                &  & 20 & (11,382 - 487,130) & (13,852 - 437,834) & (18,919 - 437,338) & (38,975 - 431,918) \\
                &  & 30 & (12,808 - 397,203) & (14,195 - 396,047) & (22,866 - 457,689) & (28,631 - 443,698) \\
                &  & 40 & (12,638 - 416,281) & (11,459 - 426,981) & (13,228 - 458,549) & (11,046 - 539,261) \\
                \cline{2-7}
                & \multirow[t]{4}{*}{CPM} & 10 & (8,846 - 101,514) & (11,180 - 130,487) & (18,018 - 196,222) & (24,631 - 252,037) \\
                &  & 20 & (10,135 - 138,579) & (12,548 - 160,643) & (18,594 - 272,613) & (34,740 - 309,918) \\
                &  & 30 & (13,324 - 131,822) & (13,905 - 187,287) & (19,630 - 309,910) & (25,336 - 362,964) \\
                &  & 40 & (10,276 - 137,885) & (9,668 - 172,354) & (11,652 - 318,112) & (10,164 - 460,960) \\
                \cline{2-7}
                & \multirow[t]{4}{*}{P} & 10 & (4,869 - 107,320) & (7,280 - 131,091) & (26,214 - 401,200) & (33,044 - 402,400) \\
                &  & 20 & (13,387 - 397,250) & (14,875 - 400,000) & (20,514 - 400,000) & (41,672 - 401,500) \\
                &  & 30 & (14,040 - 357,220) & (14,597 - 383,300) & (24,958 - 416,500) & (30,845 - 410,000) \\
                &  & 40 & (13,399 - 388,350) & (12,142 - 400,000) & (14,276 - 420,200) & (13,445 - 425,300) \\
                \cline{2-7}
                & \multirow[t]{4}{*}{PM} & 10 & (9,411 - 95,034) & (12,009 - 120,790) & (19,437 - 181,190) & (27,027 - 228,730) \\
                &  & 20 & (11,094 - 126,085) & (13,370 - 149,960) & (19,904 - 253,620) & (38,064 - 280,750) \\
                &  & 30 & (14,447 - 120,470) & (15,302 - 169,340) & (21,343 - 283,676) & (28,185 - 322,340) \\
                &  & 40 & (11,207 - 125,618) & (10,485 - 158,235) & (13,102 - 279,622) & (12,811 - 350,630) \\

                \cline{1-7} \cline{2-7}
                \multirow[t]{16}{*}{0.9} & \multirow[t]{4}{*}{CP} & 10 & (34,100 - 40,182) & (38,994 - 45,757) & (49,261 - 57,879) & (68,381 - 81,094) \\
                &  & 20 & (38,792 - 48,811) & (53,202 - 64,253) & (67,547 - 83,987) & (91,994 - 105,760) \\
                &  & 30 & (44,022 - 56,069) & (49,212 - 60,573) & (63,700 - 76,707) & (84,109 - 99,272) \\
                &  & 40 & (36,139 - 44,539) & (58,586 - 72,916) & (71,040 - 89,848) & (102,813 - 130,129) \\
                \cline{2-7}
                & \multirow[t]{4}{*}{CPM} & 10 & (33,209 - 37,569) & (46,128 - 53,213) & (56,511 - 63,490) & (67,903 - 78,975) \\
                &  & 20 & (48,180 - 56,155) & (52,082 - 58,743) & (76,779 - 92,844) & (90,568 - 100,698) \\
                &  & 30 & (47,660 - 55,441) & (50,164 - 57,012) & (62,232 - 71,079) & (84,263 - 96,570) \\
                &  & 40 & (43,644 - 51,977) & (53,544 - 63,578) & (65,762 - 80,697) & (89,383 - 107,653) \\
                \cline{2-7}
                & \multirow[t]{4}{*}{P} & 10 & (34,675 - 39,504) & (39,926 - 44,677) & (50,496 - 56,448) & (70,814 - 78,295) \\
                &  & 20 & (40,313 - 46,953) & (54,735 - 62,431) & (69,698 - 81,372) & (94,354 - 103,090) \\
                &  & 30 & (46,075 - 53,546) & (51,141 - 58,272) & (65,470 - 74,601) & (86,465 - 96,538) \\
                &  & 40 & (37,627 - 42,755) & (60,758 - 70,283) & (74,059 - 86,140) & (106,715 - 125,277) \\
                \cline{2-7}
                & \multirow[t]{4}{*}{PM} & 10 & (33,839 - 36,859) & (47,377 - 51,795) & (57,553 - 62,331) & (70,311 - 76,260) \\
                &  & 20 & (49,627 - 54,503) & (52,803 - 57,931) & (79,396 - 89,766) & (92,595 - 98,472) \\
                &  & 30 & (49,026 - 53,882) & (51,158 - 55,890) & (63,410 - 69,740) & (86,830 - 93,688) \\
                &  & 40 & (45,312 - 50,052) & (55,501 - 61,315) & (68,754 - 77,149) & (93,449 - 102,899) \\
                \cline{1-7} \cline{2-7}

                \multicolumn{7}{c}{\parbox{1.2\textwidth}{\vspace{0.3cm}
                \textbf{Note.} Annual incomes are measured in US dollars. Entries average centre-point endpoints over 100 splits; estimation and conformalization use log earnings and fixed education/experience conformalization groups. In the Method column, P stands for prediction intervals $\hat{C}_I$ at a given education level and years of experience. CP stands for the conformal prediction interval $\tilde{C}_{\mathrm{loc}}$. PM and CPM are prediction intervals and conformal prediction intervals, respectively, constructed with imputed income for individuals who report their income in a range.}}
            \end{tabular}
        }
    }
\end{table}
 } {\footnotesize
    \begin{table}
        \caption{Coverage of conformal prediction intervals} {\scriptsize\resizebox{\textwidth}{!}{%
    \renewcommand{\arraystretch}{1}
    \begin{tabular}{cccccccccc}
        \toprule
        & $\alpha$ & \multicolumn{2}{c}{0.10} & \multicolumn{2}{c}{0.25} & \multicolumn{2}{c}{0.50} & \multicolumn{2}{c}{0.90} \\
        & Data Type & Censored & Imputed & Censored & Imputed & Censored & Imputed & Censored & Imputed \\
        Education & Experience &  &  &  &  &  &  &  &  \\
        \midrule
        \multirow[t]{4}{*}{12} & 10 & \textbf{0.905} & 0.862 & \textbf{0.765} & 0.668 & \textbf{0.512} & 0.433 & \textbf{0.107} & 0.073 \\
        & 20 & \textbf{0.906} & 0.851 & \textbf{0.761} & 0.675 & \textbf{0.515} & 0.478 & \textbf{0.105} & 0.075 \\
        & 30 & \textbf{0.912} & 0.837 & \textbf{0.762} & 0.669 & \textbf{0.508} & 0.460 & \textbf{0.104} & 0.071 \\
        & 40 & \textbf{0.912} & 0.829 & \textbf{0.760} & 0.665 & \textbf{0.506} & 0.461 & \textbf{0.102} & 0.072 \\
        \cline{1-10}
        \multirow[t]{4}{*}{14} & 10 & \textbf{0.902} & 0.871 & \textbf{0.763} & 0.695 & \textbf{0.518} & 0.475 & \textbf{0.105} & 0.085 \\
        & 20 & \textbf{0.911} & 0.848 & \textbf{0.757} & 0.692 & \textbf{0.508} & 0.459 & \textbf{0.103} & 0.084 \\
        & 30 & \textbf{0.919} & 0.836 & \textbf{0.758} & 0.689 & \textbf{0.505} & 0.459 & \textbf{0.105} & 0.076 \\
        & 40 & \textbf{0.908} & 0.834 & \textbf{0.754} & 0.680 & \textbf{0.512} & 0.458 & \textbf{0.107} & 0.079 \\
        \cline{1-10}
        \multirow[t]{4}{*}{16} & 10 & \textbf{0.907} & 0.836 & \textbf{0.756} & 0.696 & \textbf{0.509} & 0.460 & \textbf{0.104} & 0.082 \\
        & 20 & \textbf{0.909} & 0.806 & \textbf{0.755} & 0.666 & \textbf{0.511} & 0.429 & \textbf{0.101} & 0.083 \\
        & 30 & \textbf{0.906} & 0.774 & \textbf{0.766} & 0.638 & \textbf{0.504} & 0.405 & \textbf{0.102} & 0.082 \\
        & 40 & \textbf{0.914} & 0.802 & \textbf{0.774} & 0.641 & \textbf{0.511} & 0.412 & \textbf{0.103} & 0.079 \\
        \cline{1-10}
        \multirow[t]{4}{*}{18} & 10 & \textbf{0.912} & 0.827 & \textbf{0.758} & 0.685 & \textbf{0.506} & 0.445 & \textbf{0.101} & 0.089 \\
        & 20 & \textbf{0.908} & 0.804 & \textbf{0.754} & 0.661 & \textbf{0.504} & 0.428 & \textbf{0.102} & 0.089 \\
        & 30 & \textbf{0.912} & 0.819 & \textbf{0.764} & 0.650 & \textbf{0.508} & 0.427 & \textbf{0.105} & 0.085 \\
        & 40 & \textbf{0.908} & 0.827 & \textbf{0.761} & 0.615 & \textbf{0.499} & 0.393 & \textbf{0.104} & 0.085 \\
        \cline{1-10}
        \bottomrule

        \multicolumn{10}{l}{\parbox{\textwidth}{\vspace{0.3cm}
                \textbf{Note.} Entries average held-out whole-bracket coverage over 100 splits. Prediction sets $\tilde C_{\mathrm{loc}}(X_i)$ are evaluated at each individual\textquotesingle s covariates. Groups fix education and use experience windows $[e-5,e+5)$ around the displayed centres. Censored and Imputed refer to fitting and conformalization with reported brackets and imputed incomes, respectively; both are evaluated against reported brackets. Bold values are closest to $1-\alpha$.
        }}
    \end{tabular}
}
}\label{tab:table2}
\end{table}}

\section{Conclusion}\label{sec:conclusion}

The object of interest in this paper is a prediction set for an outcome given a set of covariates when this outcome is censored. Censoring is a common issue in economic and other data. We first characterise the oracle prediction set, which is the smallest set that maintains a given miscoverage level under censoring. The characterisation leads to a feasible estimation strategy based on the observed sample. We then provide consistent estimators for this oracle prediction set based on nonparametric estimation of the conditional distribution of the observed lower and upper brackets given the conditioning variables. We allow for the prediction set to consist of multiple intervals and consider both random and fixed censoring. Furthermore, we use recent results from conformal inference to obtain a conformal prediction set that maintains a finite sample miscoverage property using a set-based score function. We find that our procedures perform well with simulated data. We also apply our procedures to UK job adverts and the US Current Population Survey. There is increasing interest in prediction set estimation, and several intriguing research directions are related to this paper. One promising avenue is to consider prediction sets under partial identification beyond interval censoring. While conformal inference offers finite-sample guarantees on coverage, it remains a research challenge to explore the more traditional statistical inference on the estimated prediction set. Additionally, the design of conformity scores in the context of partial identification presents an interesting and important topic for further investigation.
Another promising direction is prediction of individual treatment effects when outcomes are censored or interval-valued. This would connect our framework with the literature on conformal inference for counterfactuals and individual treatment effects \citep{lei2021ConformalInference}, while accounting for the additional uncertainty induced by censoring.

\textbf{Declaration of generative AI and AI-assisted technologies in the manuscript preparation process}: In preparation for the manuscript, ChatGPT was used to format the figures and tables in the LaTeX document and to provide spelling and grammar checks. GitHub Copilot and Codex provided code suggestions and auto-completions for the Python code implementing the simulation and empirical study. Refine.ink was also used for proofreading. After using these tools, the authors reviewed and edited the content as needed and take full responsibility for the content of the published article.

\bibliographystyle{apalike}
\bibliography{lib.bib,additional.bib}

\newpage
\appendix
\begin{center}
    \textbf{\Large Supplementary Material for ``Prediction Sets and Conformal Inference with Interval Outcomes''}
\end{center}
\section{Proofs}\label{appendix: proofs}
\subsection{Proofs for Section \ref{sec:setup}}

In this subsection, we prove the results in Section \ref{sec:setup} and clarify the difference between our approach and the classical level set estimation.

First we state a simple and useful lemma that allows us to construct a density function that is in the identified set under interval censoring given the joint density of $(Y^{L}, Y^{U})$.
\begin{lemma}\label{lemma: tool for construct Y density}

    Let $f_{Y^L,Y^U}$ be the joint density of $(Y^L,Y^U)$. For any Borel set $A\subseteq\mathbb R$ and $\theta\in(\pi/2,\pi]$, define
    \begin{equation*}
        T_\theta(\ell,u)=\frac{u-\ell \tan\theta}{1-\tan\theta},
        \qquad
        A_\theta=\{(\ell,u):\ell\leq u,\ T_\theta(\ell,u)\in A\}.
    \end{equation*}
    For $\theta=\pi/2$, set $T_{\pi/2}(\ell,u)=\ell$ and $A_{\pi/2}=\{(\ell,u):\ell\leq u,\ \ell\in A\}$. Then $P_{\theta}(A) = P_{Y^{L}, Y^{U}}(A_{\theta})$ defines a probability measure $P_{\theta}$ on $\mathbb{R}$, and $P_{\theta}$ is in the identified set of the distribution of the latent outcome $Y$.
\end{lemma}

\begin{proof}
    For $\theta\in[\pi/2,\pi]$, the definition gives that $T_\theta(\ell,u) \in[\ell,u].$ Thus $P_\theta$ is the push-forward distribution of a measurable selection from the random interval $[Y^L,Y^U]$ \citep{molchanov2018RandomSets}. Its law is compatible with the observed interval under the selection mechanism $Y=T_\theta(Y^L,Y^U)$, and hence $P_\theta\in\mathcal P_I$.
\end{proof}

With Lemma \ref{lemma: tool for construct Y density} at hand, we can now illustrate why the level set of $Y$ is not in general identified given the joint density of $(Y^{L}, Y^{U})$. Given the density $f_{Y}$ of a random variable $Y$, the (upper) level set of $Y$ at level $\lambda$ is defined as $L(Y; \lambda) = \cqty{y: f_{Y}(y) \geq \lambda}$. When the random variable $Y$ is interval censored, we instead observe the random vector $(Y^{L}, Y^{U})$ and we could define the level set of $(Y^{L},Y^{U})$ as
\begin{equation}
    L((Y^{L},Y^{U}); \lambda) = \cqty{(y^{L}, y^{U}): f_{Y^{L}, Y^{U}}(y^{L}, y^{U}) \geq \lambda}.
\end{equation}
We note here two points. First, the level set of $Y$ is in general not identified given the joint density of $(Y^{L}, Y^{U})$ and second, in terms of prediction for $Y$, our proposed prediction set is more efficient and has clearer interpretation than a level set of $(Y^{L}, Y^{U})$.

The following example shows that the level set of $Y$ is not in general identified without additional assumptions on the mechanism of censoring.
\begin{example}\label{example: level set not identified}
    Let $a_0<a_1<b_0<b_1$, and suppose $Y^L$ and $Y^U$ are independent uniform random variables on $[a_0,a_1]$ and $[b_0,b_1]$, respectively. Choose
    \[
        0<\lambda\leq\min\left\{\frac{1}{a_1-a_0},\frac{1}{b_1-b_0}\right\}.
    \]
    The compatible selections $Y_1=Y^L$ and $Y_2=Y^U$ have densities
    \[
        f_{Y_1}(y)=\frac{\indicator{y\in[a_0,a_1]}}{a_1-a_0},
        \qquad
        f_{Y_2}(y)=\frac{\indicator{y\in[b_0,b_1]}}{b_1-b_0}.
    \]
    Hence $L(Y_1;\lambda)=[a_0,a_1]$ and $L(Y_2;\lambda)=[b_0,b_1]$ are disjoint. Both latent distributions are nonatomic and compatible with the same joint distribution of the observed endpoints, so the latent density level set is not identified.
\end{example}

We now prove \autoref{prop: level set not identified} which generalises the idea in the previous example and shows the failure of identification of the level set of $Y$.

\begin{proof}[Proof of Proposition \ref{prop: level set not identified}]
    If $\lambda\leq0$, take the compatible selection $Y^\prime=Y^L$ and a nonnegative version of its density. Its level set at $\lambda$ is all of $\mathbb R$. Hence assume $\lambda>0$. Let $(a,b)$ be a Lebesgue point of $f_{Y^L,Y^U}$ such that $f_{Y^L,Y^U}(a,b)>0$. Consider the sets
    \begin{equation*}
        A_k=\{(y^L,y^U):a-k^{-1}\leq y^L\leq a,
        \ b\leq y^U\leq b+k^{-1}\},
    \end{equation*}
    which shrink nicely to $(a,b)$. By the Lebesgue differentiation theorem \citep[see][Theorem~3.21]{folland1999RealAnalysis}, for all sufficiently large $k$,
    \begin{equation*}
        \delta=P((Y^L,Y^U)\in A_k)
        =\int_{A_k}f_{Y^L,Y^U}(y^L,y^U)\dd y^L\dd y^U>0.
    \end{equation*}
    Fix such a $k$, any $y\in[a,b]$, and $0<\epsilon<\delta/\lambda$. Let
    \begin{equation*}
        I_{y,\epsilon}
        =[y-\epsilon/2,y+\epsilon/2]\cap[a,b],
        \qquad m_{y,\epsilon}=\mu(I_{y,\epsilon})>0.
    \end{equation*}
    Let $V$ be independent of $(Y^L,Y^U)$ and uniform on $I_{y,\epsilon}$, and define $Y'=V$ when $(Y^L,Y^U)\in A_k$ and $Y'=Y^L$ otherwise. Since every $[y^L,y^U]$ with $(y^L,y^U)\in A_k$ contains $[a,b]$, we have $Y^L\leq Y'\leq Y^U$ almost surely. Hence the law $P'$ of $Y'$ belongs to $\mathcal P_I$ and has density
    \begin{equation*}
        f_{Y'}(z)
        =\int f_{Y^L,Y^U}(z,u)
        \mathbb{I}\{(z,u)\notin A_k\}\dd u
        +\frac{\delta}{m_{y,\epsilon}}
        \mathbb{I}\{z\in I_{y,\epsilon}\}.
    \end{equation*}
    Since $m_{y,\epsilon}\leq\epsilon$ and $y\in I_{y,\epsilon}$, $f_{Y'}(y)\geq\delta/m_{y,\epsilon}\geq\delta/\epsilon>\lambda$. Thus $y\in L'(\lambda)$.
\end{proof}

Since $P(Y\in A) \geq P([Y^{L},Y^U] \subseteq A)$, it is clear that the outer region of the level set with $1-\alpha$ coverage of $(Y^{L}, Y^{U})$ also has $1-\alpha$ coverage for $Y$. But the outer region is in general not the smallest valid prediction set.

\begin{proof}[Proof of Lemma \ref{lemma: identified set}]
    Since the observed endpoints are finite, $S=[Y^L,Y^U]$ is a measurable nonempty compact random interval. Necessity follows from $\{S\subseteq I\}\subseteq\{Y\in I\}$ for each interval $I$.

    For sufficiency, let $Q_x$ be a probability kernel satisfying the inequalities in \eqref{eq:identified set} for $P_X$-almost every $x$. For any finite disjoint union $A=\bigsqcup_{j=1}^m I_j$ of bounded closed intervals, connectedness of $S$ gives
    \[
        Q_x(A)=\sum_{j=1}^m Q_x(I_j)
        \geq\sum_{j=1}^m P(S\subseteq I_j\mid X=x)
        =P(S\subseteq A\mid X=x).
    \]
    For any nonempty compact $K\subseteq\mathbb R$, the sets $K_n=\{y:\inf_{z\in K}|y-z|\leq1/n\}$ are finite unions of closed intervals and decrease to $K$. Continuity from above therefore yields
    \[
        Q_x(K)\geq P(S\subseteq K\mid X=x).
    \]
    The inequality also holds for $K=\varnothing$.

    Let $Q(\dd x,\dd y)=P_X(\dd x)Q_x(\dd y)$ and $\Gamma=\{X\}\times S$. For every compact $B\subseteq\mathbb R^{d+1}$, its sections $B_x=\{y:(x,y)\in B\}$ are compact, so
    \[
        Q(B)=\int Q_x(B_x)\,P_X(\dd x)
        \geq\int P(S\subseteq B_x\mid X=x)\,P_X(\dd x)
        =P(\Gamma\subseteq B).
    \]
    Since $\Gamma$ is a measurable compact random set, Artstein's selection theorem \citep[Theorem~A.1]{molinari2020Chapter5} gives a coupling with a selection $(X,Y')\in\Gamma$ having law $Q$. As $\Gamma$ determines $(X,Y^L,Y^U)$, this coupling preserves the observed law and satisfies $Y^L\leq Y'\leq Y^U$, proving sharpness.
\end{proof}

The following lemma shows that even though our approach is conservative in terms of coverage, its conservativeness is only due to partial identification, which is necessary without assuming the mechanism of censoring.

\begin{proof}[Proof of Lemma \ref{lemma: inf P(y) = P(yl,yu)}]
    It is clear that $P(Y\in C\mid X=x)\geq P([Y^L,Y^U]\subseteq C\mid X=x)$ for all $P\in\mathcal P_I$, and we only need to construct a compatible law attaining equality. Let
    \begin{equation*}
        g_x=\frac{\dd P_{Y^L,Y^U\mid X=x}}{\dd\mu_x}
    \end{equation*}
    be the Radon--Nikodym derivative with respect to a sigma-finite dominating measure $\mu_x$ on $\mathbb R^2$. Assume without loss of generality that
    \begin{equation*}
        -\infty=u_0<l_1<u_1<\cdots<l_M<u_M<l_{M+1}=\infty.
    \end{equation*}
    The half space $\mathbb R^{2,+}=\{(y^L,y^U):y^U\geq y^L\}$ is partitioned into the following disjoint regions:
    \begin{align*}
        A_m&=\{(y^L,y^U):l_m\leq y^L\leq y^U\leq u_m\},\\
        A'_m&=\{(y^L,y^U):l_m\leq y^L\leq u_m<y^U<l_{m+1}\},\\
        A''_m&=\{(y^L,y^U):l_m\leq y^L\leq u_m,
        \ y^U\geq l_{m+1}\},\\
        B_m&=\{(y^L,y^U):u_m<y^L<l_{m+1},\ y^L\leq y^U\},
    \end{align*}
    where $A_m$ and $A'_m$ are defined for $m=1,\ldots,M$, $A''_m$ for $m=1,\ldots,M-1$, and $B_m$ for $m=0,\ldots,M$. Recalling the definition of $A_\theta$ in \autoref{lemma: tool for construct Y density}, define the following positive measures on the real line, for any Borel set $A\subseteq\mathbb R$:
    \begin{align*}
        P_1(A)
        &=\int_{A_{\pi/2}}g_x(y^L,y^U)
        \mathbb{I}\left\{\cup_m(A_m\cup B_m)\right\}
        \dd\mu_x(y^L,y^U),\\
        P_2(A)
        &=\int_{A_\pi}g_x(y^L,y^U)
        \mathbb{I}\left\{\cup_m A'_m\right\}
        \dd\mu_x(y^L,y^U),\\
        P_{3,m}(A)
        &=P((Y^L,Y^U)\in A''_m\mid X=x)\,
        U(u_m,l_{m+1})(A),\qquad m=1,\ldots,M-1,
    \end{align*}
    where $U(a,b)$ is the uniform distribution over $(a,b)$. Then
    \begin{equation*}
        P'(A)=P_1(A)+P_2(A)+\sum_{m=1}^{M-1}P_{3,m}(A)
    \end{equation*}
    defines a probability measure. It is the conditional law of a compatible selection: select $Y^L$ on $\bigcup_m(A_m\cup B_m)$, select $Y^U$ on $\bigcup_m A'_m$, and draw uniformly from $(u_m,l_{m+1})$ on $A''_m$ using auxiliary independent randomness. Moreover, this selection belongs to $C$ if and only if $[Y^L,Y^U]\subseteq C$. Applying the same construction at every $x$ yields a compatible joint law in $\mathcal P_I$, and
    \begin{equation*}
        P'(Y\in C\mid X=x)
        =P([Y^L,Y^U]\subseteq C\mid X=x),
    \end{equation*}
    which proves the claim.

    Alternatively, let $d(z,C)=\inf_{c\in C}\abs{z-c}$ and define
    \begin{equation*}
        s_C(y_l,y_u)
        =\min\arg\max_{z\in[y_l,y_u]}d(z,C).
    \end{equation*}
    Since $C$ is closed, $d(\cdot,C)$ is continuous. By the measurable maximum theorem \citep[Theorem~18.19]{aliprantis2006InfiniteDimensional}, the argmax correspondence is measurable and compact-valued. Hence its smallest element is measurable and belongs to $[y_l,y_u]$. Moreover,
    \begin{equation*}
        s_C(y_l,y_u)\in C
        \quad\Longleftrightarrow\quad
        [y_l,y_u]\subset C.
    \end{equation*}
    Indeed, if the interval is contained in $C$, its maximum distance from $C$ is zero; otherwise the maximum distance is strictly positive and every maximizer lies outside $C$. Thus $Y'=s_C(Y^L,Y^U)$ is a compatible measurable selection whose joint law with $X$ belongs to $\mathcal P_I$ and satisfies
    \begin{equation*}
        P(Y'\in C\mid X=x)
        =P([Y^L,Y^U]\subset C\mid X=x).
    \end{equation*}
\end{proof}

\subsection{Proofs for Section \ref{sec:consistency}}

In this subsection, we first prove the uniform consistency of the proposed prediction set when the oracle prediction set is an interval under random censoring. Then the result is extended to the case allowing for fixed censoring and when the oracle prediction set is a union of multiple intervals.

When the oracle prediction set is a single interval for each \(x\in \mathcal{X}\), let \(C_{I}^{*}(x)=[\tau_{0}(x), \tau_{1}(x)]\).
\begin{proof}[Proof of \autoref{thm:consistency}]
    For any \(\epsilon>0\), let \(\delta\) be the margin in Assumption~\ref{asmp:identification} at distance \(\epsilon/2\).

    Let \(\nu\) be chosen such that \(0 < \nu < \min\pqty{\frac{\epsilon}{4}, {\frac{\epsilon_0}{2},} \frac{1}{2}\pqty{\frac{\delta}{2c_{2}}}^{\frac{1}{\gamma_{2}}}}\), and fix \(0 < \kappa < \min\pqty{\frac{\delta}{2}, c_{1}\nu^{\gamma_{1}}}\). Let \(E\) denote the event \(\sup_{x\in\mathcal X_n}\sup_{a\leq b}\abs{P_n(a,b;x)-P(a,b;x)}<\kappa.\) Under Assumption \ref{asmp:estimation}, \(P(E^{c}) \to 0\) as \(n\to \infty\). On the event \(E\), for any \(x\in\mathcal X_n\), we want to show
    \begin{equation*}
        \mu\pqty{[\hat\tau_0(x),\hat\tau_1(x)]\mathbin{\bigtriangleup}
        [\tau_0(x),\tau_1(x)]}\leq\epsilon.
    \end{equation*}
    Let
    \begin{equation*}
        E_{0,x}=\left\{\omega:
            \mu\pqty{[\hat\tau_0(x),\hat\tau_1(x)]\mathbin{\bigtriangleup}
        [\tau_0(x),\tau_1(x)]}>\epsilon\right\}.
    \end{equation*}
    Consider the following three cases,
    \begin{enumerate}
        \item \(E_{1,x} = \cqty{\omega: \hat{\tau}_{1}(x) - \hat{\tau}_{0}(x)\leq \tau_{1}(x) - \tau_{0}(x)}\)
        \item \(E_{2,x} = \cqty{\omega: \hat{\tau}_{1}(x)- \hat{\tau}_{0}(x) > \tau_{1}(x) - \tau_{0}(x)+ 2\nu}\)
        \item \(E_{3,x} = \cqty{\omega: \tau_{1}(x) - \tau_{0}(x) < \hat{\tau}_{1}(x) - \hat{\tau}_{0}(x) \leq \tau_{1}(x) - \tau_{0}(x) + 2\nu}\).
    \end{enumerate}
    We will show \(E_{0,x}\subseteq E^{c}\) by showing that \((E_{1,x} \cup E_{2,x}\cup E_{3,x})\cap E_{0,x} \subseteq E^{c}\) for {all \(x\in\mathcal X_n\)}.

    \begin{enumerate}

        \item In the first case, we have \(\hat{\tau}_{1}-\hat{\tau}_{0} \leq \tau_{1}(x) - \tau_{0}(x)\). If \(E_{0,x}\) also holds, then \(P (\hat{\tau}_{0},\hat{\tau}_{1}; x) < 1 - \alpha- \delta\) by Assumption \ref{asmp:identification}, while \(P_{n}(\hat{\tau}_{0},\hat{\tau}_{1}; x) \geq 1 - \alpha\), which implies that
            \begin{equation*}
                \abs{P_{n}(\hat{\tau}_{0}, \hat{\tau}_{1};x) - P\pqty{\hat{\tau}_{0}, \hat{\tau}_{1};x}} > \delta > \kappa.
            \end{equation*}
            Hence \(E_{1,x}\cap E_{0,x}\subseteq E^{c}\).
        \item In the second case, \(\hat{\tau}_{1}-\hat{\tau}_{0} > \tau_{1}(x) - \tau_{0}(x) + 2\nu\). Let \((\bar{\tau}_{0}, \bar{\tau}_{1}) = (\tau_{0}- \nu, \tau_{1}+ \nu)\). Then by Assumption \ref{asmp:regularity}, \(P(\bar{\tau}_{0}, \bar{\tau}_{1};x)\geq 1 - \alpha + c_{1} \nu^{\gamma_1}\). Since \(\bar{\tau}_{1}-\bar{\tau}_{0} < \hat{\tau}_{1}- \hat{\tau}_{0}\), the optimality of \([\hat\tau_0,\hat\tau_1]\) implies \(P_{n}\pqty{\bar{\tau}_{0}, \bar{\tau}_{1}; x} < 1 - \alpha\). Again, this implies
            \begin{equation*}
                \abs{P_n(\bar\tau_0(x),\bar\tau_1(x);x)
                -P(\bar\tau_0(x),\bar\tau_1(x);x)}
                >c_1\nu^{\gamma_1}>\kappa,
            \end{equation*}
            and hence \(E_{2,x}\subseteq E^{c}\).
        \item
            Let $\Delta=\hat\tau_1(x)-\hat\tau_0(x) -\{\tau_1(x)-\tau_0(x)\}\in(0,2\nu]$ and set $[a_0,a_1]=[\hat\tau_0(x)+\Delta,\hat\tau_1(x)]$. Then $a_1-a_0=\tau_1(x)-\tau_0(x)$ and $\mu([a_0,a_1]\mathbin{\bigtriangleup} [\hat\tau_0(x),\hat\tau_1(x)])=\Delta$. Hence, on $E_{0,x}$, the reverse triangle inequality gives $\mu([a_0,a_1]\mathbin{\bigtriangleup} [\tau_0(x),\tau_1(x)])>\epsilon-\Delta>\epsilon/2$. Assumption~\ref{asmp:identification}, applied at $\epsilon/2$, gives $P(a_0,a_1;x)<1-\alpha-\delta$. Since $[\hat\tau_0(x),\hat\tau_1(x)] \subseteq[a_0-\Delta,a_1+\Delta]$, monotonicity and Assumption~\ref{asmp:regularity} give $P(\hat\tau_0, \hat\tau_1;x)\leq P(a_0-\Delta,a_1+\Delta;x) \leq P(a_0,a_1;x)+c_2\Delta^{\gamma_2} <1-\alpha-\delta/2$. Empirical feasibility therefore implies an estimation error larger than $\delta/2>\kappa$, so $E_{3,x}\cap E_{0,x}\subseteq E^c$.
    \end{enumerate}
    As a result, on $E$, \(\sup_{x\in\mathcal X_n}\mu\pqty{[\hat\tau_0(x),\hat\tau_1(x)]\mathbin{\bigtriangleup}[\tau_0(x),\tau_1(x)]}\leq\epsilon.\) Since $P(E^c)\to0$, the stated $o_p(1)$ conclusion follows.
\end{proof}

\begin{proof}[Proof of \autoref{thm:consistency fixed}]
    Fix $\epsilon>0$, and let $\delta$ be the margin in Assumption~\ref{asmp:identification} at distance $\epsilon$. Define
    \begin{equation*}
        E_n=\left\{\sup_{x\in\mathcal X_n}\sup_{t_0\leq t_1}
        \abs{P_n(t_0,t_1;x)-P(t_0,t_1;x)}\leq\psi_n\right\}.
    \end{equation*}
    Assumption~\ref{asmp:Approximate solution fixed} gives $P(E_n)\to1$. On $E_n$, the oracle interval is empirically feasible because
    \begin{equation*}
        P_n(C_I^*(x);x)\geq P(C_I^*(x);x)-\psi_n
        \geq1-\alpha-\psi_n.
    \end{equation*}
    Hence optimality gives $\mu(\hat C_I(x))\leq\mu(C_I^*(x))$. If also $\mu(\hat C_I(x)\mathbin{\bigtriangleup}C_I^*(x))>\epsilon$, then Assumption~\ref{asmp:identification} gives $P(\hat C_I(x);x)<1-\alpha-\delta$. Empirical feasibility therefore implies
    \begin{equation*}
        P_n(\hat C_I(x);x)-P(\hat C_I(x);x)
        >\delta-\psi_n>\psi_n
    \end{equation*}
    for all sufficiently large $n$, contradicting $E_n$. This holds simultaneously for every $x\in\mathcal X_n$, and the result follows.
\end{proof}

\begin{proof}[Proof of \autoref{thm:consistency multiple}]
    For any $\epsilon > 0$, let $\delta$ be chosen as in Assumption \ref{asmp:identification multiple}. Let $E$ denote the event when $\sup_{C\in\overline{\mathcal C}_M}\sup_{x\in\mathcal X_n} \abs{P_n(C;x)-P(C;x)}\leq\psi_n$. For sufficiently large $n$, $\psi_n<\frac{\delta}{2}$. For any $x\in\mathcal X_n$, if $\mu(\hat C(x)\bigtriangleup C^*(x))>\epsilon$ and $\mu(\hat C(x))\leq\mu(C^*(x))$, then $P(\hat C(x);x)\leq1-\alpha-\delta$ and $P_n(\hat C(x);x)\geq1-\alpha-\psi_n$. Hence $\abs{P_n(\hat C(x);x)-P(\hat C(x);x)}>\frac{\delta}{2}>\psi_n$. On the other hand, if $\mu(\hat C(x)\bigtriangleup C^*(x))>\epsilon$ and $\mu(\hat C(x))>\mu(C^*(x))$, then since $C^*(x)\in\mathcal C_{\bar M}\subseteq\overline{\mathcal C}_M$, the optimality of $\hat C(x)$ implies that $C^*(x)$ is not empirically feasible, so $P_n(C^*(x);x)<1-\alpha-\psi_n$, while $P(C^*(x);x)\geq1-\alpha$, and therefore $\abs{P_n(C^*(x);x)-P(C^*(x);x)}>\psi_n$. That is, $P\pqty{\sup_{x\in\mathcal X_n} \mu(\hat C(x)\bigtriangleup C^*(x))>\epsilon} \leq P(E^c)\to0$.
\end{proof}

\subsection{Proofs for results in Section \ref{sec:conformal}}
\begin{proof}[Proof of Theorem \ref{thm:marginal}]

    Let $\mathcal T$ be the $\sigma$-field generated by the sample split and the training observations. Conditional on $\mathcal T$, the scores $\{s_i=s(Y_i^L,Y_i^U,X_i;\hat C): i\in\mathcal I_2\cup\{n+1\}\}$ are exchangeable. Lemma~2 of \citet{romano2019ConformalizedQuantile}, applied conditionally, therefore gives $P(s_{n+1}\leq\vartheta_{1-\alpha}\mid\mathcal T)\geq1-\alpha$.  By \eqref{eq:score event equivalence}, this event is equivalent to $[Y_{n+1}^L,Y_{n+1}^U]\subseteq\tilde C(X_{n+1})$. Interval censoring then implies $Y_{n+1}\in\tilde C(X_{n+1})$. Integrating over $\mathcal T$ proves the result.
\end{proof}
\begin{proof}[Proof of Theorem \ref{thm:consistency of conformal}]
    Write $\hat C_{n_1}$ for the estimator $\hat C$ fitted on the training sample of size $n_1$. Let $\mathcal T_{n_1}$ be the $\sigma$-field generated by the sample split and the observations in $\mathcal I_1$.  For a fresh observation $(X,Y^L,Y^U)$ independent of $\mathcal T_{n_1}$, write \(\hat S_{n_1}=s(Y^L,Y^U,X;\hat C_{n_1})\), \(S^*=s(Y^L,Y^U,X;C^*)\), and \(F_{n_1}(u)=P(\hat S_{n_1}\leq u\mid\mathcal T_{n_1}).\) Define \( D_{n_1}=\sup_{x\in\mathcal X_{n_1}} \mu\pqty{\hat C_{n_1}(x)\bigtriangleup C^*(x)}.\) By \autoref{thm:consistency multiple}, $D_{n_1}=o_p(1)$.

    Fix $\epsilon>0$ and choose $0<t<\min\{\epsilon/2,\eta_0\}$.  By \eqref{eq:right crossing oracle score}, $\Delta_t:=P(S^*\leq t)-(1-\alpha)>0$. Write the maximal-component representation of the oracle as $C^*(x)=\bigsqcup_{m=1}^{M_*(x)}[a_m(x),b_m(x)]$, and let $\ell_*(x)$ denote the length of its shortest component. Set $\ell_*(x)=0$ on the exceptional null set where this representation does not hold. Since $C^*(X)\in\mathcal C_{\bar M}$ almost surely, $\ell_*(X)>0$ almost surely.  We can therefore choose $r>0$ such that $r+t<\epsilon$ and $P_X\{\ell_*(X)\leq2r\}<\Delta_t/4$. Define $G_{n_1,r}=\{x\in\mathcal X_{n_1}:\ell_*(x)>2r\}$ and $A_{n_1,\epsilon}=\{D_{n_1}<\min\{r,\epsilon/2\}\}$. The set $G_{n_1,r}$ is deterministic.  Moreover, $P(A_{n_1,\epsilon})\to1$, and, for all sufficiently large $n_1$,
    \begin{equation*}
        P_X(X\notin G_{n_1,r})
        \leq P_X(X\notin\mathcal X_{n_1})
        +P_X\{\ell_*(X)\leq2r\}
        <\Delta_t/2.
    \end{equation*}

    On $A_{n_1,\epsilon}$, for every $x\in G_{n_1,r}$, $C^*(x)\subseteq\mathcal E(\hat C_{n_1}(x),r)$. Indeed, otherwise there is a point $y\in C^*(x)\setminus\mathcal E(\hat C_{n_1}(x),r)$.  If $[a,b]$ is the oracle component containing $y$, then $[a,b]\cap[y-r,y+r]$ has measure at least $r$ and is disjoint from $\hat C_{n_1}(x)$. This implies $\mu\pqty{\hat C_{n_1}(x)\bigtriangleup C^*(x)}\geq r$, a contradiction. By monotonicity of positive dilation and $\mathcal E(\mathcal E(C,r),t)=\mathcal E(C,r+t)$,
    \begin{equation*}
        \mathcal E(C^*(x),t)
        \subseteq\mathcal E(\hat C_{n_1}(x),r+t)
        \subseteq\mathcal E(\hat C_{n_1}(x),\epsilon),
    \end{equation*}
    and \autoref{eq:score event equivalence} gives $\{S^*\leq t,\ X\in G_{n_1,r}\} \subseteq\{\hat S_{n_1}\leq\epsilon\}$. Therefore, on $A_{n_1,\epsilon}$,
    \begin{align*}
        F_{n_1}(\epsilon)
        &\geq P(S^*\leq t,\ X\in G_{n_1,r}\mid\mathcal T_{n_1})\\
        &=P(S^*\leq t,\ X\in G_{n_1,r})\\
        &\geq P(S^*\leq t)-P_X(X\notin G_{n_1,r})
        >(1-\alpha)+\Delta_t/2
    \end{align*}
    for all sufficiently large $n_1$.

    Let $\delta_\epsilon>0$ be the margin in Assumption~\ref{asmp:identification multiple}(1) at distance $\epsilon$. Reducing it if necessary, take $0<\delta_\epsilon<1-\alpha$.  Fix an $x\in\mathcal X_{n_1}$ and work on $A_{n_1,\epsilon}$.  If $\mathcal E(\hat C_{n_1}(x),-\epsilon)=\emptyset$, then $P(\hat S_{n_1}\leq-\epsilon\mid X=x,\mathcal T_{n_1})=0 \leq1-\alpha-\delta_\epsilon$. Otherwise, write the maximal-component representation as $\hat C_{n_1}(x)=\bigsqcup_j[\hat a_j(x),\hat b_j(x)]$ and define
    \begin{equation*}
        B_{n_1,\epsilon}(x)
        :=\bigsqcup_{j:\,\hat b_j(x)-\hat a_j(x)>3\epsilon/2}
        [\hat a_j(x)+3\epsilon/4,\hat b_j(x)-3\epsilon/4].
    \end{equation*}
    Since $\mathcal E(\hat C_{n_1}(x),-\epsilon)$ is nonempty, some fitted component has length at least $2\epsilon$.  Thus $B_{n_1,\epsilon}(x)$ is a nonempty member of $\mathcal C_M$ and $\mathcal E(\hat C_{n_1}(x),-\epsilon) \subseteq B_{n_1,\epsilon}(x)$.  At least one component loses $3\epsilon/2$ of measure, so $\mu\pqty{B_{n_1,\epsilon}(x)\bigtriangleup\hat C_{n_1}(x)} \geq3\epsilon/2$.  Hence, on $A_{n_1,\epsilon}$,
    \begin{align*}
        \mu(B_{n_1,\epsilon}(x))
        &\leq\mu(\hat C_{n_1}(x))-3\epsilon/2\\
        &\leq\mu(C^*(x))+D_{n_1}-3\epsilon/2
        <\mu(C^*(x))-\epsilon,
    \end{align*}
    while the triangle inequality for symmetric difference gives
    \begin{align*}
        \mu\pqty{B_{n_1,\epsilon}(x)\bigtriangleup C^*(x)}
        &\geq
        \mu\pqty{B_{n_1,\epsilon}(x)\bigtriangleup\hat C_{n_1}(x)}-D_{n_1}\\
        &>3\epsilon/2-\epsilon/2=\epsilon.
    \end{align*}
    Conditional on each realisation of $\mathcal T_{n_1}$, $B_{n_1,\epsilon}(x)$ is a fixed member of $\mathcal C_M$.  Assumption~\ref{asmp:identification multiple}(1) therefore gives
    \begin{align*}
        P(\hat S_{n_1}\leq-\epsilon\mid X=x,\mathcal T_{n_1})
        &\leq P([Y^L,Y^U]\subseteq B_{n_1,\epsilon}(x)
        \mid X=x,\mathcal T_{n_1})\\
        &=P(B_{n_1,\epsilon}(x);x)\\
        &\leq1-\alpha-\delta_\epsilon.
    \end{align*}
    Thus the same bound holds for every $x\in\mathcal X_{n_1}$. Using the trivial bound outside $\mathcal X_{n_1}$ gives
    \begin{align*}
        F_{n_1}(-\epsilon)
        &\leq(1-\alpha-\delta_\epsilon)P_X(X\in\mathcal X_{n_1})
        +P_X(X\notin\mathcal X_{n_1})\\
        &\leq1-\alpha-\delta_\epsilon
        +P_X(X\notin\mathcal X_{n_1})
        \leq1-\alpha-\delta_\epsilon/2
    \end{align*}
    for all sufficiently large $n_1$ on $A_{n_1,\epsilon}$.

    Set $c_\epsilon=\frac12\min\{\Delta_t,\delta_\epsilon\}>0$.  For all sufficiently large $n_1$, the preceding bounds imply that, on $A_{n_1,\epsilon}$, $F_{n_1}(-\epsilon)\leq1-\alpha-c_\epsilon$ and $F_{n_1}(\epsilon)\geq1-\alpha+c_\epsilon$.

    Define
    \begin{equation*}
        \hat F_{n_2}(u) :=\frac1{n_2}\sum_{j\in\mathcal I_2} \indicator{s(Y_j^L,Y_j^U,X_j;\hat C_{n_1})\leq u}, \qquad p_{n_2}:=(1-\alpha)(1+n_2^{-1}).
    \end{equation*}
    Conditional on $\mathcal T_{n_1}$, the conformity scores in the calibration sample are i.i.d. with CDF $F_{n_1}$.  The conditional Dvoretzky--Kiefer--Wolfowitz inequality gives
    \begin{equation*}
        P\left(\sup_u\abs{\hat F_{n_2}(u)-F_{n_1}(u)}>c_\epsilon/2 \;\middle|\;\mathcal T_{n_1}\right) \leq2\exp(-n_2c_\epsilon^2/2).
    \end{equation*}
    For sufficiently large $n_2$, $(1-\alpha)/n_2<c_\epsilon/2$. Hence, on $A_{n_1,\epsilon}$ and the DKW event, $\hat F_{n_2}(-\epsilon)<p_{n_2} <\hat F_{n_2}(\epsilon)$.  By the definition of the empirical quantile in the conformal procedure, $\vartheta_{1-\alpha} =\inf\{u:\hat F_{n_2}(u)\geq p_{n_2}\}$, so the last two inequalities imply $-\epsilon<\vartheta_{1-\alpha}\leq\epsilon$.  Since $A_{n_1,\epsilon}$ is $\mathcal T_{n_1}$-measurable, taking expectations in the conditional DKW bound gives
    \begin{equation*}
        P(\abs{\vartheta_{1-\alpha}}>\epsilon) \leq P(A_{n_1,\epsilon}^c) +2\exp(-n_2c_\epsilon^2/2) \longrightarrow0.
    \end{equation*}
    Thus $\vartheta_{1-\alpha}=o_p(1)$.

    For every $C\in\overline{\mathcal C}_M$ and every $u\in\mathbb R$, $\mu\pqty{\mathcal E(C,u)\bigtriangleup C}\leq2M\abs{u}$.  This follows because each of the at most $M$ components gains or loses at most $2\abs u$ in measure, and overlaps can only reduce the change.  Hence, by \autoref{eq:conformal prediction set} and the triangle inequality,
    \begin{align*}
        \sup_{x\in\mathcal X_{n_1}}
        \mu\pqty{\tilde C(x)\bigtriangleup C^*(x)}
        &\leq
        \sup_{x\in\mathcal X_{n_1}}
        \mu\pqty{\tilde C(x)\bigtriangleup\hat C_{n_1}(x)}+D_{n_1}\\
        &\leq2M\abs{\vartheta_{1-\alpha}}+D_{n_1}
        =o_p(1),
    \end{align*}
    as required.
\end{proof}

\begin{remark}\label{remark:fixed censoring conformal error}
    Let $S^*=s(Y^L,Y^U,X;C^*)$, $F_*(t)=P(S^*\leq t)$, and $r_*:=\inf\{t\geq0:F_*(t)>1-\alpha\}<\infty$.  Thus $r_*$ is the lowest oracle-score threshold at which coverage can exceed $1-\alpha$, and measures the score gap induced by fixed censoring.  If all the conditions of Theorem~\ref{thm:consistency of conformal} hold except condition \eqref{eq:right crossing oracle score}, then, for every $\eta>0$,
    \begin{equation*}
        P\left\{
            \sup_{x\in\mathcal X_{n_1}}
            \mu\pqty{\tilde C(x)\bigtriangleup C^*(x)}
            >2Mr_*+\eta
        \right\}\longrightarrow0.
    \end{equation*}
    The factor $2Mr_*$ arises because enlarging each of at most $M$ interval components by $r_*$ adds at most $r_*$ at each endpoint.  Thus $r_*=0$ recovers consistency, whereas a positive score gap may leave the conformal set conservatively enlarged.  Consistency can instead be restored by conformalizing at $1-\alpha-\lambda_{n_1,n_2}$, where $\lambda_{n_1,n_2}\to0$ and, with $\pi_{n_1}=P_X(X\notin\mathcal X_{n_1})$, $\psi_{n_1}+\pi_{n_1}+n_2^{-1/2}=o(\lambda_{n_1,n_2})$.  The tradeoff is finite-sample coverage $1-\alpha-\lambda_{n_1,n_2}$ rather than $1-\alpha$, but the nominal level is recovered asymptotically.
\end{remark}
\begin{proof}[Proof of Theorem~\ref{thm:local validity}]
    Conditional on the training sample and the calibration and test observations' cell memberships, the conformity scores of calibration observations in $A_k$ and the score of a test observation in $A_k$ are exchangeable. If $r_k=\lceil(n_{2,k}+1)(1-\alpha)\rceil\leq n_{2,k}$, the rank argument gives conditional coverage of the reported interval at least $r_k/(n_{2,k}+1)\geq1-\alpha$. Otherwise, the procedure returns $\mathbb R$. Integrating over the conditioning variables gives local validity for the reported interval and hence for the latent outcome. The same argument applies conditional on training when the partition is constructed using only the training sample.
\end{proof}

\subsection{Identification of the oracle prediction set}\label{appendix:identification}
\label{appendix:additive-censoring-separation}

We give sufficient conditions for identification of the oracle prediction set, including a single interval as the case $K=1$. Consider observed intervals generated around a randomly drawn centre from a set of $K$ centres. Let $K\in\mathbb N$ and $0<\alpha<1$, and suppose that
\[
    [Y^L,Y^U]=[m_J(X)-R_L,\,m_J(X)+R_U],
\]
where $m_1,\ldots,m_K$ are measurable and ordered so that $m_1(x)\leq\cdots\leq m_K(x)$ for every $x$, and $R_L$ and $R_U$ are the left and right extensions from the centre. Conditional on $X=x$, suppose that $J$ is uniform on $\{1,\ldots,K\}$ and independent of $(R_L,R_U)$, that $0<R_L,R_U<\infty$ almost surely, and that the joint CDF
\[
    H_x(s,t)=P(R_L\leq s,R_U\leq t\mid X=x)
\]
is continuous. Choose $s_x,t_x>0$ such that $H_x(s_x,t_x)\geq1-\alpha$, and put $V_x=K(s_x+t_x)$.

When $\alpha<1/K$, let $F_{L,x}$ and $F_{U,x}$ be the conditional CDFs of $R_L$ and $R_U$ given $X=x$, and define their lower quantiles
\[
    \ell_x=F_{L,x}^{-1}(1-K\alpha),\qquad
    u_x=F_{U,x}^{-1}(1-K\alpha).
\]
We will consider partitions of the centres into $N(x)$ consecutive groups satisfying
\begin{equation}\label{eq:centre-group-gaps}
    m_{j+1}(x)-m_j(x)
    \begin{cases}
        \leq\ell_x+u_x, & \text{if $j$ and $j+1$ are in the same group},\\
        >V_x, & \text{if they are in different groups}.
    \end{cases}
\end{equation}

\begin{theorem}\label{thm:random-radius-separation}
The following hold.
\begin{enumerate}
    \item For every $x$, an oracle prediction set exists, and every oracle prediction set has at most $K$ components. In particular, every oracle prediction set belongs to $\mathcal C_M$ for $M\geq K$.
    \item For a fixed $x$, if $\alpha<1/K$, $H_x$ is strictly concave on $[\ell_x,V_x]\times[u_x,V_x]$, and the centres admit a partition satisfying \eqref{eq:centre-group-gaps}, then the oracle $C^*(x)$ is unique and has exactly $M_*(x)=N(x)\leq K$ components, one containing the centres in each group. Moreover, for every $\epsilon>0$ there is a $\delta_x(\epsilon)>0$ such that, for every $C\in\overline{\mathcal C}$,
    \begin{equation}\label{eq:pointwise-centre-separation}
        \mu(C)\leq\mu(C^*(x)),\quad
        \mu(C\bigtriangleup C^*(x))>\epsilon
        \quad\Longrightarrow\quad
        P([Y^L,Y^U]\subseteq C\mid X=x)
           <1-\alpha-\delta_x(\epsilon).
    \end{equation}
    \item Suppose the conditions in part 2 hold for every $x\in\mathcal X$. If $\mathcal X\subseteq\mathbb R^d$ is compact, $x\mapsto m_j(x)-m_1(x)$ is continuous for $j=2,\ldots,K$, and $(x,s,t)\mapsto H_x(s,t)$ is jointly continuous on $\mathcal X\times[0,\infty)^2$, then, for every $\epsilon>0$, there is a $\delta(\epsilon)>0$ such that \eqref{eq:pointwise-centre-separation} holds with $\delta_x(\epsilon)$ replaced by $\delta(\epsilon)$, uniformly for all $x\in\mathcal X$ and $C\in\overline{\mathcal C}$. In particular, Assumption~\ref{asmp:identification multiple} holds with $\bar M=K$.
\end{enumerate}
\end{theorem}

The concavity and gap conditions are sufficient. The gap conditions ensure that centres in the same group are close enough to belong to the same oracle component, while centres in different groups are too far apart to belong to one component. The level $\alpha$ is fixed in the theoretical analysis, although the oracle and its component count may change if $\alpha$ is varied. The example below illustrates the concavity condition.

Part 3 allows the relative positions of the centres and the conditional distribution of $(R_L,R_U)$ to vary continuously with $x$. Only the relative positions enter because a common translation preserves coverage, volume and symmetric-difference measure.

In particular, $K=1$ gives the single-interval model $[Y^L,Y^U]=[m_1(X)-R_L,m_1(X)+R_U]$. The restriction $\alpha<1/K$ is then automatic and there are no centre-gap conditions. Under the strict-concavity condition in part 2, $C_I^*(x)=C^*(x)$ is a unique interval. If $(R_L,R_U)$ is also independent of $X$, then $H_x=H$ and changing $x$ only translates the distribution by $m_1(x)$. The separation gap is therefore uniform in $x$, verifying both Assumptions~\ref{asmp:identification} and~\ref{asmp:identification multiple}, the latter with $\bar M=1$.

The following compactness result is used to prove both existence and separation in Theorem~\ref{thm:random-radius-separation}.

\begin{lemma}\label{lemma:centre-compactness}
Fix $x$, let $m_1,\ldots,m_K$ be real numbers, and suppose $Y^L<Y^U$ almost surely conditional on $X=x$. Let $(D_n)$ be a sequence in $\mathcal C$ such that every component of every $D_n$ contains at least one $m_j$, $\sup_n\mu(D_n)<\infty$, and
\[
    p:=\limsup_{n\to\infty}P([Y^L,Y^U]\subseteq D_n\mid X=x)>0.
\]
Then there exist a subsequence $(D_{n_k})$ and a set $D\in\mathcal C$ such that
\[
    \mu(D_{n_k}\bigtriangleup D)\longrightarrow0,\qquad
    P([Y^L,Y^U]\subseteq D\mid X=x)\geq p.
\]
In particular, $\mu(D_{n_k})\to\mu(D)$.
\end{lemma}

\begin{proof}
First take a subsequence whose coverage converges to $p$. Every set in the sequence has at most $K$ components. Passing to a further subsequence, fix the number of components and which centres each contains. Each component contains a centre and has bounded length, so its endpoints are bounded. Along another subsequence all endpoints converge.

Let $D'$ be the union of the limiting intervals, allowing degenerate or touching intervals. Endpoint convergence gives $\mu(D_{n_k}\bigtriangleup D')\to0$. Any observed interval contained in infinitely many $D_{n_k}$ is contained in $D'$, so the reverse Fatou inequality yields
\[
    P([Y^L,Y^U]\subseteq D'\mid X=x)
    \geq\limsup_{k\to\infty}P([Y^L,Y^U]\subseteq D_{n_k}\mid X=x)=p.
\]
Merge touching limiting intervals and remove isolated singletons to obtain $D$. Since the observed interval has positive length almost surely, this changes neither coverage nor volume. The positive coverage of $D$ ensures that it is nonempty, and hence $D\in\mathcal C$. The stated conclusions follow.
\end{proof}

\begin{proof}[Proof of Theorem~\ref{thm:random-radius-separation}]
For parts 1 and 2, fix $x$, suppress dependence on $x$ except in $P(C;x)$, and write $q=1-\alpha$ and
\[
    P(C;x)=P([Y^L,Y^U]\subseteq C\mid X=x),\qquad
    H(s,t)=P(R_L\leq s,R_U\leq t\mid X=x).
\]
By Lemma~\ref{lemma: inf P(y) = P(yl,yu)}, validity at $x$ is equivalent to $P(C;x)\geq q$. We first prove existence and the component bound under the initial assumptions, and then impose the additional conditions for uniqueness and separation.

\emph{Existence and component bound.} Each observed interval contains $m_J$. Hence removing from $C$ any component that contains no $m_j$ leaves $P(C;x)$ unchanged. Every oracle must therefore have at most $K$ components. Otherwise, removing such a component would reduce its volume without reducing coverage.

Since the extensions are finite almost surely, choose finite $s,t>0$ with $H(s,t)\geq q$. The set
\[
    C_0=\bigcup_{j=1}^K[m_j-s,m_j+t]
\]
is valid and $\mu(C_0)\leq K(s+t)$. Choose a minimising sequence of valid sets and remove their components containing no centre. The resulting sequence $(D_n)$ still minimises volume and has bounded volumes. By Lemma~\ref{lemma:centre-compactness}, a subsequence converges in symmetric-difference measure to a valid set $D\in\mathcal C$, with $\mu(D_n)\to\mu(D)$. Thus $D$ attains the minimum.

\emph{The oracle components.} Now impose the additional conditions and take $s,t$ to be the values chosen before the theorem. By the definition of lower quantiles, $\ell\leq s$ and $u\leq t$. Hence
\[
    \ell+u\leq s+t\leq V.
\]
Also, $\ell,u>0$ because the extensions are positive almost surely and $1-K\alpha>0$.

Every valid set $C$ contains every centre. If it omits $m_j$, then
\[
    P(C;x)\leq 1-\frac1K<q.
\]
Write $[a_j,b_j]$ for the component of $C$ containing $m_j$, and set $L_j=m_j-a_j$ and $U_j=b_j-m_j$. Conditional independence gives
\[
    P(C;x)=\frac1K\sum_{j=1}^K H(L_j,U_j)\geq q.
\]
As each summand is at most one, for every $j$,
\[
    1-K\alpha\leq H(L_j,U_j)
    \leq\min\{F_L(L_j),F_U(U_j)\}.
\]
The marginal quantiles therefore give $L_j\geq\ell$ and $U_j\geq u$. Thus every valid set contains $[m_j-\ell,m_j+u]$ for each $j$. The intervals associated with consecutive centres in the same group overlap or touch, so those centres must lie in one component.

The valid set constructed above has volume at most $V$. An oracle therefore has volume at most $V$, and cannot join two neighbouring groups because any component joining them would have length greater than $V$. Since an oracle has no component without a centre, it has exactly one component for each group.

\emph{Uniqueness.} Let $I_1,\ldots,I_N$ be the index sets of the groups. For an oracle with components $[A_h,B_h]$, its coverage is
\[
    \frac1K\sum_{h=1}^{N}\sum_{j\in I_h}
       H(m_j-A_h,B_h-m_j).
\]
Each pair of arguments belongs to $[\ell,V]\times[u,V]$. The lower bounds follow from validity, and the upper bounds follow from the oracle's volume bound.

Suppose there are two distinct oracles $C_1$ and $C_2$. Average their corresponding endpoints and call the resulting set $\bar C$. Its intervals remain disjoint and contain the same groups of centres, with
\[
    \mu(\bar C)=\frac{\mu(C_1)+\mu(C_2)}{2},\qquad
    P(\bar C;x)>\frac{P(C_1;x)+P(C_2;x)}{2}\geq q.
\]
The strict inequality follows from strict concavity of $H$, since at least one corresponding endpoint pair differs. By continuity, the endpoints of $\bar C$ can be moved slightly inward while retaining coverage at least $q$. This produces a valid set with smaller volume, a contradiction. The oracle is therefore unique.

Since $Y^L<Y^U$ almost surely, deleting singleton components preserves coverage, volume and symmetric-difference measure. Both separation claims therefore reduce to $C\in\mathcal C$, with margins below $q$ covering an empty remainder.

\emph{Pointwise separation.} Fix $\epsilon>0$. If the stated positive gap failed, there would be a sequence $C_n$ with
\[
    \mu(C_n)\leq\mu(C^*),\qquad
    \mu(C_n\bigtriangleup C^*)>\epsilon,\qquad
    P(C_n;x)\longrightarrow q\text{ from below}.
\]
Indeed, uniqueness already excludes a valid set satisfying these volume and distance restrictions.

Remove components containing no centre to obtain $D_n\subseteq C_n$ with $P(D_n;x)=P(C_n;x)$. Since $P(C_n;x)\to q>0$, the sets $D_n$ are nonempty for all sufficiently large $n$. Lemma~\ref{lemma:centre-compactness} gives a subsequence converging in symmetric-difference measure to $D\in\mathcal C$, with $P(D;x)\geq q$ and $\mu(D)\leq\mu(C^*)$. Uniqueness identifies $D$ with $C^*$. In particular, $\mu(D_n\bigtriangleup C^*)\to0$ and $\mu(D_n)\to\mu(C^*)$ along the subsequence. But
\[
    \mu(C_n\bigtriangleup C^*)
    \leq\mu(C_n\setminus D_n)+\mu(D_n\bigtriangleup C^*)
    \leq\mu(C^*)-\mu(D_n)+\mu(D_n\bigtriangleup C^*)\longrightarrow0,
\]
a contradiction. Thus the coverage deficit is bounded away from zero. Choosing a smaller positive constant gives the strict inequality in \eqref{eq:pointwise-centre-separation}.

\emph{Uniform separation.} For part 3, restore dependence on $x$ and write $v(x)=\mu(C^*(x))$. Translate the centres and all candidate sets by $-m_1(x)$, without changing notation. Coverage, volume and symmetric-difference measure are unchanged, and the translated centres are continuous and bounded on $\mathcal X$.

For each $x$, finiteness of the extensions gives a finite $T_x$ with $H_x(T_x,T_x)>q$. Joint continuity and a finite open cover of the compact set $\mathcal X$ give a common $T$ with $\inf_{x\in\mathcal X}H_x(T,T)>q$. The set $\bigcup_{j=1}^K[m_j(x)-T,m_j(x)+T]$ is valid at $x$, so
\[
    \sup_{x\in\mathcal X}v(x)\leq2KT.
\]

We first extend the endpoint argument to varying covariates. Suppose $x_n\to x_0$ and $D_n\in\mathcal C$ have bounded volumes, contain every $m_j(x_n)$, and have no component without a centre. Along a subsequence, fix the number $N$ of components and their centre index sets $I_1,\ldots,I_N$. Their endpoints are bounded, so a further subsequence satisfies $A_{h,n}\to A_h$ and $B_{h,n}\to B_h$ for every $h$. Let $D'=\bigcup_{h=1}^N[A_h,B_h]$. Joint continuity gives
\[
    \lim_n P(D_n;x_n)
    =\frac1K\sum_{h=1}^N\sum_{j\in I_h}
       H_{x_0}(m_j(x_0)-A_h,B_h-m_j(x_0))
    \leq P(D';x_0).
\]
The inequality allows for additional coverage when limiting intervals touch. If the limit on the left is positive, merging touching intervals and removing isolated singletons gives $D\in\mathcal C$ with
\begin{equation}\label{eq:centre-varying-x-limit}
    \mu(D_n\bigtriangleup D)\longrightarrow0,\qquad
    \mu(D_n)\longrightarrow\mu(D),\qquad
    P(D;x_0)\geq\lim_nP(D_n;x_n)
\end{equation}
along this subsequence. Removing singletons preserves coverage because the observed interval has positive length almost surely at $x_0$.

We next prove continuity of the oracle. Fix $x_0$ and write its components as $[A_h,B_h]$. For every centre in component $h$, the bounds already proved give
\[
    L_j:=m_j(x_0)-A_h\geq\ell_{x_0},\qquad
    U_j:=B_h-m_j(x_0)\geq u_{x_0},\qquad
    L_j+U_j\leq V_{x_0}.
\]
Hence $L_j\leq V_{x_0}-u_{x_0}<V_{x_0}$ and $U_j\leq V_{x_0}-\ell_{x_0}<V_{x_0}$. For sufficiently small $\eta>0$, the expanded intervals in
\[
    C^\eta=\bigcup_{h=1}^{N(x_0)}[A_h-\eta,B_h+\eta]
\]
remain disjoint, and all $(L_j+\eta,U_j+\eta)$ remain in the strict-concavity rectangle. Monotonicity and strict concavity imply
\[
    H_{x_0}(L_j+\eta,U_j+\eta)>H_{x_0}(L_j,U_j).
\]
Indeed, equality would make $H_{x_0}$ constant on the segment between these two points, contradicting strict concavity. Therefore
\[
    P(C^\eta;x_0)>q,\qquad
    \mu(C^\eta)\leq v(x_0)+2K\eta.
\]
If $x_n\to x_0$, continuity of the centres keeps them in the same components of $C^\eta$ for all sufficiently large $n$, and joint continuity gives $P(C^\eta;x_n)\to P(C^\eta;x_0)>q$. Thus
\[
    \limsup_n v(x_n)\leq v(x_0)+2K\eta.
\]
Letting $\eta\downarrow0$ gives $\limsup_n v(x_n)\leq v(x_0)$. Each oracle $C^*(x_n)$ contains every centre and has no component without a centre. By \eqref{eq:centre-varying-x-limit}, every subsequential endpoint limit is valid at $x_0$ and has volume at most $v(x_0)$. Uniqueness identifies every such limit with $C^*(x_0)$. Consequently,
\begin{equation}\label{eq:centre-oracle-continuity}
    v(x_n)\longrightarrow v(x_0),\qquad
    \mu(C^*(x_n)\bigtriangleup C^*(x_0))\longrightarrow0.
\end{equation}

Finally, if uniform separation failed for some $\epsilon>0$, there would be $x_n\in\mathcal X$ and $C_n\in\mathcal C$ with
\[
    \mu(C_n)\leq v(x_n),\qquad
    \mu(C_n\bigtriangleup C^*(x_n))>\epsilon,\qquad
    P(C_n;x_n)\longrightarrow q.
\]
Compactness gives a subsequence with $x_n\to x_0$. Since $q>1-1/K$, each $C_n$ contains every centre for all sufficiently large $n$. Remove components containing no centre to obtain $D_n\subseteq C_n$ with unchanged coverage. By \eqref{eq:centre-varying-x-limit} and \eqref{eq:centre-oracle-continuity}, a further subsequence converges to a valid set $D$ at $x_0$ with $\mu(D)\leq v(x_0)$. Uniqueness gives $D=C^*(x_0)$, and hence
\begin{align*}
    \mu(C_n\bigtriangleup C^*(x_n))
    &\leq \mu(C_n\setminus D_n)
          +\mu(D_n\bigtriangleup C^*(x_0))
          +\mu(C^*(x_0)\bigtriangleup C^*(x_n))\\
    &\leq v(x_n)-\mu(D_n)
          +\mu(D_n\bigtriangleup C^*(x_0))
          +\mu(C^*(x_0)\bigtriangleup C^*(x_n))
      \longrightarrow0,
\end{align*}
a contradiction. The coverage deficit is therefore bounded away from zero uniformly in $x$. Taking a smaller positive bound gives the strict inequality claimed in part 3.
\end{proof}

The concavity condition holds, for example, for the following family of independent Weibull distributions. Suppose that, conditional on $X=x$, $R_L$ and $R_U$ are independent, with respective CDFs
\[
    F_{L,x}(s)=1-e^{-(s/\sigma_L(x))^{\kappa_L(x)}},\qquad
    F_{U,x}(t)=1-e^{-(t/\sigma_U(x))^{\kappa_U(x)}},\quad s,t\geq0,
\]
where $\sigma_L(x),\sigma_U(x)>0$ and $0<\kappa_L(x),\kappa_U(x)\leq1$. Suppose also that $(R_L,R_U)$ is independent of $J$ conditional on $X=x$, as in the theorem. If $\alpha<1/(2K)$, the marginal quantiles are
\[
    \ell_x=\sigma_L(x)[-\log(K\alpha)]^{1/\kappa_L(x)},\qquad
    u_x=\sigma_U(x)[-\log(K\alpha)]^{1/\kappa_U(x)}.
\]
With $c_\alpha=-\log(1-\sqrt{1-\alpha})$, one may choose
\[
    s_x=\sigma_L(x)c_\alpha^{1/\kappa_L(x)},\qquad
    t_x=\sigma_U(x)c_\alpha^{1/\kappa_U(x)},\qquad
    V_x=K(s_x+t_x).
\]
Indeed, $H_x(s_x,t_x)=1-\alpha$. To check strict concavity, suppress dependence on $x$ and let $F$ be a Weibull CDF with shape $0<\kappa\leq1$ and scale $\sigma>0$. For $v>0$, put $z=(v/\sigma)^\kappa$. Then
\[
    F''(v)<0,\qquad
    \frac{-F(v)F''(v)}{F'(v)^2}
    =(e^z-1)\left(1+\frac{1-\kappa}{\kappa z}\right)
    \geq e^z-1.
\]
On $[\ell,V]\times[u,V]$, both marginal CDFs are at least $1-K\alpha$, so each ratio above is at least $(1-K\alpha)/(K\alpha)>1$. Independence gives $H(s,t)=F_L(s)F_U(t)$ and
\[
    \det\nabla^2H(s,t)
    =F_L'(s)^2F_U'(t)^2
    \left\{
    \frac{-F_L(s)F_L''(s)}{F_L'(s)^2}
    \frac{-F_U(t)F_U''(t)}{F_U'(t)^2}-1
    \right\}>0.
\]
The diagonal entries of the Hessian are negative, so $H$ is strictly concave on the rectangle. The theorem applies whenever the stated centre-gap conditions also hold. In particular, shape parameter one gives the exponential distribution. If $\mathcal X$ is compact, the scale and shape parameters and each $m_j-m_1$ are continuous on $\mathcal X$, and the centre-gap conditions hold throughout $\mathcal X$, part 3 gives uniform separation.

\section{Conditional distribution estimation}
\subsection{Uniform consistency of the kernel smoothing estimator}\label{appendix: kernel smoothing}
We give sufficient conditions for uniform consistency of the kernel estimator of \(P(C;x)=P([Y^L,Y^U]\subseteq C\mid X=x)\). The observations $Z_i=(X_i,Y_i^L,Y_i^U)$, with endpoint pairs $W_i=(Y_i^L,Y_i^U)$, are i.i.d., and \(d\) and \(M\) are fixed. Suppose $X$ has Lebesgue density $p_X$. For \(C=\bigcup_{m=1}^{M'}[a_m,b_m]\in\overline{\mathcal C}_M\), let
\[
    f_C(y^L,y^U)=\sum_{m=1}^{M'}\indicator{a_m\leq y^L\leq y^U\leq b_m},
    \qquad g_C(x)=p_X(x)P(C;x).
\]
The definition includes exact observations with \(y^L=y^U\). Write
\[
    P_n(C;x)=\frac{\hat g_{n,C}(x)}{\hat p_n(x)},\qquad
    \hat g_{n,C}(x)=\frac{1}{nh^d}\sum_{i=1}^n f_C(W_i)K\pqty{\frac{x-X_i}{h}},\qquad
    \hat p_n(x)=\frac{1}{nh^d}\sum_{i=1}^n K\pqty{\frac{x-X_i}{h}}.
\]
Set \(P_n(C;x)=0\) if the denominator vanishes. The class \(\Phi_M=\{f_C:C\in\overline{\mathcal C}_M\}\cup\{1\}\) is a VC class with envelope 1. Approximation of endpoints from outside by rational endpoints gives pointwise measurability, including at atoms. Under the kernel condition below, its products with kernel translates and dilates also have polynomial covering numbers \citep{vaart2013WeakConvergence}.

\begin{assumption}[Kernel]\label{asmp:kernel}
    \(K(u)=\prod_{j=1}^dK_0(u_j)\), where \(K_0\) is nonnegative, continuous, symmetric, of bounded variation, supported on \([-1,1]\), and \(\int K_0(s)\dd s=1\).
\end{assumption}
\begin{assumption}[Bandwidth]\label{asmp:bandwidth}
    The deterministic bandwidth sequence \(h=h_n\) satisfies
    \[
        h\to0,\qquad \frac{nh^d}{\log n}\to\infty,\qquad
        \frac{\log(1/h)}{\log\log n}\to\infty.
    \]
\end{assumption}

The following lemma follows from Theorem~4 of \citet{einmahl2005UniformBandwidth}, applied to \(\Phi_M\) with coefficient functions \(c_\phi\equiv1\), \(d_\phi\equiv0\), and volume bandwidth proportional to \(h^d\) after rescaling the kernel support.
\begin{lemma}\label{lemma:kernel-compact}
    Let \(\mathcal X'\) be compact, with a closed neighbourhood \(J\subseteq\operatorname{int}(\mathcal X)\) on which \(p_X\) is continuous and bounded below by \(c>0\). Under Assumptions~\ref{asmp:kernel} and \ref{asmp:bandwidth}, with probability one,
    \begin{equation}\label{eq:kernel-centred-rate}
        \sup_{x\in\mathcal X'}\sup_{C\in\overline{\mathcal C}_M}
        |\hat g_{n,C}(x)-\Ep\hat g_{n,C}(x)|
        +\sup_{x\in\mathcal X'}|\hat p_n(x)-\Ep\hat p_n(x)|
        =O\pqty{\sqrt{\frac{\log(1/h)}{nh^d}}}.
    \end{equation}
\end{lemma}

It remains to bound the bias. Suppose \(p_X\) and all \(g_C\) have twice continuously differentiable versions on an open neighbourhood of \(J\), with
\[
    \sup_{x\in J}\|D_x^2p_X(x)\|
    +\sup_{C\in\overline{\mathcal C}_M}\sup_{x\in J}\|D_x^2g_C(x)\|<\infty.
\]
Taylor expansion and symmetry of \(K\) then give
\[
    \sup_{x\in\mathcal X'}\sup_{C\in\overline{\mathcal C}_M}
    |\Ep\hat g_{n,C}(x)-g_C(x)|
    +\sup_{x\in\mathcal X'}|\Ep\hat p_n(x)-p_X(x)|=O(h^2).
\]
No joint density of the endpoints is required, so discrete brackets and exact observations are allowed under these smoothness conditions in \(x\). Combining the bounds gives \(\inf_{x\in\mathcal X'}\hat p_n(x)\geq c/2\) eventually almost surely. Using
\[
    |P_n(C;x)-P(C;x)|\leq\frac{2}{c}
    \left\{|\hat g_{n,C}(x)-g_C(x)|+|\hat p_n(x)-p_X(x)|\right\},
\]
we obtain, with probability one,
\begin{equation}\label{eq:kernel-uniform-rate}
    \sup_{x\in\mathcal X'}\sup_{C\in\overline{\mathcal C}_M}|P_n(C;x)-P(C;x)|
    =O(r_n),\qquad r_n=\sqrt{\frac{\log(1/h)}{nh^d}}+h^2.
\end{equation}

For expanding compact regions \(\mathcal X_n\subseteq\{x:p_X(x)>0\}\), assume \(p_X\) is bounded and \(p_X\) and all \(g_C\), extended by zero where \(p_X=0\), are twice continuously differentiable on \(\mathbb R^d\), with the preceding Hessian bound holding globally. The bounded-envelope argument in the proof of Theorem~1 of \citet{einmahl2005UniformBandwidth} applies to the same product class, whose covering numbers are polynomial and whose variances are uniformly \(O(h^d)\). Thus the centred and bias bounds hold globally. With \(\delta_n=\inf_{x\in\mathcal X_n}p_X(x)\) and \(r_n/\delta_n\to0\), the same ratio bound yields
\[
    \sup_{x\in\mathcal X_n}\sup_{C\in\overline{\mathcal C}_M}|P_n(C;x)-P(C;x)|
    =O(\delta_n^{-1}r_n)=o(1)\qquad\text{almost surely}.
\]
The factor \(\delta_n^{-1}\) also appears in \citet[Theorems~8--9]{hansen2008UniformConvergence}. Here the bound is additionally uniform over \(C\).

\subsection{Conditional distribution estimation with high-dimensional covariates}
\label{appendix:distributional-estimation}

We fit the distributional random forest of \citet{cevid2022DistributionalRandomForests} to the observed joint response $(Y^L,Y^U)$ given $X$. We use its observation weights to estimate $P([Y^L,Y^U]\subseteq C\mid X=x)$. For finite observed endpoints, the construction uses the training sample $\mathcal I_1$ as follows.
\begin{enumerate}[label=\arabic*.,leftmargin=*]
    \item \emph{Fit the random forest.} Use $X_i$ as the predictor vector and $(Y_i^L,Y_i^U)$ as the joint response for $i\in\mathcal I_1$. Follow Algorithm~1 and Appendix~A of \citet{cevid2022DistributionalRandomForests} for preprocessing, tree construction and default parameter settings. Categorical predictors are represented by indicators, and each response coordinate is centred by its training-sample mean and divided by its training-sample standard deviation. For each tree, draw a subsample without replacement and divide it into two parts. Use the first part to choose splits with their random Fourier approximation to the maximum mean discrepancy (MMD) criterion, and use the second part to populate the fitted leaves.

    \item \emph{Calculate the weights at $x$.} Let $B$ be the number of fitted trees. Pass $x$ through each tree using its fitted cuts. Let $\mathcal L_b(x)$ be the indices of the second-part observations in the leaf reached in tree $b$. Give each observation in this leaf weight $1/|\mathcal L_b(x)|$ and average across trees,
    \begin{equation}\label{eq:drf-containment-weights}
        w_i(x)=\frac{1}{B}\sum_{b=1}^{B}
        \frac{\indicator{i\in\mathcal L_b(x)}}{|\mathcal L_b(x)|},
        \qquad i\in\mathcal I_1.
    \end{equation}
    These weights satisfy $w_i(x)\geq0$ and $\sum_{i\in\mathcal I_1}w_i(x)=1$.

    \item \emph{Estimate $P([Y^L,Y^U]\subseteq C\mid X=x)$.} For a candidate set $C$, add the weights of observations whose entire reported interval lies in $C$,
    \begin{equation}\label{eq:drf-containment-estimator}
        P_{n_1}^{\mathrm{DRF}}(C;x)
        =\sum_{i\in\mathcal I_1}w_i(x)
          \indicator{[Y_i^L,Y_i^U]\subseteq C}.
    \end{equation}
    For a single interval $C=[a,b]$, observation $i$ contributes precisely when $a\leq Y_i^L$ and $Y_i^U\leq b$.

    \item \emph{Construct the prediction set.} For each $x$, use the endpoint pairs and weights $w_i(x)$ in the algorithms of Appendix~\ref{appendix:optimized-algorithms}. Minimise the total length of $C$ subject to $P_{n_1}^{\mathrm{DRF}}(C;x)\geq q$, with $q=1-\alpha$, or $q=1-\alpha-\psi_{n_1}$ for the relaxed constraint. Denote the resulting set by $\hat C(x)$, using a fixed method to break ties.
\end{enumerate}

After fitting $x\mapsto\hat C(x)$ on the training sample, hold the fitted procedure fixed and compute $s_i=s(Y_i^L,Y_i^U,X_i;\hat C)$ for each $i\in\mathcal I_2$ using Equation~\eqref{eq:score function}. Compute $\vartheta_{1-\alpha}$ from these scores using the empirical-quantile formula in Section~\ref{sec:conformal} with large enough $n_2$. The conformalized prediction set is
\[
    \tilde C(x)=\mathcal E(\hat C(x),\vartheta_{1-\alpha}),
\]
where $\mathcal E$ is the endpoint adjustment in Equation~\eqref{eq:conformal prediction set}.

\section{Algorithms and computation time}
\label{appendix:optimized-algorithms}

We describe three algorithms for computing \(\hat C(x)\) from the estimated probabilities \(P_{n_1}(C;x)\), using the kernel weights in Equation~\eqref{eq:Phat} or the forest weights in Equation~\eqref{eq:drf-containment-weights}. It is enough to consider observed endpoints, since shrinking each component to the hull of the observations assigned to it cannot increase its length or reduce its weighted sample probability.

\paragraph{\(M=1\), endpoint scan.}
For estimating a single prediction interval, we find the shortest feasible candidate by pairing each observed lower endpoint with the smallest upper endpoint that satisfies the sample constraint. Fix \(x\) and let \(\mathcal I_1(x)=\{i\in\mathcal I_1:w_i(x)>0\}\) contain the positive-weight training observations, with \(n_x=|\mathcal I_1(x)|\). Throughout, the weights sum to one and the required level is \(q=1-\alpha\), or \(q=1-\alpha-\psi_{n_1}\) for the relaxed constraint, with \(q\in(0,1]\). Process the distinct lower endpoints \(t_0\) in decreasing order. At each \(t_0\), insert observations with lower endpoint \(t_0\) into a Fenwick tree indexed by the ordered upper endpoints. The tree then contains the observations with \(Y_i^L\geq t_0\). If their total weight is at least \(q\), let \(t_1\) be the smallest upper endpoint whose cumulative weight is at least \(q\). The shortest candidate is \(\hat C_I(x)=[\hat\tau_0(x),\hat\tau_1(x)]\), at the chosen constraint level. The endpoint scan takes \(O(n_x\log n_x)\) time and \(O(n_x)\) memory.

\paragraph{General \(M\), exact MILP.}
For multiple intervals, we minimise their total length by jointly choosing the component endpoints and assigning observations to components through binary variables. For computation, we allow zero-length and touching components, merging touching components to obtain a set in \(\overline{\mathcal C}_M\). This accommodates exact observations and leaves the infimum over \(\mathcal C_M\) unchanged. Let \(\ell=\min_{i\in\mathcal I_1(x)}Y_i^L\), \(r=\max_{i\in\mathcal I_1(x)}Y_i^U\), and \(H=r-\ell\). The case \(H=0\) is immediate. Suppress the dependence of the component endpoints \(t_{0m}(x),t_{1m}(x)\) on \(x\). With activity indicators \(u_m\), and assignment indicators \(z_{im}\), consider
\begin{align*}
    \min\quad &\sum_{m=1}^{M}(t_{1m}-t_{0m})\\
    \mathtext{s.t.}\quad
    &\ell\leq t_{0m}\leq t_{1m}\leq r,
    \quad 0\leq t_{1m}-t_{0m}\leq Hu_m,\\
    &t_{1m}\leq t_{0,m+1},\quad u_m\geq u_{m+1}, &&m<M,\\
    &z_{im}\leq u_m,\quad \sum_{m=1}^{M}z_{im}\leq1,\\
    &t_{0m}+(r-Y_i^L)z_{im}\leq r,
    \quad t_{1m}-(Y_i^U-\ell)z_{im}\geq\ell,\\
    &u_m\leq\sum_{i\in\mathcal I_1(x)}z_{im},
    \quad \sum_{i\in\mathcal I_1(x)}w_i(x)\sum_{m=1}^{M}z_{im}\geq q,\\
    &u_m,z_{im}\in\{0,1\}.
\end{align*}
The fitted endpoints are denoted by \(\hat\tau_{0m}(x),\hat\tau_{1m}(x)\), as in the main text, and their active intervals form \(\hat C(x)\). Constraints involving \(i\) and \(m\) hold for \(i\in\mathcal I_1(x)\) and \(m=1,\ldots,M\). If \(z_{im}=1\), the endpoint constraints require \([Y_i^L,Y_i^U]\subseteq[t_{0m},t_{1m}]\). Conversely, any feasible set gives a feasible assignment. The optimum therefore minimises \(\mu(C)\) over \(C\in\overline{\mathcal C}_M\) subject to \(P_{n_1}(C;x)\geq q\), the computational version of the sample problem in Assumption~\ref{asmp:approximate solution multiple}.

\paragraph{Dynamic programming.}
We build the prediction set one interval at a time, keeping the shortest solution for each number of intervals and required coverage level. To make the number of coverage levels finite, we round each candidate interval's estimated probability down to a multiple of \(\Delta>0\).

First, sort the distinct observed endpoints as \(z_1<\cdots<z_B\), where \(B\leq2n_x\). Let \(\mathcal A\) contain the candidate pairs \((j,k)\), with \(z_j\) an observed lower endpoint, \(z_k\) an observed upper endpoint, and \(j\leq k\). For each candidate, compute
\[
    r_{jk}=P_{n_1}([z_j,z_k];x),\qquad
    c_{jk}=\left\lfloor r_{jk}/\Delta\right\rfloor,
    \qquad Q=\left\lceil q/\Delta\right\rceil.
\]
Here \(c_{jk}\) is its rounded probability measured in units of \(\Delta\), and \(Q\) is the required number of units.

The recursion compares two choices, omit the endpoint \(z_k\), or add an interval ending there. Let \(D_{m,k,s}\) be the minimum total length using at most \(m\) intervals with endpoints among \(z_1,\ldots,z_k\) and at least \(s\) coverage units. Set \(D_{m,k,0}=0\) and all infeasible states to \(+\infty\), including \(D_{0,k,s}\) and \(D_{m,0,s}\) for \(s>0\). Then
\[
    D_{m,k,s}=\min\left\{D_{m,k-1,s},
        \min_{j:(j,k)\in\mathcal A}
        \left[D_{m-1,j-1,(s-c_{jk})_+}+z_k-z_j\right]\right\}.
\]
The first term omits \(z_k\). The second adds \([z_j,z_k]\) to a solution ending before \(z_j\), so the intervals remain separated. Touching intervals can be merged without changing their union. Tracing the selected intervals back from \(D_{M,B,Q}\) gives \(\hat C(x)\).

Rounding down preserves the original sample constraint. The selected intervals are separated, so their estimated probabilities add without double counting. If \(\mathcal S\) denotes their index pairs, then
\[
    P_{n_1}(\hat C(x);x)=\sum_{(j,k)\in\mathcal S}r_{jk}
    \geq\Delta\sum_{(j,k)\in\mathcal S}c_{jk}
    \geq\Delta Q\geq q.
\]
Rounding loses less than \(\Delta\) per interval. Consequently, if \(V_M(u)\) is the exact minimum volume at level \(u\) and \(V_{M,\Delta}(q)\) is the computed volume, then
\[
    V_M(q)\leq V_{M,\Delta}(q)\leq V_M(q+M\Delta),
    \qquad q+M\Delta\leq1.
\]

A smaller \(\Delta\) gives a finer approximation but requires more computation. Choosing \(\Delta=\varepsilon/M\) limits the increase in the required coverage level in this bound to \(\varepsilon\). The method is exact when all candidate probabilities are multiples of \(\Delta\), for example when all weights \(w_i(x)\) are. For asymptotic use, one may let \(M\Delta_{n_1}\to0\), provided the exact optimum is stable to this change in coverage level.

The computational cost depends on the number of distinct endpoints \(B\) and coverage units \(Q\). Time and memory costs are \(O(n_x\log n_x+B^2+MB^2Q)\) and \(O(B^2+BQ)\), respectively.

\paragraph{Computation time.}
We compare computation times and solution accuracy under Model A to assess when each method is practical. Table~\ref{tab:optimization-runtime} reports the computation times. Here \(n\) denotes the full sample size, with \(n_1\) observations used for estimation. For each sample size in the table, we generated ten independent samples and evaluated the local problem at \(x=0.75\), with \(q=0.9\), a 75\% training split, and the Epanechnikov kernel with \(h=0.35\). We evaluated the MILP on all ten samples for \(n\leq1000\) and for \(n=2000,M=2\), and the endpoint scan and dynamic program for every \(n\). For the remaining MILP configurations, preliminary runs reached a 1800-second limit, so we did not complete all ten repetitions.

We measure the time required to solve each optimisation problem after constructing the local sample and its weights. The computations used Python 3.12 and SciPy 1.18 with the HiGHS backend on ARM64 macOS. We set \(M\Delta=0.003\), giving \((\Delta,Q)=(0.0015,600)\) for \(M=2\) and \((0.001,900)\) for \(M=3\). For each sample, we run the endpoint scan 200 times and report the average time per run.

\begingroup
\sisetup{group-minimum-digits=4}
\def\mr#1#2#3{\shortstack{\num{#1}\\[-1pt]{\scriptsize[\num{#2},\,\num{#3}]}}}
\def\timeround#1{\num[round-mode=figures,round-precision=3,round-pad=false]{#1}}
\def\tr#1#2#3{\shortstack{\timeround{#1}\\[-1pt]{\scriptsize[\timeround{#2},\,\timeround{#3}]}}}
\def\sr#1#2#3#4{\shortstack{\(\timeround{#1}\times10^{#4}\)\\[-1pt]{\scriptsize\([\timeround{#2},\,\timeround{#3}]\times10^{#4}\)}}}
\begin{table}[htbp]
    \centering
    \small
    \setlength{\tabcolsep}{5pt}
    \caption{Computation time in seconds. Runtimes and local \(n_x\) are reported as median [minimum, maximum] over ten samples, except for the time-limited MILP entries described below.}
    \label{tab:optimization-runtime}
    \begin{tabular}{@{}rcccccc@{}}
        \toprule
        & & Endpoint scan & \multicolumn{2}{c}{\(M=2\)} &
        \multicolumn{2}{c}{\(M=3\)}\\
        \cmidrule(lr){4-5}\cmidrule(l){6-7}
        \(n\) & Local \(n_x\) & \(M=1\) & MILP & DP & MILP & DP\\
        \midrule
        250
        & \mr{43}{32}{49}
        & \sr{7.7}{6.4}{9.1}{-5}
        & \tr{0.168}{0.111}{0.208}
        & \tr{0.018}{0.015}{0.023}
        & \tr{0.593}{0.371}{1.102}
        & \tr{0.031}{0.016}{0.041}\\
        \addlinespace[4pt]
        500
        & \mr{91.5}{80}{102}
        & \sr{1.53}{1.30}{2.91}{-4}
        & \tr{1.505}{0.701}{2.073}
        & \tr{0.078}{0.057}{0.099}
        & \tr{5.625}{3.346}{10.861}
        & \tr{0.132}{0.101}{0.161}\\
        \addlinespace[4pt]
        \num{1000}
        & \mr{178}{165}{189}
        & \sr{2.84}{2.62}{3.86}{-4}
        & \tr{10.715}{7.508}{29.939}
        & \tr{0.293}{0.256}{0.600}
        & \tr{93.537}{43.885}{178.950}
        & \tr{0.497}{0.431}{0.681}\\
        \addlinespace[4pt]
        \num{2000}
        & \mr{349}{320}{388}
        & \sr{5.90}{5.32}{8.77}{-4}
        & \tr{161.878}{84.490}{512.975} & \tr{1.127}{1.035}{1.371}
        & \(>\num{1800}^{\dagger}\) & \tr{1.926}{1.512}{2.345}\\
        \addlinespace[4pt]
        \num{4000}
        & \mr{694.5}{651}{739}
        & \sr{1.258}{1.160}{1.341}{-3}
        & \(>\num{1800}^{\dagger}\) & \tr{4.367}{3.697}{4.896}
        & \(>\num{1800}^{\dagger}\) & \tr{7.375}{6.562}{8.260}\\
        \bottomrule
    \end{tabular}

    \smallskip
    \begin{minipage}{0.95\textwidth}
        \footnotesize Runtimes are rounded to at most three significant figures. Local \(n_x\) is the number of positive-weight training observations. \(\dagger\) Each attempted run reached the \num{1800}-second limit without certifying an optimum.
    \end{minipage}
\end{table}
\endgroup

The dynamic program produced feasible prediction sets with volumes close to the MILP optimum. All 100 dynamic-program runs satisfied the original sample constraint. We checked optimality using 100 completed MILP runs, comprising 30 checks for \(M=1\) and the 70 runs for \(M=2,3\) reported in the table. The solver certified all 100 solutions as optimal with zero reported gap. In the 70 comparisons for \(M=2,3\), the dynamic program's volume differed from the MILP optimum by less than or equal to \(10^{-7}\) in 57 cases. The median percentage increase in volume was zero for each \(M\), and the largest increases were 0.85\% for \(M=2\) and 0.55\% for \(M=3\).

The dynamic program also took less time in the typical paired comparison. Across the 40 comparisons for \(M=2\) and 30 for \(M=3\) with certified MILP optima, the median ratio of its runtime to the MILP runtime was 4.0\% and 2.2\%, respectively. All ten MILP runs at \(n=2000,M=2\) finished within the limit, while the other larger-sample configurations reached it in every preliminary run. These results apply to the simulation design and computing environment used here.

\section{Differential conformal procedure}\label{appendix: multiple scores}

To allow different endpoint adjustments, we conformalize a nested family of prediction sets fitted on the training sample, following \cite{chernozhukov2021DistributionalConformal} and \cite{gupta2022NestedConformal}.

Suppose we have estimated a measurable family $\{\hat C_q(x):q\in[0,1]\}$ with the training set $\mathcal I_1$, where $q$ denotes the nominal coverage level. The family satisfies $\hat C_q(x)\subseteq\hat C_{q'}(x)$ for $q\leq q'$. Its endpoints may change at different rates as $q$ increases. All fitting and adjustments to obtain nesting use only the training set.

For an observed interval $[y^L,y^U]$, define the conformity score
\[
    s(y^L,y^U,x)=\inf\{q:[y^L,y^U]\subseteq\hat C_q(x)\},
\]
assuming that finite infima are attained and taking $\inf\varnothing=+\infty$. This is the smallest level at which the prediction set contains the entire observed interval. Calculate $s_j=s(Y_j^L,Y_j^U,X_j)$ for $j\in\mathcal I_2$.

Let $n_2=\abs{\mathcal I_2}$ and $k=\lceil(n_2+1)(1-\alpha)\rceil$. Choose $\hat q$ as the $k$th order statistic of the conformity scores in the calibration sample augmented by one value $+\infty$, and set
\[
    \tilde C(x)=\hat C_{\hat q}(x),
    \qquad \hat C_{+\infty}(x)=\mathbb R.
\]
For $q=1-\beta$, this selects the largest $\beta$ for which at least $k$ calibration intervals are contained in $\hat C_{1-\beta}$, using $\mathbb R$ if no such $\beta$ exists.

Under exchangeability, the calibration and test scores are exchangeable conditional on the training set. Nesting and attainment give $s(y^L,y^U,x)\leq q$ if and only if $[y^L,y^U]\subseteq\hat C_q(x)$. The rank argument in Theorem~\ref{thm:marginal} therefore yields
\[
    \Prob\{[Y^L_{n+1},Y^U_{n+1}]
          \subseteq\tilde C(X_{n+1})\}\geq1-\alpha.
\]
Since $Y_{n+1}\in[Y^L_{n+1},Y^U_{n+1}]$, the same coverage bound holds for the latent outcome.

Table~\ref{tab:differential-calibration} illustrates a gain from different endpoint adjustments in a specific setting with known outcome bounds and fixed reporting thresholds. We compare a nested family that uses this reporting rule with expanding both endpoints by the same amount. The experiment uses 300 training observations, 200 calibration observations and 10,000 independent test observations, over 40 repetitions. There are no varying predictors. The target is $90\%$ coverage of the entire reported interval. Volume is total length after conformalization.

\begin{table}[htbp]
\centering\footnotesize
\caption{Differential conformalization with fixed reporting thresholds at nominal coverage $0.90$. Entries are means with Monte Carlo standard errors in parentheses. Coverage is reported for the observed bracket and the latent outcome. Volume is measured after conformalization.}
\label{tab:differential-calibration}
\begin{tabular}{lrrr}
\toprule
Method & Coverage & $Y$ coverage & Volume \\
\midrule
\multicolumn{4}{l}{\textit{Fixed reporting thresholds}} \\
Common dilation & 0.9999 (0.0000) & 1.0000 (0.0000) & 9.240 (0.019) \\
Family using the reporting rule & 0.9479 (0.0025) & 0.9479 (0.0025) & 5.642 (0.019) \\
\addlinespace
\bottomrule
\end{tabular}
\end{table}

The outcome is known to lie in $[-5,5]$. Generate $V\sim N(0.75,0.85^2)$ and set $Y=\max\{-5,\min(5,V)\}$. Values at or below $-2.5$ are reported as $[-5,-2.5]$, values at or above $1.4$ as $[1.4,5]$, and values between these thresholds are reported exactly. Thus the reporting rule is the same for every observation. Both methods start from the empirical $0.05$ and $0.95$ quantiles of the recorded values, using $-2.5$ and $1.4$ for the lower- and upper-tail reports.

The baseline moves both endpoints by the same amount, staying within $[-5,5]$. Following \citet{sesia2026ClippedOutcomes}, the alternative uses the known reporting rule. When the right endpoint reaches $1.4$, it is extended directly to $5$. When the left endpoint reaches $-2.5$, it is extended directly to $-5$. These intervals form a nested family and are conformalized using the same procedure as the baseline. A real-valued expansion parameter can index the family instead of $q$, and the coverage argument is unchanged. This avoids forcing a large adjustment at one end merely because the other end must cover a whole reported interval. For example, extending a right endpoint from $1.4$ to $5$ requires an equal-expansion baseline to move the left endpoint down by $3.6$ as well, even if that extra extension is unnecessary.

Mean length falls from $9.240$ to $5.642$, while coverage falls from $0.9999$ to $0.9479$, still above the $0.90$ target. {Using the known reporting rule substantially reduces unnecessary expansion in this example.} This example shows how the differential procedure can reduce the conservatism of equal endpoint expansion when the outcome bounds and reporting thresholds are known.

\section{Metrics for sets}\label{appendix: set metrics}
Two common choices of measurement for the difference between two intervals are the Hausdorff distance and the volume of symmetric difference. The \textit{Hausdorff distance} between two sets \(A, B \subseteq \mathbb{R}\) is defined as
\begin{equation*}
    d_{H}(A,B) = \max\cqty{\sup_{a \in A}\inf_{b \in B}\abs{a-b}, \sup_{b\in B}\inf_{a\in A}\abs{a-b}}.
\end{equation*}

The \textit{symmetric difference} \(\bigtriangleup\) between two sets \(A,B\) is defined as \(A\bigtriangleup B = (A\setminus B) \cup (B\setminus A)\), and the volume of the symmetric difference \(\mu(A\bigtriangleup B)\) has also been studied in the literature. On general measurable sets this is a pseudometric. Sets that differ only on a Lebesgue-null set have distance zero.

In general, these two metrics are not equivalent, but they are equivalent for two non-disjoint intervals on the real line. In particular, for a sequence of intervals \(A_{n} = [a_{n0}, a_{n1}]\) and \(B = [b_{0}, b_{1}]\), \(b_{0}< b_{1}\), it can be shown that \(d_{H}(A_{n}, B)\to 0\) if and only if \(\mu(A_{n} \bigtriangleup B) \to 0\) as \(n\to \infty\).

Therefore, in terms of consistency for prediction intervals, we can use either metric and we will proceed with the volume of symmetric difference and show that \(\hat{C}_{I}(x)\) is consistent in the sense that \(\sup_{x\in \mathcal X_n}\mu(\hat C_I(x)\bigtriangleup C_I^*(x)) = o_{p}(1)\).  If \(\inf_n\inf_{x\in\mathcal X_n}\mu(C_I^*(x))>0\), then uniform consistency holds under both notions of set difference.

However, when the prediction set is a union of intervals, these two metrics are not equivalent. In fact, consider \(A \subseteq \mathbb{R}\) and \(B = A \cup \cqty{b}\), for some \(b\in A^{c}\) such that \(d(b,A) > 0\), then the Hausdorff distance between \(A\) and \(B\) is \(d_{H}(A,B) = d(b,A)>0\), while the volume of the symmetric difference is \(0\). Thus, the volume of symmetric difference has the advantage of being less sensitive to a difference of negligible sets and seems more appropriate for practical prediction set estimation.

\section{Estimation and conformal inference based on quantile regression}\label{sec:quantile}

As we have suggested, there is a simpler way to construct valid prediction sets with quantile regression, although the resulting prediction set is not generally optimal in terms of length.

Let $q_\beta^a(x)=\inf\{t:P(Y^a\leq t\mid X=x)\geq\beta\}$ denote the lower conditional quantile for $a\in\{L,U\}$ and $\beta\in(0,1)$. By Bonferroni's inequality, for two events $A,B$, \(P(A\cap B) \geq \max\pqty{0, P(A) + P(B) - 1}\), and given $P(Y^{L}\geq q^{L}_{\frac{\alpha}{2}}(x)\mid X=x) \geq 1 - \alpha/2$ and $P(Y^{U}\leq q^{U}_{1 - \frac{\alpha}{2}}(x)\mid X = x) \geq 1- \alpha/2$ by the definition of quantiles, we have
\begin{equation*}
    P \pqty{Y^{L}\geq q^{L}_{\frac{\alpha}{2}}(x), Y^{U}\leq q^{U}_{1 - \frac{\alpha}{2}}(x)\mid X = x} \geq 1 - \alpha
\end{equation*}
If the endpoint quantile estimators are consistent, then $\hat C_q(x)=[\hat q^L_{\alpha/2}(x),\hat q^U_{1-\alpha/2}(x)]$ consistently estimates the valid population prediction interval \(C_q(x)= [q^{L}_{\frac{\alpha}{2}}(x),q^{U}_{1 - \frac{\alpha}{2}}(x)]\) which is not necessarily an optimal prediction set.
There is a vast body of research on quantile regression, see \cite{koenker2017QuantileRegression}, \cite{athey2019GeneralizedRandom} for random forest based quantile regression and \cite{li2021NonparametricQuantile} for nonparametric quantile regression with mixed continuous and discrete data.

Once we have \(\hat{C}_{q}(x)\), it is straightforward to use the conformalization procedure to obtain finite-sample valid prediction sets. Let \(s_{j}^{q} = s(Y_{j}^{L}, Y_{j}^{U}, X_{j}; \hat{q}_{\frac{\alpha}{2}}^{L}, \hat{q}_{1- \frac{\alpha}{2}}^{U})\). Set $\vartheta^q_{1-\alpha}=q(\{s_j^q:j\in\mathcal I_2\};(1-\alpha)(1+1/n_2))$, with large enough $n_2$ as in Section~\ref{sec:conformal}. The conformalized quantile prediction interval is
\begin{equation*}
    \tilde{C}_{q}(x) = \bqty{\hat{q}^{L}_{\frac{\alpha}{2}}(x) - \vartheta^{q}_{1-\alpha}, \hat{q}^{U}_{1 - \frac{\alpha}{2}}(x) + \vartheta^{q}_{1-\alpha}}.
\end{equation*}
\section{Diagnostics and Robustness Checks for Empirical Studies}\label{appendix:empirical-checks}\label{appendix:empirical-m2}\label{appendix:cps-modality}

In this section, we conduct robustness checks for the sensitivity to the choice of $M$ in the empirical studies. We report the volumes of the estimator $\hat C$ and the conformalized prediction set $\tilde C_{\mathrm{loc}}$ separately. We find that, at the 90\% coverage target, allowing $M=2$ intervals does not reduce mean conformalized length on the log scale in either application.

We compare $1-\alpha\in\{90\%,75\%,50\%,10\%\}$ using 100 paired random splits, assigning 60\% to training, 20\% to conformalization and 20\% to evaluation. CPS fitting uses a product Epanechnikov kernel with bandwidths $(0.34287434,2.00268184)$ for education and experience, respectively. CPS conformalization uses the fixed education levels 12, 14, 16 and 18 and experience windows $[e-5,e+5)$ for $e\in\{10,20,30,40\}$, as displayed in Table~\ref{tab:empirical-m2-summary}, covering 31,814 observations. The cleaned UK sample contains 6,859 adverts with both coordinates, lower salary bounds at most \pounds300,000 and upper bounds at most \pounds600,000. Zero lower bounds are replaced by \pounds1 before taking logarithms.

For UK local conformalization, we use training locations alone to form eight rectangular regions by recursive median splits, with at least 300 training observations per region, and fix the partition before conformalization. The product Epanechnikov bandwidth expands from $0.175$ degrees in each coordinate to include at least 100 training observations. Both choices of $M$ use the same splits, regions and fitting weights. Conformalization follows Equation~\eqref{eq:local conformal pred set}, without post-conformalization smoothing. Conditional on training, the local coverage argument in Theorem~\ref{thm:local validity} applies within these regions. For the city diagnostics, reporting groups extend $0.175$ degrees in each coordinate from the six city centres.

Panel A of Table~\ref{tab:empirical-m2-summary} reports mean volumes before and after conformalization. Within each split, it pools the displayed CPS groups using their evaluation counts and averages over the full UK evaluation sample, then averages over splits. At $1-\alpha=90\%$, increasing $M$ from 1 to 2 reduces the mean volume of $\hat C$ before conformalization from 3.057 to 2.909 for CPS and from 2.041 to 1.908 in the UK job advert application. However, the mean volume of $\tilde C_{\mathrm{loc}}$ increases from 3.279 to 3.343 (1.9\%) for CPS and from 2.252 to 2.266 (0.6\%) for UK adverts. Coverage changes from 90.92\% to 90.52\% for CPS and from 90.46\% to 90.43\% for UK adverts. Hence, the smaller estimator does not yield a smaller conformalized prediction set on average at this target. For CPS, changes in conformalized volume at the 75\% and 50\% targets are small, whereas $M=2$ reduces it by 7.2\% at the 10\% target. In the UK job advert dataset, mean conformalized volume increases by 3.9\%, 5.1\% and 30.7\% at $1-\alpha=75\%,50\%$ and $10\%$, respectively. Achieved coverage also differs across methods.

Panel B of Table~\ref{tab:empirical-m2-summary} shows the 90\% results by group. Mean log length is larger with $M=2$ in 14 of the 16 CPS groups and all six UK city reporting groups. The ordering can change when length is measured in dollars or pounds. Mean length then falls by 2.4\% for CPS and 1.0\% for UK adverts. All predictions are retained in the summaries, and none is unbounded or empty for the eight-region procedure at any of the four targets.

\begin{table}[!htbp]
\centering
\small
\setlength{\tabcolsep}{4pt}
\renewcommand{\arraystretch}{0.98}
\caption{Sensitivity to $M$ in the two empirical applications}
\label{tab:empirical-m2-summary}
\label{tab:empirical-m2-local}
\begin{tabular}{lrrrrrrr}
\toprule
$1-\alpha$ & \multicolumn{2}{c}{Coverage (\%)} & \multicolumn{2}{c}{Volume before} & \multicolumn{2}{c}{Volume after} & $\Delta$ volume \\
 & $M=1$ & $M=2$ & $M=1$ & $M=2$ & $M=1$ & $M=2$ & \\
\midrule
\multicolumn{8}{l}{\textit{A. Overall results}} \\
\multicolumn{8}{l}{\textit{US CPS}} \\
90\% & 90.92 & 90.52 & 3.057 & 2.909 & 3.279 & 3.343 & +0.064 (0.005)\\
75\% & 76.00 & 75.65 & 1.766 & 1.679 & 1.867 & 1.870 & +0.003 (0.002)\\
50\% & 50.92 & 50.54 & 0.880 & 0.798 & 0.941 & 0.930 & -0.011 (0.001)\\
10\% & 10.38 & 10.31 & 0.125 & 0.051 & 0.188 & 0.175 & -0.014 (0.001)\\
\addlinespace
\multicolumn{8}{l}{\textit{UK IT job advert dataset, full evaluation sample}} \\
90\% & 90.46 & 90.43 & 2.041 & 1.908 & 2.252 & 2.266 & +0.014 (0.005)\\
75\% & 75.90 & 75.86 & 1.271 & 1.212 & 1.376 & 1.429 & +0.053 (0.003)\\
50\% & 51.88 & 51.19 & 0.747 & 0.679 & 0.848 & 0.891 & +0.043 (0.004)\\
10\% & 11.25 & 11.09 & 0.142 & 0.069 & 0.235 & 0.307 & +0.072 (0.004)\\
\addlinespace
\midrule
\multicolumn{8}{l}{\textit{B. Results by group at $1-\alpha=90\%$}} \\
Group & \multicolumn{2}{c}{Coverage (\%)} & \multicolumn{2}{c}{Volume before} & \multicolumn{2}{c}{Volume after} & $\Delta$ volume \\
 & $M=1$ & $M=2$ & $M=1$ & $M=2$ & $M=1$ & $M=2$ & \\
\midrule
\multicolumn{8}{l}{\textit{CPS education / experience (years)}} \\
12 / 10 & 90.46 & 90.48 & 3.321 & 3.060 & 3.971 & 4.076 & +0.104\\
12 / 20 & 90.64 & 90.48 & 3.406 & 3.228 & 3.772 & 3.889 & +0.117\\
12 / 30 & 91.20 & 90.79 & 3.237 & 3.139 & 3.458 & 3.586 & +0.128\\
12 / 40 & 91.25 & 90.75 & 3.299 & 3.229 & 3.428 & 3.633 & +0.205\\
14 / 10 & 90.22 & 90.23 & 3.032 & 2.806 & 3.352 & 3.349 & -0.003\\
14 / 20 & 91.12 & 90.58 & 3.204 & 3.082 & 3.364 & 3.403 & +0.040\\
14 / 30 & 91.85 & 90.56 & 3.227 & 3.103 & 3.294 & 3.435 & +0.141\\
14 / 40 & 90.83 & 90.27 & 3.375 & 3.294 & 3.495 & 3.533 & +0.038\\
16 / 10 & 90.69 & 90.38 & 2.858 & 2.723 & 3.022 & 3.037 & +0.016\\
16 / 20 & 90.89 & 90.51 & 2.923 & 2.761 & 3.088 & 3.081 & -0.007\\
16 / 30 & 90.61 & 90.22 & 2.954 & 2.799 & 3.136 & 3.183 & +0.047\\
16 / 40 & 91.36 & 90.89 & 3.280 & 3.116 & 3.440 & 3.501 & +0.061\\
18 / 10 & 91.24 & 90.92 & 2.556 & 2.437 & 2.698 & 2.758 & +0.060\\
18 / 20 & 90.80 & 90.54 & 2.371 & 2.275 & 2.510 & 2.548 & +0.039\\
18 / 30 & 91.24 & 90.91 & 2.574 & 2.429 & 2.726 & 2.757 & +0.031\\
18 / 40 & 90.83 & 90.09 & 3.057 & 2.833 & 3.479 & 3.556 & +0.077\\
\addlinespace
\multicolumn{8}{l}{\textit{UK IT job advert dataset, city reporting groups}} \\
Birmingham & 88.92 & 89.62 & 1.870 & 1.772 & 2.045 & 2.119 & +0.074\\
Bristol & 90.76 & 91.37 & 1.838 & 1.749 & 2.051 & 2.135 & +0.084\\
Edinburgh & 92.70 & 92.72 & 2.143 & 1.997 & 2.425 & 2.527 & +0.102\\
Glasgow & 93.07 & 92.52 & 2.234 & 2.017 & 2.516 & 2.548 & +0.032\\
London & 90.94 & 91.04 & 1.945 & 1.929 & 2.083 & 2.146 & +0.063\\
Manchester & 90.33 & 90.51 & 2.258 & 2.111 & 2.650 & 2.669 & +0.019\\
\addlinespace
\bottomrule
\end{tabular}
\par\smallskip
\begin{minipage}{0.98\textwidth}\footnotesize
\textit{Notes.} Coverage requires $[Y_i^L,Y_i^U]\subseteq\tilde C_{\mathrm{loc}}(X_i)$. Volumes are measured on the log scale before and after conformalization. CPS conformalization uses the displayed fixed groups, and UK conformalization uses geographic cells fixed from training locations. UK overall results include all held-out adverts, with city windows used for reporting. $\Delta$ is conformalized volume for $M=2$ minus $M=1$. Parentheses give the standard error of its mean over 100 random splits of the same dataset.
\end{minipage}
\end{table}

Figure~\ref{fig:cps-income-diagnostics} provides additional information for the CPS results. In panel (a), density estimates from exact reports have one broad peak at bandwidths of \$10,000 and \$20,000. Finer smoothing reveals smaller peaks, with 53\% of exact incomes reported in multiples of \$5,000. Panel (b) uses all 39,333 reports, including 8,819 income ranges. The black and dashed lines show the minimum and maximum possible percentages in each income range, and the shading shows the remaining reporting uncertainty. These plots supplement the direct assessment of $M$ above.

\clearpage
\begin{figure}[p]
    \centering
    \captionsetup{font=small}
    \begin{subfigure}{0.96\textwidth}
        \centering
        \includegraphics[width=\linewidth]{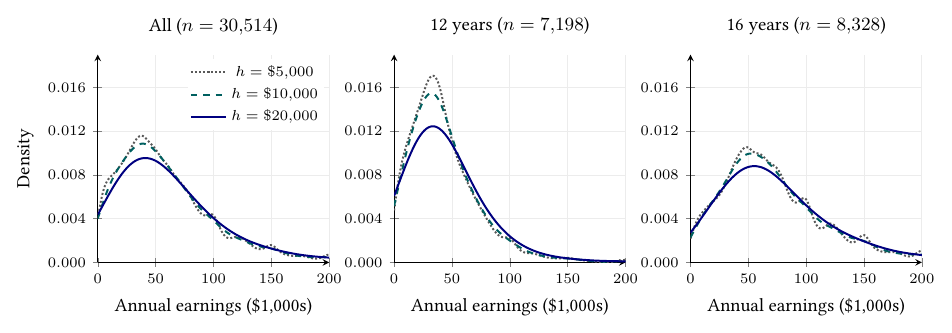}
        \caption{Density of exact reported incomes}
        \label{fig:cps-exact-modality}
    \end{subfigure}

    \medskip
    \begin{subfigure}{0.86\textwidth}
        \centering
        \includegraphics[width=\linewidth]{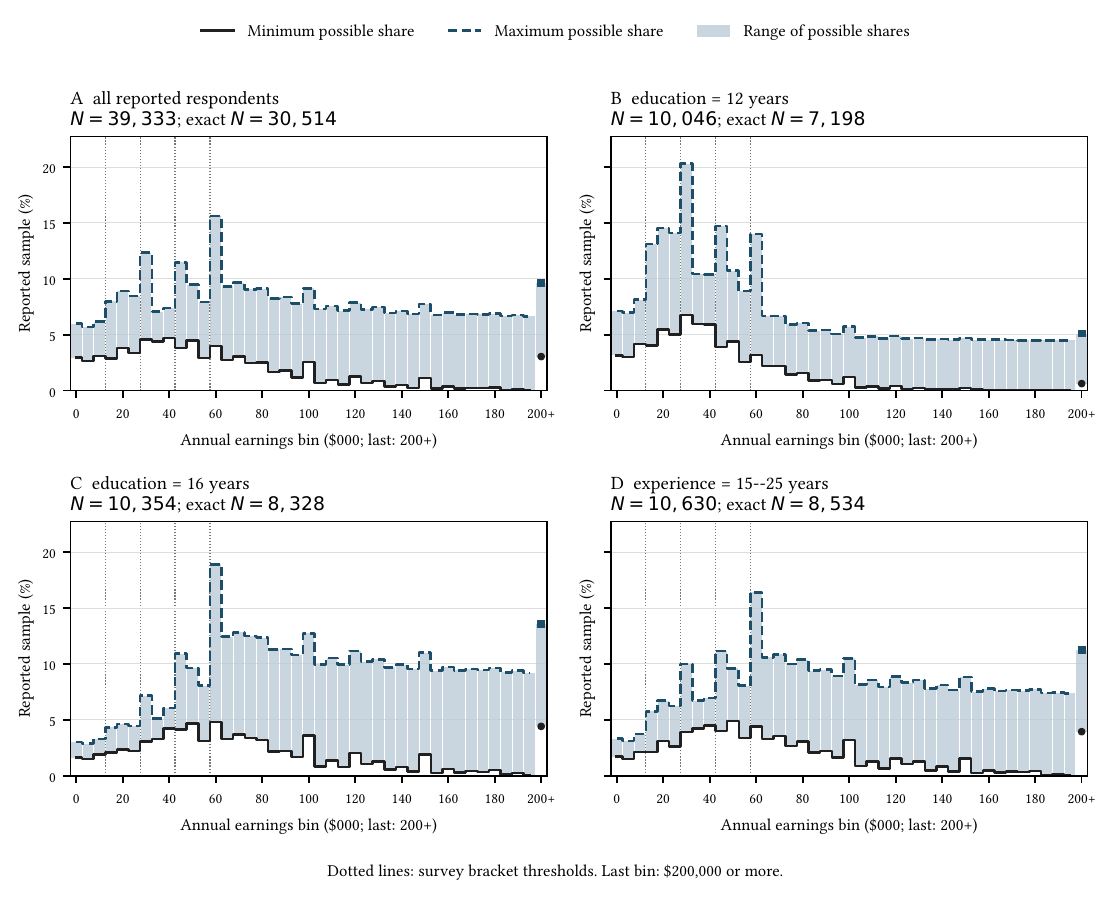}
        \caption{Possible income shares using all reports}
        \label{fig:cps-modality}
    \end{subfigure}
    \caption{Supplementary CPS income-distribution diagnostics. Panel (a) shows Gaussian density estimates for all exact reporters and those with 12 or 16 years of education, using bandwidths $h$. All exact incomes contribute, including those above the displayed \$200,000 limit. Panel (b) shows the minimum possible share in each income range in black and the maximum possible share as a dashed line. Shading marks the range between them. Income ranges have width \$5,000, with a final range for \$200,000 or more. Dotted vertical lines mark survey reporting thresholds.}
    \label{fig:cps-income-diagnostics}
\end{figure}

\vspace{2em}

\end{document}